%% file: 02_dmpc_for_cooperation_journal_extended_version.tex
\documentclass[journal, twoside]{IEEEtran}  % for arxiv
\usepackage{cite}
\usepackage{amsmath,amssymb,amsfonts}
\usepackage{mathtools}
\usepackage{algorithmic}
\usepackage{graphicx}
\usepackage{textcomp}
\def\BibTeX{{\rm B\kern-.05em{\sc i\kern-.025em b}\kern-.08em
    T\kern-.1667em\lower.7ex\hbox{E}\kern-.125emX}}

\usepackage[utf8]{inputenc}
\usepackage{tikz}
\usepackage{pgfplots}
\DeclareUnicodeCharacter{2212}{−}
\usepgfplotslibrary{groupplots,dateplot}
\usetikzlibrary{patterns,shapes.arrows}
\pgfplotsset{compat=newest}
\newlength\axisheight
\newlength\axiswidth

\definecolor{morange}{RGB}{255,104,25}
\definecolor{mpurple}{RGB}{245,0,228}
\definecolor{mblue}{RGB}{19,15,255}

\newtheorem{theorem}{Theorem}
\newtheorem{assumption}{Assumption}
\newtheorem{definition}{Definition}
\newtheorem{remark}{Remark}
\newtheorem{lemma}{Lemma}
\newtheorem{algorithm}{Algorithm}
\newtheorem{proof}{Proof}  % needed for IEEEtran

\newcommand{\mr}[1]{\mathrm{#1}}
\newcommand{\ci}{{\mr{c}, i}}
\DeclareMathOperator*{\argmin}{argmin}

\usepackage{fancyhdr}
\fancypagestyle{firststyle}
{
  \fancyhf{}
  \fancyfoot[C]{\scriptsize
  \copyright 2024 IEEE. Personal use of this material is permitted. Permission 
  from IEEE must be obtained for all other uses, in any current or future 
  media, including reprinting/republishing this material for advertising or 
  promotional purposes, creating new collective works, for resale or 
  redistribution to servers or lists, or reuse of any copyrighted 
  component of this work in other works. 
  This is an extended version of \cite{Koehler2024extended}; all extensions are given in the appendix.}
  \setlength{\footskip}{10pt}
}

\begin{document}
\title{Distributed MPC for Self-Organized Cooperation of Multiagent Systems - Extended Version}
\author{Matthias K\"ohler, Matthias A. M\"uller, and Frank Allg\"ower
\thanks{F. Allg\"ower and M. A. M\"uller are thankful that this work was funded by the Deutsche Forschungsgemeinschaft (DFG, German Research Foundation) -- AL 316/11-2 - 244600449.
F. Allgöwer is thankful that this work was funded by the Deutsche Forschungsgemeinschaft (DFG, German Research Foundation) under Germany’s Excellence Strategy -- EXC 2075 -- 390740016.}
\thanks{M. K\"ohler and F. Allg\"ower are with the University of Stuttgart, Institute for Systems Theory and Automatic Control, 70550 Stuttgart, Germany (e-mails: koehler@ist.uni-stuttgart.de, allgower@ist.uni-stuttgart.de).}
\thanks{M. A. M\"uller is with the Leibniz University Hannover, Institute of Automatic Control, 30167 Hanover, Germany (e-mail: mueller@irt.uni-hannover.de).}
}

\maketitle
\thispagestyle{firststyle}

\begin{abstract}
We present a sequential distributed model predictive control (MPC) scheme for cooperative control of multiagent systems with dynamically decoupled heterogeneous nonlinear agents subject to individual constraints.
In the scheme, we explore the idea of using tracking MPC with artificial references to let agents coordinate their cooperation without external guidance.
Each agent combines a tracking MPC with artificial references, the latter penalized by a suitable coupling cost.
They solve an individual optimization problem for this artificial reference and an input that tracks it, only communicating the former to its neighbors in a communication graph.
This puts the cooperative problem on a different layer than the handling of the dynamics and constraints, loosening the connection between the two.
We provide sufficient conditions on the formulation of the cooperative problem and the coupling cost for the closed-loop system to asymptotically achieve it.
Since the dynamics and the cooperative problem are only loosely connected, classical results from distributed optimization can be used to this end.
We illustrate the scheme's application to consensus and formation control.
\end{abstract}

% \begin{IEEEkeywords}
%     Predictive control, distributed control, multiagent systems
% \end{IEEEkeywords}

%===============================================================================
\section{Introduction}
% Multi-agent systems and cooperative tasks:
Distributed control of large-scale systems receives unabated attention in the scientific community due to further developments in communication technology and complexity of systems and objectives.
In particular, control of multiagent systems comprising independent agents that want to achieve a cooperative goal has various applications, including, e.g., coordination of vehicles~\cite{Smith2005,Dimarogonas2009,Ritz2013,Zheng2017,Fonseca2019,Klausen2020},
\cite{Ryan2004, Cao2013, Yu2021} (see the references therein) for diverse purposes, sensor networks~\cite{OlfatiSaber2008}, and air traffic management~\cite{Tomlin1998}.

% Distributed MPC:
Since most of these systems are nonlinear and subject to constraints, distributed model predictive control (MPC) is a capable choice as the control strategy.
Distributed MPC has been successfully applied to various systems architectures using a wide range of solution strategies, see, e.g.,~\cite{Scattolini2009, Mueller2017b}.

% DMPC for cooperation:
Many cooperative tasks can be formulated based on consensus, i.e., agents agreeing on a common value in some output~\cite{OlfatiSaber2007}.
For multiagent systems subject to input constraints that should reach a consensus, distributed MPC schemes have been developed for double-integrator dynamics~\cite{Cheng2015}
and for homogeneous linear dynamics~\cite{Li2015}.
In~\cite{Johansson2006}, consensus is achieved by tracking an optimal consensus point that has been previously computed by solving a centralized optimization problem.
A non-iterative distributed MPC scheme with local optimization problems is proposed in~\cite{Zheng2017} for heterogeneous vehicle platoons.
In the local open-loop optimal control problems, the distance to a communicated trajectory is penalized.
In addition, an average of the assumed neighbors' outputs is used as a terminal equality constraint, requiring that in the first few time steps the optimization problems are feasible, but allowing that the desired set point of each agent need not be known \emph{a priori}.
Self-organized output consensus through distributed MPC for heterogeneous linear agents subject to input and state constraints is achieved in~\cite{Hirche2020}.
The proposed scheme combines a tracking MPC formulation with a suitable cost such that the agents continuously move a tracked artificial reference towards consensus, eventually leading to convergence of the agents' outputs.
A more general framework for cooperative tasks for multiagent systems with nonlinear dynamics and coupling in constraints is presented in~\cite{Muller.2012}, where a sequential distributed MPC scheme is designed in order to stabilize a set instead of only a specific equilibrium.

% Tracking MPC using artificial references:
A tracking MPC scheme using artificial references has been introduced in~\cite{Limon.2008, Limon2018}, where a steady state is concurrently computed and tracked by the system through the same optimization problem, with the goal to track an externally supplied steady state.
The main advantages of this scheme over a standard MPC formulation are a larger region of attraction and the independence of recursive feasibility, i.e., the continuous operation of the scheme, of the external reference.

% Extension to distributed tracking MPC:
Moreover, an extension to the distributed linear setting is presented in~\cite{Ferramosca2011}.
In~\cite{Conte2016}, also in the linear setting, an optimization problem with artificial references that is suitable for distributed optimization is formulated and the distributed synthesis of the used ingredients is presented.
Both approaches could be used to track an external reference that achieves a cooperative goal.
This idea is used in~\cite{Carron2020a} for nonlinear multiagent systems to solve the coverage problem.
In~\cite{Santana2022}, a similar iterative distributed MPC scheme as in~\cite{Ferramosca2011} is proposed for economic operation of input-coupled linear systems, also using artificial references to increase the domain of attraction, and a Taylor series approximation of the economic objective.

% Contribution:
In this paper, we propose a sequential distributed MPC scheme for nonlinear multiagent systems.
Similarly to the works cited above, each agent uses a local tracking MPC formulation with an artificial reference.
However, no external references are supplied, and the eventual solution of the cooperative goal is not known \emph{a priori}.
Instead, the reference is implicitly provided by a suitably designed cost that encodes the cooperative goal and connects neighboring agents' references. 
Hence, tracking and cooperative coordination are combined in one scheme.
This enables the multiagent system to achieve the cooperative goal through self-organization without external inputs given, e.g., by a global reference governor. 
Nonetheless, due to the use of artificial references, the connection between the cooperative problem and the agents' dynamics is sufficiently loose to allow for an almost decoupled analysis.
We present sufficient conditions on this cost for cooperation and the associated cooperation set, to which the outputs of the agents should converge.
We show that a related set on state level is asymptotically stable for the closed loop, resulting in asymptotic achievement of the encoded cooperative goal.
Due to the use of artificial references, a large region of attraction is expected, while they are the only communicated data.

In~\cite{Hirche2020}, we used a similar idea to achieve self-organized consensus in the outputs of heterogeneous linear agents.
In contrast, in this paper, we consider heterogeneous nonlinear agents and provide a framework for general cooperative tasks, which includes consensus as a special case as illustrated in Section~\ref{ssec:consensus}, and allows agents more flexibility.

The following parts are structured as follows.
The setup and control goal is introduced in Section~\ref{sec:setup} and the distributed MPC scheme for self-organized cooperation is presented in Section~\ref{sec:scheme}.
In Section~\ref{sec:analysis}, we prove the asymptotic achievement of the control goal of the closed-loop system.
We apply the scheme to two examples, consensus and formation control, in Section~\ref{sec:example}, before providing concluding remarks in~\ref{sec:conclusion}.

\subsection{Notation}
The interior of a set $\mathcal{S}$ is denoted by $\mr{int}\,\mathcal{S}$.
The non-negative reals are denoted by $\mathbb{R}_{\ge 0}$. $\mathbb{N}_0$ denotes the natural numbers including 0.
The set of integers from $a$ to $b$, $a \le b$, is denoted by $\mathbb{I}_{a:b}$.
The identity matrix of dimension $n$ is $I_{n}$.
The maximum and minimum eigenvalues of a matrix $A=A^\top$ are denoted by $\lambda_{\max}(A)$ and $\lambda_{\min}(A)$.
Given a positive (semi-)definite matrix $A=A^\top$, the corresponding (semi-)norm is written as $\Vert x \Vert_A = \sqrt{x^\top A x}$.
Given a collection of $m$ vectors $v_i \in \mathbb{R}^{n_i}$, $i \in \mathbb{I}_{1:m}$, we denote the stacked vector by $v = \mathrm{col}_{i=1}^m(v_i) = [v_1^\top \dots v_m^\top]^\top$.
The Euclidean norm of $x \in \mathbb{R}^n$ is denoted by $\Vert x \Vert$ and the maximum norm by $\Vert x \Vert_\infty$.
The Euclidean distance of $x \in \mathbb{R}^n$ to a closed set $\mathcal{X} \subseteq \mathbb{R}^n$ is denoted by $\vert x \vert_\mathcal{X}$, i.e., $\vert x \vert_\mathcal{X} = \min_{\tilde{x} \in \mathcal{X}} \Vert x - \tilde{x} \Vert$.
$P_{\mathcal{X}}\left[\cdot\right]$ denotes the projection onto the set $\mathcal{X}$.
The block-diagonal matrix with $m$ block-diagonal elements $A_i$ is denoted by $\mathrm{diag}_{i=1}^m (A_i)$.
The convex hull of a set $\mathcal{A}$ is denoted by $\mathrm{co}(\mathcal{A})$.

%===============================================================================
\section{Nonlinear multiagent system}\label{sec:setup}
We consider a multiagent system with $m \in \mathbb{N}$ agents where the agents have heterogeneous, nonlinear discrete-time dynamics of the following form
\begin{subequations}
    \begin{align}
        x_i(t + 1) &= f_i(x_i(t), u_i(t))\label{eq:agent_dynamics}, \\
        y_i(t) &= h_i(x_i(t), u_i(t))\label{eq:output_relation}
    \end{align}
\end{subequations}
with state $x_i(t) \in \mathbb{X}_i \subseteq \mathbb{R}^{n_i}$, input $u_i(t) \in \mathbb{U}_i \subseteq \mathbb{R}^{q_i}$, and output $y_i(t) \in \mathbb{Y}_i \subseteq \mathbb{R}^{p}$ at time $t\in\mathbb{N}_0$, continuous $f_i: \mathbb{X}_i \times \mathbb{U}_i \to \mathbb{X}_i$, and continuous $h_i: \mathbb{X}_i \times \mathbb{U}_i \to \mathbb{Y}_i$.
We assume that the state and input should satisfy pointwise-in-time constraints $(x_i(t), u_i(t)) \in \mathcal{Z}_i \subset \mathbb{X}_i \times \mathbb{U}_i$ for all $t \in \mathbb{N}_0$, where $\mathcal{Z}_i$ is compact.
Furthermore, throughout the paper, we assume that state measurements are available.

The communication topology of the multiagent system is modeled by a graph $\mathcal{G} = (\mathcal{V}, \mathcal{E})$ with vertices $\mathcal{V}$ and edges $\mathcal{E}$.
We assign to each agent a vertex and connect agent $i$ and agent $j$ with a directed edge $e_{ij}$, pointing at vertex $j$, if the latter can receive information from the former.
The set of neighbors of agent $i$ then contains all agents from which it may receive information, i.e., $\mathcal{N}_i = \{j \in\mathcal{V} \mid e_{ji} \in \mathcal{E}\}$.
We assume that communication is bilateral, i.e., if $(i,j)\in\mathcal{E}$, so is $(j,i)\in\mathcal{E}$, and that the graph is connected.

The goal is to achieve cooperation between the agents, more precisely, their outputs should converge to some set that characterizes the cooperation goal.
For example, if this goal would be (output) consensus, then the outputs should converge to the set of equal output values.
In addition, we require that the constraints are satisfied at all times.
Finally, the agents should only rely on information provided by neighbors and decide themselves on how to attain the cooperation goal.
In order to achieve this goal, we propose a sequential distributed MPC scheme in the next section.

%===============================================================================
\section{Distributed MPC for cooperation}\label{sec:scheme}
The sequential distributed MPC scheme that we develop in this section comprises essentially two parts.
In a tracking part, the nonlinear dynamics and constraints are handled by using MPC for tracking~\cite{Limon.2008,Koehler2020b}.
Moreover, a suitably designed cost penalizes the deviation of artificial outputs, called \emph{cooperation outputs}, from the cooperative goal and plays a comparable role as the offset cost in~\cite{Limon.2008}.
In a sequential scheme, each agent computes a cooperation output and an input sequence to track it, by solving an optimization problem.

We will restrict the possible cooperation outputs to those which have a corresponding (cooperation) equilibrium of~\eqref{eq:agent_dynamics} strictly inside the constraint set.
\begin{assumption}\label{asm:reference}
    The cooperation output satisfies
    $
    y_{\mr{c},i} \in \mathcal{Y}_i \subseteq
    \{
    y_i
    \mid
    \exists r_i = (x_i, u_i) \in \mathcal{Z}_i^\mr{c},
    y_i = h_i(x_i, u_i)
    \}%
    $
    where
    $
    \mathcal{Z}_i^\mr{c} \subseteq
    \{
    r_i
    \mid
    r_i = (x_i, u_i) \in \mr{int} \, \mathcal{Z}_i, \,
    x_i = f_i(x_i, u_i)
    \}
    ,
    $
    $\mathcal{Y}_i$ is compact and convex, and $\mathcal{Z}_i^\mr{c}$ is compact.
\end{assumption}
A cooperation output $y_{\mr{c}, i}$ and corresponding cooperation equilibrium $r_{\mr{c},i} = (x_{\mr{c},i}, u_{\mr{c},i})$ that satisfy Assumption~\ref{asm:reference} are called \emph{admissible} in the following.

For simplicity, we use the following quadratic stage cost:
\begin{equation*}
        \ell_i(x_i, u_i, r_{\mr{c},i}) = \Vert x_i - x_{\mr{c},i} \Vert^2_{Q_i} + \Vert u_i - u_{\mr{c},i} \Vert^2_{R_i}
\end{equation*}
where $\ell_i: \mathbb{X}_i \times \mathbb{U}_i \times \mathcal{Z}_i^\mr{c} \to \mathbb{R}_{\ge 0}$, and $Q_i \in \mathbb{R}^{n_i \times n_i}$, $R_i \in \mathbb{R}^{q_i \times q_i}$ are positive-definite weighting matrices.
With this, we define the tracking cost
$
    J_i^\mr{tr}(x_i(\cdot \vert t), u_i(\cdot \vert t), r_{\mr{c},i}) = \sum_{k=0}^{N-1} \ell_i(x_i(k \vert t), u_i(k \vert t), r_{\mr{c},i})
    + V_i^\mr{f}(x_i(N \vert t), r_{\mr{c},i})
$
where $N \in \mathbb{N}_0$ is the prediction horizon, $x_i(k \vert t)$, $u_i(k \vert t)$ denote the $k$-step ahead predicted state and input trajectory from time $t$, $r_{\mr{c},i}$ is the cooperation equilibrium, and $V_i^\mr{f}(x_i(N \vert t), r_{\mr{c},i})$ is the momentarily specified terminal cost.

Before we define the necessary terminal ingredients, we introduce the following auxiliary optimization problem which simplifies the analysis of our scheme.
\begin{subequations}\label{eq:MPC_for_tracking}
    \begin{align}
        W_i(x_i(t), r_{\mr{c},i}) &= \min_{u_i(\cdot \vert t)} J_i^\mr{tr}(x_i(\cdot \vert t), u_i(\cdot \vert t), r_{\mr{c},i})
        \label{eq:MPC_for_tracking_cost}
        \\
        \intertext{subject to}
        x_i(0 \vert t) &= x_i(t), \label{eq:MPC_for_tracking_IC}
        \\
        x_i(k+1 \vert t) &= f_i(x_i(k \vert t), u_i(k \vert t)), \; k \in \mathbb{I}_{0:N-1}, 
        \label{eq:MPC_for_tracking_dynamics}
        \\
        (x_i(k \vert t), u_i(k \vert t)) &\in \mathcal{Z}_i, k \in \mathbb{I}_{0:N-1}, 
        \label{eq:MPC_for_tracking_constraints}
        \\
        x_i(N \vert t) &\in \mathcal{X}_i^\mr{f}(r_{\mr{c},i})
    \end{align}
\end{subequations}
with the momentarily specified terminal set $\mathcal{X}_i^\mr{f}(\cdot)$.

In order to achieve the tracking goal while satisfying the constraints, we need suitable terminal ingredients that satisfy the following assumption for the tracking MPC scheme.
These are standard assumptions in MPC for tracking (cf.~\cite[Assumption 3]{Limon2018},~\cite[Assumption 2]{Koehler2020b}).
\begin{assumption}\label{asm:stabilising_terminal_ingredients}
There exist a control law
        $
        k_{i}^\mathrm{f}: \mathbb{X}_i \times \mathcal{Z}_i^\mr{c} \to \mathbb{U}_i
        $, a continuous terminal cost
        $
        V_i^\mr{f}: \mathbb{X}_i \times \mathcal{Z}_i^\mr{c} \to \mathbb{R}_{\ge 0}
        $ and a compact terminal set
        $
        \mathcal{X}_i^\mr{f}(r_{\mr{c},i}) \subseteq \mathbb{X}_i
        $ such that for any $r_{\mr{c},i} \in \mathcal{Z}_i^\mr{c}$ and any $x_i \in \mathcal{X}_i^\mr{f}(r_{\mr{c},i})$
        \begin{subequations}
            \begin{align}
                V_i^\mr{f}(x_i^{+}, r_{\mr{c},i}) - V_i^\mr{f}(x_i, r_{\mr{c},i}) 
                &\le  - \ell_i(x_i, k^\mr{f}_i(x_i, r_{\mr{c},i}), r_{\mr{c},i}), \label{eq:terminal_cost_decrease}
                \\
                (x_i, k^\mr{f}_i(x_i, r_{\mr{c},i})) &\in \mathcal{Z}_i,
                \\
                x_i^{+} &\in \mathcal{X}_i^\mr{f}(r_{\mr{c},i})
            \end{align}
        \end{subequations}
        where $x_i^{+} = f_i(x_i, k^\mr{f}_i(x_i, r_{\mr{c},i}))$.

        Moreover, there exist constants $c_{\mr{u},i}, \epsilon_i > 0$ for any admissible (satisfying Assumption~\ref{asm:reference}) cooperation equilibrium $r_{\mr{c},i}$ such that for any $x_i \in \mathbb{X}_i$ with $\Vert x_i - x_{\mr{c},i} \Vert_{Q_i} \le \epsilon_i$, the MPC for tracking problem~\eqref{eq:MPC_for_tracking} is feasible and
        \begin{equation}\label{eq:tracking_value_function_upper_bound}
            W_i(x_i, r_{\mr{c},i}) \le c_{\mr{u},i} \Vert x_i - x_{\mr{c},i} \Vert_{Q_i}^2.
        \end{equation}
\end{assumption}
See, e.g., \cite[Lemma 5]{Koehler2020b} for a sufficient condition for Assumption~\ref{asm:stabilising_terminal_ingredients}, and~\cite{Kohler.2020} for ways to compute these generalized terminal ingredients offline using linear matrix inequalities.
If the system is locally uniformly finite time controllable, a terminal equality constraint, i.e., $\mathcal{X}_i^\mr{f}(r_{\mr{c},i}) = \{x_{\mr{c},i}\} $ and $V_i^\mr{f} \equiv 0$, is also sufficient~\cite{Koehler2020b}.

We now turn our attention to characterizing cooperative goals, which we do with a set and an associated cost used in the objective function of the MPC.
\begin{definition}\label{def:cooperation_set_and_cost}
    A non-empty set $\mathcal{Y}^{\mathrm{c}} \subseteq \prod_{i=1}^m \mathcal{Y}_i$ is called an \emph{output cooperation set} if it is compact and convex, and the cooperative goal is achieved whenever $y \in \mathcal{Y}^{\mathrm{c}}$.
    A continuously differentiable function $V^\mr{c}: \prod_{i=1}^m \mathcal{Y}_i \to \mathbb{R}_{\ge 0}$ is an associated \emph{cost for cooperation} if it has the following properties.
    First, there exist $\underline{\alpha}, \bar{\alpha} \in \mathcal{K}_\infty $ such that
    \begin{equation}\label{eq:cost_for_cooperation_set_indicator}
        \underline{\alpha}(\vert y \vert_{\mathcal{Y}^{\mathrm{c}}}) \le V^\mathrm{c}(y) \le \bar{\alpha}(\vert y \vert_{\mathcal{Y}^{\mathrm{c}}}),
    \end{equation}
    i.e., it gives a good indication of the distance of the global output to the cooperation set.
    Second, $V^\mr{c}$ is convex.
    Third, it should be separable such that it is compatible with the communication topology $\mathcal{G}$ of the multiagent system
    \begin{equation}\label{eq:cost_for_cooperation_distributable}
        V^\mathrm{c}(y) = \sum_{i=1}^{m} \sum_{j\in\mathcal{N}_i} V_{ij}^\mathrm{c}(y_i, y_{j}).
    \end{equation}
\end{definition}
In Section V, we illustrate how these design parameters may be chosen for the examples of consensus and formation control.
We assume the cost of cooperation $V^\mr{c}$ to be convex since it simplifies our analysis considerably.
As we will see later, the main property that we need is that at least one agent can reduce $V^\mr{c}$ by moving its cooperation output in a suitable direction.
In Section V.B, we demonstrate in a simulation that this assumption could possibly be weakened.

The main principle we employ to solve the control goal, i.e., to achieve $y(t) = \mathrm{col}_{i=1}^m(y_i(t)) \to \mathcal{Y}^{\mathrm{c}}$ as $t \to \infty$, is to add a cooperation output $y_{\mr{c},i}(t)$ as an additional decision variable in the MPC for tracking problem~\eqref{eq:MPC_for_tracking} and use the cost for cooperation as a suitable cost to penalize its distance to the cooperation set $\mathcal{Y}^{\mathrm{c}}$.
The individual tracking problems are thus coupled through the cost for cooperation, each agent decides which cooperation output to track and coordination between the agents is necessary.
Later, we will derive sufficient conditions on the cost for cooperation, such that the cooperation outputs $y_{\mr{c},i}(t)$ converge to the cooperation set and the outputs $y_i(t)$ follow due to the tracking part.

With all necessary assumptions outlined, we define the local objective function for agent $i$ as
\begin{align*}
    &J_i(x_i(t), u_i(\cdot \vert t), y_{\mr{c}, i}(t), \bar{y}_{\mathcal{N}_i}(t)) 
    =
    J_i^\mr{tr}(x_i(\cdot \vert t), u_i(\cdot \vert t), r_{\mr{c},i}(t))
    \\
    &\phantom{====}+ \sum_{j\in\mathcal{N}_i} V_{ij}^\mr{c}(y_{\mr{c}, i}(t), \bar{y}_{\mr{c}, j}(t))
    + V_{ji}^\mr{c}(\bar{y}_{\mr{c}, j}(t), y_{\mr{c}, i}(t))
\end{align*}
where $\bar{y}_{\mr{c}, j}(t)$ denotes the communicated cooperation output that was sent to agent $i$ by its neighbor $j$, and $y_{\mr{c}, i}(t) = h(r_{\mr{c},i}(t))$ (with a slight abuse of notation).
The terms $V_{ji}^\mr{c}(\bar{y}_{\mr{c}, j}(t), y_{\mr{c}, i}(t))$ are included so that agent $i$ does not increase its neighbors' cost more than it reduces its own, facilitating cooperation.

Then, agent $i$'s MPC problem is
\begin{subequations}\label{eq:MPC_for_cooperation}
    \begin{align}
        \min_{u_i(\cdot \vert t), y_{\mr{c}, i}(t)} &J_i(x_i(t), u_i(\cdot \vert t), y_{\mr{c}, i}(t), \bar{y}_{\mathcal{N}_i}(t))
        \label{eq:MPC_for_cooperation_objective_function}
        \\
        \intertext{subject to \eqref{eq:MPC_for_tracking_IC}, \eqref{eq:MPC_for_tracking_dynamics}, \eqref{eq:MPC_for_tracking_constraints}}
        x_i(N \vert t) &\in \mathcal{X}_i^\mr{f}(r_{\mr{c}, i}(t)),
        \\
        y_{\mr{c}, i}(t) &\in \mathcal{Y}_i. \label{eq:MPC_for_cooperation_cooperation_output_constraints}
    \end{align}
\end{subequations}
The solution of~\eqref{eq:MPC_for_cooperation} at time $t$ is denoted by $u_i^*(\cdot \vert t)$ and $y_{\mr{c}, i}^*(t)$ and depends on the current measurement $x_i(t)$ and the available communicated cooperation outputs $\bar{y}_{\mathcal{N}_i}(t)$ of the neighbors.
(If the solution is not unique, a static map needs to be used that selects a unique element of the solution set.)

As is proven momentarily, the cooperation goal is then achieved by the following algorithm.
\begin{algorithm}[Sequential DMPC scheme for cooperation]\label{alg:sequential_MPC}
    Initialization of agent $i$: Assume anything for $\bar{y}_{j}(0)$ for $j \in \mathcal{N}_i$, e.g., $\bar{y}_{j}(0) \in \mathcal{Y}_i$, and compute $u_i^*(\cdot \vert 0)$ and $y_{\mr{c}, i}^*(0)$. Share $y_{\mr{c}, i}^*(0)$ with neighbors.\newline
    At each time step $t \in \mathbb{N}$:
    \begin{enumerate}
        \item Sequentially, for $i = 1,\dots, m$:
        If agent $i$ has received $y_{\mr{c}, j}^*(t)$, set $\bar{y}_{j}(t) = y_{\mr{c}, j}^*(t)$, if not, set $\bar{y}_{j}(t) = y_{\mr{c}, j}^*(t-1)$.
        Then, agent $i$ solves~\eqref{eq:MPC_for_cooperation} and sends $y_{\mr{c}, i}^*(t)$ to its neighbors.
        \label{alg:sequential_MPC_step1}
        \item Each agent applies $u_{\mathrm{MPC},i}(t) = u_i^*(0 \vert t)$. Go to Step 1.
    \end{enumerate}
\end{algorithm}
\begin{remark}
    As far as recursive feasibility and stability of the cooperation set, i.e., fulfillment of the cooperation goal, is concerned, it is not necessary that all agents solve their optimization problem in the same time step.
    Instead, this could also be spread over many, e.g., $m$, time steps.
    The reason is that each agent can simply continue to track its currently optimal cooperation output until it is their turn in the sequence.
    Thus, agents may solve~\eqref{eq:MPC_for_cooperation} in parallel, if they hold their cooperation output constant, i.e., set $y_{\mathrm{c},i}^*(t) = y_{\mathrm{c},i}^*(t-1)$, when it is not their turn.
    See also~\cite{Muller.2014,Trodden.2013} for a comparable discussion.
    In addition, although in this scheme we assume instant and lossless communication, this suggests that communication delays or dropouts could be handled by the scheme, perhaps subject to minor adaptions.
    Furthermore, as also pointed out in \cite[Section 3.2.3]{Muller.2014}, it is simple to show that two non-neighboring systems can optimize in parallel, i.e., perform Step~\ref{alg:sequential_MPC_step1}) of Algorithm~\ref{alg:sequential_MPC} concurrently. 
    Scalability of the proposed distributed MPC scheme hence depends on the topology of the communication graph $\mathcal{G}$.
    A rigorous investigation of this is open to future work.
\end{remark}

%===============================================================================
\section{Analysis of the closed-loop system}\label{sec:analysis}
We prove in this section that the global closed-loop system asymptotically achieves the cooperative goal. 
It is given by
\begin{subequations}\label{eq:global_closed_loop}
    \begin{align}
        x(t + 1) &= f(x(t), u_\mathrm{MPC}(t)), & x(0) = x_0,
        \label{eq:global_closed_loop_a}
        \\
        y(t) &= h(x(t), u_\mathrm{MPC}(t)), &
    \end{align}
\end{subequations}
with some initial condition $x_0 \in \mathbb{X} = \prod_{i=1}^m \mathbb{X}_i$, $ f = \mathrm{col}_{i=1}^m(f_i)$, $h = \mathrm{col}_{i=1}^m(h_i)$, $x(t) = \mathrm{col}_{i=1}^m(x_i(t))$, and similar for $u_\mathrm{MPC}(t)$ and $y(t)$.
First, we show that the closed loop~\eqref{eq:global_closed_loop} satisfies the state and input constraints by proving recursive feasibility of~\eqref{eq:MPC_for_cooperation} in Algorithm~\ref{alg:sequential_MPC}.
\begin{theorem}\label{thm:recursive_feasibility}
    Let Assumptions~\ref{asm:reference} and \ref{asm:stabilising_terminal_ingredients} hold.
    Then, if Algorithm~\ref{alg:sequential_MPC} is applied for any initial condition $x_0$ for which~\eqref{eq:MPC_for_cooperation} is feasible, the optimization problem~\eqref{eq:MPC_for_cooperation} in Step 1 of Algorithm~\ref{alg:sequential_MPC} is recursively feasible for all agents $i \in \mathcal{V}$.
    Consequently, the closed-loop system~\eqref{eq:global_closed_loop} satisfies the constraints, i.e., $(x(t), u_\mathrm{MPC}(t)) \in \mathcal{Z} = \mathcal{Z}_1 \times \dots \times \mathcal{Z}_m$ for all $t \in \mathbb{N}_0$.
\end{theorem}
\begin{proof}
    Note that the neighbors' candidate cooperation outputs enter~\eqref{eq:MPC_for_cooperation} in the objective function~\eqref{eq:MPC_for_cooperation_objective_function} only.
    Therefore, feasibility of the problem is independent of the neighbors' decisions.
    Hence, the previously optimal cooperation output is again a feasible candidate, i.e., $\hat{y}^\mr{a}_{\mr{c},i}(t+1) \coloneqq y_{\mr{c},i}^*(t)$.
    A feasible input sequence can be generated from Assumption~\ref{asm:stabilising_terminal_ingredients} by shifting the previously optimal one and appending the terminal control law, i.e., $\hat{u}^\mr{a}_i(\cdot \vert t + 1) = (u_i^*(1 \vert t), \dots, u_i^*(N-1 \vert t), k_i^\mr{f}(x_i^*(N \vert t), r_{\mr{c}, i}^*(t)))$.
    This candidate input sequence is standard in MPC (see, e.g.,~\cite{Rawlings.2020}) and the candidate is feasible for all $i \in \mathcal{V}$.
    Constraint satisfaction of the closed loop then follows from the constraints in~\eqref{eq:MPC_for_tracking_constraints} for $i \in \mathcal{V}$ with $k=0$.
\end{proof}

Given a cooperation equilibrium $(x_{\mr{c}, i}, u_{\mr{c}, i})$, the output relation~\eqref{eq:output_relation} provides a corresponding cooperation output $y_{\mr{c}, i}$.
The following assumption, which is standard in MPC for tracking (cf. \cite[Assumption 6]{Koehler2020b}, \cite[Assumption 1]{Limon2018}) states that incrementally moving the cooperation output also continuously moves a corresponding cooperation equilibrium.
\begin{assumption}\label{asm:output_state_input_relation}
    There exist locally Lipschitz, injective functions $g_{x,i}: \mathcal{Y}_i \subset \mathbb{R}^p \to \mathbb{X}_i$ and $g_{u,i}: \mathcal{Y}_i \subset \mathbb{R}^p \to \mathbb{U}_i$, i.e.,
    $x_{\mr{c}, i} = g_{x,i}(y_{\mr{c}, i})$ and $u_{\mr{c}, i} = g_{u,i}(y_{\mr{c}, i})$ are unique for any $y_{\mr{c}, i}$.
\end{assumption}
See, e.g.,~\cite[Remark 1]{Limon2018} for a sufficient condition to satisfy Assumption~\ref{asm:output_state_input_relation} using the Jacobians of $f_i$ and $h_i$.
Given the output cooperation set $\mathcal{Y}^{\mathrm{c}}$, the state cooperation set is
$
    \mathcal{Y}^{\mathrm{c}}_x = 
    \{
        x_{\mr{c}} = g_{x}(y_{\mr{c}}) \mid y_{\mr{c}} \in \mathcal{Y}^{\mathrm{c}}
    \}
$
with $g_{x} = \mathrm{col}_{i=1}^m(g_{x,i})$.
We denote by $g_{x,i}^{-1}$ the left inverse of $g_{x,i}$, i.e., $y_{\mr{c},i} = g_{x,i}^{-1}(g_{x,i}(y_{\mr{c},i}))$, and similarly for $g_{u,i}$.

It is essential for the proposed distributed MPC scheme that the agents are able to decrease the cooperation cost.
This is possible without further coordination between agents, if a Gauss-Seidel method~\cite{Bertsekas.2016} (or block-coordinate descent) may reduce the cooperation cost in each iteration. The following assumption on Lipschitz continuity of the gradient of the decomposed parts of the cost for cooperation implies this. This assumption is standard in gradient-descent algorithms~\cite{Bertsekas.1997}.
\begin{assumption}\label{asm:cost_gradient_Lipschitz}
    For all agents $i\in\mathcal{V}$, define
    $\bar{V}_i^\mr{c}(y_i) = \sum_{j \in \mathcal{N}_i} V_{ij}^\mr{c}(y_i, \bar{y}_j) + V_{ji}^\mr{c}(\bar{y}_j, y_i)$ with $\bar{y}_j$ fixed for all $j \in \mathcal{N}_i$.
    The gradient of $\bar{V}_i^\mr{c}(y_i)$ is Lipschitz continuous on $\mathcal{Y}_i$ with Lipschitz constant $L_i$, i.e., 
    \begin{equation}\label{eq:cost_gradient_Lipschitz}
         \Vert \nabla\bar{V}_i^\mr{c}(y^1_i) - \nabla\bar{V}_i^\mr{c}(y^2_i) \Vert \le L_i \Vert y^1_i - y^2_i \Vert
    \end{equation}
    for all $y^1_i, y^2_i \in \mathcal{Y}_i$ and all $\bar{y}_j \in \mathcal{Y}_j$, $j \in \mathcal{N}_i$.
\end{assumption}
Note that Assumption~\ref{asm:cost_gradient_Lipschitz} is not needed in an implementation of Algorithm~\ref{alg:sequential_MPC}, but is only used in the following analysis.

The following lemma shows that there eventually exists an alternative feasible candidate, where the cooperation output is incrementally changed and the cost for cooperation is reduced.
Note that this candidate is not needed in the implementation of Algorithm~\ref{alg:sequential_MPC}.
\begin{lemma}\label{lm:cooperation_candidate}
    Let Assumptions~\ref{asm:reference}--\ref{asm:cost_gradient_Lipschitz} hold.
    Define the projected gradient-descent update
    $
        T(y_{\mr{c},i}^*(t)) = P_{\mathcal{Y}_i}\big[ y_{\mr{c},i}^*(t) - \tilde{\theta}_i\nabla\bar{V}_i^\mr{c}(y_{\mr{c},i}^*(t)) \big] 
    $
    with $\tilde{\theta}_i > 0$, and $\bar{V}_i^\mr{c}(\cdot)$ from Assumption~\ref{asm:cost_gradient_Lipschitz} with $\bar{y}_j(t)$ from Step~\ref{alg:sequential_MPC_step1}) of Algorithm~\ref{alg:sequential_MPC}.
    Moreover, define
    \begin{equation}\label{eq:incremental_candidate_set}
        \mathcal{V}_\mr{b}
        =
        \left\{
            i \mid \Vert x_i(t) - x_{\mr{c}, i}^*(t) \Vert_{Q_i}^2 \le \gamma_i \Vert T(y_{\mr{c},i}^*) - y_{\mr{c}, i}^* \Vert^2
        \right\}
    \end{equation}
    with some constant $\gamma_i > 0$ (further specified below in the proof of Theorem~\ref{thm:asymptotic_stability}) and $y_{\mr{c},i}^* = y_{\mr{c},i}^*(t)$.
    Let~\eqref{eq:MPC_for_cooperation} be feasible at time $t$ for agent $i \in \mathcal{V}_\mathrm{b}$.
    Then, the following is satisfied.

    1) there exist $\theta_i \in (0, 1]$ and $\tilde{\theta}_i \in (0, \frac{2}{\theta_i L_i})$, such that there exists a feasible candidate $(\hat{u}_i^\mr{b}(\cdot \vert t + 1), \hat{y}_{\mr{c},i}^\mr{b}(t+1))$ at time $t+1$, where $\hat{y}_{\mr{c},i}^\mr{b}(t+1) = y_{\mr{c},i}^*(t) + \theta_i ( T(y_{\mr{c},i}^*(t)) - y_{\mr{c},i}^*(t))$.
        
    2) Furthermore, with the corresponding predicted state sequence $\hat{x}^\mr{b}_i(\cdot \vert t+1)$ and cooperation equilibrium $\hat{r}_i^{\mathrm{b}}(t+1)$, the tracking cost satisfies
        \begin{align}\label{eq:candidate_tracking_cost_upper_bound}
            &J_i^\mr{tr}(\hat{x}^\mr{b}_i(\cdot \vert t+1), \hat{u}^\mr{b}_i(\cdot \vert t+1), \hat{r}^\mr{b}_i(t+1))
            \notag
            \\
            &\le
            2c_{\mr{u},i} \Vert x_i(t+1) - x^*_{\mr{c},i}(t) \Vert_{Q_i}^2
            \notag
            \\
            &\phantom{\le{}}
            + 2 L_{g_{x,i}}^2 \lambda_{\max}(Q_i)c_{\mr{u},i} \Vert \hat{y}^\mr{b}_{\mr{c},i}(t+1) - y^*_{\mr{c},i}(t)\Vert^2.
        \end{align}
        
    3) Moreover,
        \begin{align}\label{eq:cooperation_cost_decrease}
            &\bar{V}_i^\mr{c}(\hat{y}_{\mr{c},i}^\mr{b}(t+1)) - \bar{V}_i^\mr{c}(y_{\mr{c},i}^*)
            \le -\kappa_i \Vert T(y_{\mr{c},i}^*) - y_{\mr{c},i}^* \Vert^2
        \end{align}
        holds for 
        $\kappa_i = (2\theta_i - \tilde{\theta}_iL_i\theta_i^2)(2\tilde{\theta}_i)^{-1} > 0$ and $y_{\mr{c},i}^* = y_{\mr{c},i}^*(t)$.
        \label{lm:cooperation_candidate_statement3}
\end{lemma}

\begin{proof}
    \emph{Proof of~1)}:
    Note that, 
    \begin{equation}\label{eq:candidate_distance_bound}
        \Vert \hat{y}_{\mr{c},i}^\mr{b}(t+1) - y_{\mr{c},i}^*(t) \Vert = \theta_i \Vert T(y_{\mr{c},i}^*(t)) - y_{\mr{c},i}^*(t) \Vert.
    \end{equation}
    Because $\mathcal{Y}_i$ is compact, there exists $c_{\mathcal{Y}_i} > 0$ such that $\Vert T(y_{\mathrm{c},i}) - y_{\mathrm{c},i} \Vert^2 \le c_{\mathcal{Y}_i}$ for all $y_{\mathrm{c},i} \in \mathcal{Y}_i$.
    Then, if 
    $\gamma_i \le (\epsilon_i^2)(c_{\mathcal{Y}_i})^{-1}$ 
    with $\epsilon_i$ from Assumption~\ref{asm:stabilising_terminal_ingredients}, 
    \begin{equation}\label{eq:bound_current_target_distance}
        \Vert x_i(t) - x_{\mr{c},i}^*(t) \Vert_{Q_i}^2 \stackrel{\eqref{eq:incremental_candidate_set}}{\le} \gamma_i c_{\mathcal{Y}_i} \le \epsilon_i^2
    \end{equation}
    holds for all $i \in \mathcal{V}_{\mathrm{b}}$.
    From this, if also 
    $\gamma_i \le (\epsilon_i^2)(4 c_{\mr{u},i} c_{\mathcal{Y}_i})^{-1}$,
    \begin{align}\label{eq:bound_future_target_distance}
        &\Vert x_i^*(1 \vert t) - x_{\mr{c},i}^*(t) \Vert_{Q_i}^2 
        \le W_i(x_i(t), r_{\mr{c},i}^*(t))
        \notag
        \\
        &\le c_{\mr{u},i} \Vert x_i(t) - x_\ci^*(t) \Vert_{Q_i}^2 
        \stackrel{\eqref{eq:incremental_candidate_set}}{\le}
        c_{\mr{u},i} \gamma_i \Vert T(y_{\mr{c},i}^*(t)) - y_{\mr{c}, i}^*(t) \Vert^2
        \notag
        \\
        &\le c_{\mr{u},i} \gamma_i c_{\mathcal{Y}_i} \le 0.25\epsilon_i^2
    \end{align}
    where the first inequality follows from the definition of $W_i(x_i(t), r_{\mr{c},i}^*(t))$ in~\eqref{eq:MPC_for_tracking_cost} and the second inequality follows from Assumption~\ref{asm:stabilising_terminal_ingredients} and~\eqref{eq:bound_current_target_distance}.
    From Assumption~\ref{asm:output_state_input_relation}, there exists a cooperation state corresponding to the cooperation output, 
    $\hat{x}^\mr{b}_{\mr{c},i}(t+1) = g_{x,i}(\hat{y}^\mr{b}_{\mr{c},i}(t+1))$.
    Then, for all $i \in \mathcal{V}_\mr{b}$, Assumption~\ref{asm:output_state_input_relation} together with the fact that $x(t+1) = x^*(1 \vert t)$ by~\eqref{eq:global_closed_loop_a} leads to
    $
        \Vert x_i(t+1) - \hat{x}^\mr{b}_{\mr{c}, i}(t+1) \Vert_{Q_i} 
        %\\
        %
        \le \Vert x_i(t+1) - x_{\mr{c}, i}^*(t) \Vert_{Q_i} 
        + \Vert x_{\mr{c}, i}^*(t) - \hat{x}^\mr{b}_{\mr{c}, i}(t+1) \Vert_{Q_i}
        %\\
        %
        = \Vert x_i(t+1) - x_{\mr{c}, i}^*(t) \Vert_{Q_i} 
        %\\
        + \Vert g_{x,i}(y_{\mr{c},i}^*(t)) - g_{x,i}(\hat{y}^\mr{b}_{\mr{c},i}(t+1)) \Vert_{Q_i}
        %\\
        %
        \le \Vert x_i(t+1) - x_{\mr{c}, i}^*(t) \Vert_{Q_i} + L_{g_{x,i}} \Vert y_{\mr{c},i}^*(t) - \hat{y}^\mr{b}_{\mr{c},i}(t+1) \Vert_{Q_i}
        %\\
        %
        \stackrel{\eqref{eq:candidate_distance_bound}, \eqref{eq:bound_future_target_distance}}{\le} 
        \frac{\epsilon_i}{2} + \theta_i L_{g_{x,i}} \sqrt{\lambda_{\max}(Q_i)} \Vert T(y_{\mr{c}, i}^*(t)) - y_{\mr{c}, i}^*(t) \Vert
        %\\
        %
        \le \frac{\epsilon_i}{2} + \theta_i L_{g_{x,i}} \sqrt{\lambda_{\max}(Q_i)} c_{\mathcal{Y}_i}\le \epsilon_i
    $
    where the last inequality follows if $\theta_i \le (\epsilon_i)(2L_{g_{x,i}}\sqrt{\lambda_{\max}(Q_i)}c_{\mathcal{Y}_i})^{-1}$.
    Hence, Assumption~\ref{asm:stabilising_terminal_ingredients} ensures that a feasible candidate input sequence $\hat{u}^\mr{b}_i(\cdot \vert t+1)$ exists such that $\hat{u}^\mr{b}_i(\cdot \vert t+1)$ and $\hat{y}_{\mr{c},i}^\mr{b}(t+1)$ are feasible candidates in~\eqref{eq:MPC_for_cooperation}.

    \emph{Proof of~2)}:
    We provide a bound on the incurred increase in the tracking cost, if $\hat{y}_{\mr{c},i}^\mr{b}(t+1)$ is used.
    By Assumption~\ref{asm:stabilising_terminal_ingredients}, the tracking cost satisfies 
    $
        J_i^\mr{tr}(\hat{x}^\mr{b}_i(\cdot \vert t+1), \hat{u}^\mr{b}_i(\cdot \vert t+1), \hat{r}^\mr{b}_i(t+1))
        \stackrel{\eqref{eq:tracking_value_function_upper_bound}}{\le}
        c_{\mr{u},i} \Vert x_i(t+1) - \hat{x}^\mr{b}_{\mr{c},i}(t+1) \Vert_{Q_i}^2 
        \le 2c_{\mr{u},i} ( \Vert x_i(t+1) - x^*_{\mr{c},i}(t) \Vert_{Q_i}^2 
        +  \Vert x^*_{\mr{c},i}(t) - \hat{x}^\mr{b}_{\mr{c},i}(t+1) \Vert_{Q_i}^2 )
        \le
        2c_{\mr{u},i} \Vert x_i(t+1) - x^*_{\mr{c},i}(t) \Vert_{Q_i}^2
        + 2 L_{g_{x,i}}^2 \lambda_{\max}(Q_i)c_{\mr{u},i} \Vert \hat{y}^\mr{b}_{\mr{c},i}(t+1) - y^*_{\mr{c},i}(t)\Vert^2,
    $
    where the last inequality follows from Assumption~\ref{asm:output_state_input_relation}, showing~\eqref{eq:candidate_tracking_cost_upper_bound}.

    \emph{Proof of~3)}:
    As in Assumption~\ref{asm:cost_gradient_Lipschitz}, we consider again
    $
        \bar{V}_i^\mr{c}(y_{\mr{c}, i}(t)) 
        = 
        \sum_{j \in \mathcal{N}_i, j < i} V_{ij}^\mr{c}(y_{\mr{c}, i}(t), y_{\mr{c}, j}^*(t+1)) + 
        V_{ji}^\mr{c}(y_{\mr{c}, j}^*(t+1), y_{\mr{c}, i}(t))
        + 
        \sum_{j \in \mathcal{N}_i, j > i} V_{ij}^\mr{c}(y_{\mr{c}, i}(t), y_{\mr{c}, j}^*(t)) + V_{ji}^\mr{c}(y_{\mr{c}, j}^*(t), y_{\mr{c}, i}(t))
    $
    as a function of one variable.
    Due to the sequential nature of Step 1 in Algorithm~\ref{alg:sequential_MPC}, this is part of what agent $i$ minimizes, i.e.,
    neighbors with $j < i$ already computed and communicated $y_{\mr{c}, j}^*(t+1)$ and neighbors with $j > i$ did not yet compute an update to $y_{\mr{c}, j}^*(t)$.
    We may then show~\eqref{eq:cooperation_cost_decrease} as in~\cite[Proposition 3.3, p. 213]{Bertsekas.1997}. 
    For convenience, we will repeat and adapt the proof in~\cite{Bertsekas.1997} to our setting.
    Note that the projection $T(y_{\mr{c},i}^*(t))$ is the solution to 
    $\inf_{y_i \in \mathcal{Y}_i} \Vert y_i - y_{\mr{c},i}^*(t) + \tilde{\theta}_i \nabla \bar{V}_i^\mr{c}(y_{\mr{c},i}^*(t)) \Vert^2$.
    Thus, by optimality of $T(y_{\mr{c},i}^*(t))$,
    $
        (y_i - T(y_{\mr{c},i}^*(t)))^\top (T(y_{\mr{c},i}^*(t)) - y_{\mr{c},i}^*(t) + \tilde{\theta}_i \nabla \bar{V}_i^\mr{c}(y_{\mr{c},i}^*(t)) \ge 0
    $
    for all $y_i \in \mathcal{Y}_i$.
    The choice $y_i = y_{\mr{c}, i}^*(t) \in \mathcal{Y}_i$ implies
    \begin{align}\label{eq:proof_thm1_decrease_candidate_a_intermediate_1}
        (T(y_{\mr{c},i}^*) - y_{\mr{c}, i}^*)^\top \nabla \bar{V}_i^\mr{c}(y_{\mr{c},i}^*)
        \le -\tilde{\theta}_i^{-1} \Vert T(y_{\mr{c},i}^*) - y_{\mr{c}, i}^* \Vert^2
    \end{align}
    where $y_{\mr{c}, i}^* = y_{\mr{c}, i}^*(t)$.
    Note that $T(y_{\mr{c},i}^*(t)) = y_{\mr{c}, i}^*(t)$ if and only if $y_{\mr{c}, i}^*(t)$ minimizes $\bar{V}_i^\mr{c}$ over $\mathcal{Y}_i$~\cite[Proposition~3.3, p.~213]{Bertsekas.1997}.
    With Assumption~\ref{asm:cost_gradient_Lipschitz}, we can apply the Descent lemma~\cite[Lemma 2.1, p. 203]{Bertsekas.1997} to arrive at~\eqref{eq:cooperation_cost_decrease}:
    $
        \bar{V}_i^\mr{c}(\hat{y}_{\mr{c},i}^\mr{b}(t+1)) 
        \le 
        \bar{V}_i^\mr{c}(y_{\mr{c},i}^*(t)) + \theta_i(T(y_{\mr{c},i}^*(t)) - y_{\mr{c}, i}^*(t))^\top \nabla \bar{V}_i^\mr{c}(y_{\mr{c},i}^*(t)) 
        \phantom{\le{}} + 0.5L_i\theta_i^2 \Vert T(y_{\mr{c},i}^*(t)) - y_{\mr{c}, i}^*(t) \Vert^2
        \stackrel{\eqref{eq:proof_thm1_decrease_candidate_a_intermediate_1}}{\le} 
        \bar{V}_i^\mr{c}(y_{\mr{c},i}^*(t)) - (2\theta_i - \tilde{\theta}_iL_i\theta_i^2)(2\tilde{\theta}_i)^{-1} \Vert T(y_{\mr{c},i}^*(t)) - y_{\mr{c}, i}^*(t) \Vert^2.
    $
\end{proof}

In order to show that the global closed-loop output $y(t)$ converges to the output cooperation set $\mathcal{Y}^{\mathrm{c}}$, we prove that the closed loop~\eqref{eq:global_closed_loop} is asymptotically stable with respect to the corresponding state set $\mathcal{Y}^{\mathrm{c}}_x$.
We define the set of feasible states as $\mathcal{X}_N = \{ x \in \mathbb{X} \mid \eqref{eq:MPC_for_cooperation} \text{ is feasible } \forall i \in \mathcal{V} \}$.
Consider the solution of~\eqref{eq:MPC_for_cooperation} inserted into each local tracking cost and local cost for cooperation part summed up over all agents, i.e., 
$
    V(x) 
    = 
    \sum_{i=1}^m
    J_i^\mr{tr}(x_i^*, u_i^*, r_{\mr{c},i}^*)
    +
    V^\mr{c}(y_{\mr{c}}^*)
$, where $x_i^*$ and $u_i^*$ are the optimal sequences.

We first establish a lower bound on $V(x)$.
\begin{lemma}\label{lm:Lyapunov_function_lower_bound}
    Let Assumption~\ref{asm:output_state_input_relation} hold.
    Then, there exists $\alpha_\mr{lo}\in\mathcal{K}_\infty $ such that $V(x) \ge \alpha_\mr{lo}(\vert x \vert_{\mathcal{Y}^{\mathrm{c}}_x})$ for all $x \in \mathcal{X}_N$.
\end{lemma}
\begin{proof}
    Recall $\vert y_\mathrm{c}^* \vert_{\mathcal{Y}^{\mathrm{c}}} = \Vert y_\mathrm{c}^* - y_{\mathcal{Y}^{\mathrm{c}}}^* \Vert $ where $ y_{\mathcal{Y}^{\mathrm{c}}}^*= \arg\min_{\tilde{y}\in\mathcal{Y}^{\mathrm{c}}} \Vert y_\mathrm{c}^* - \tilde{y} \Vert $ and $x_{\mathcal{Y}^{\mathrm{c}}_x}^* = g_{x}(y_{\mathcal{Y}^{\mathrm{c}}}^*)$.
    We relate this to a state difference using Assumption~\ref{asm:output_state_input_relation}:
    $
    L_{g, x}\Vert y_\mathrm{c}^* - y_{\mathcal{Y}^{\mathrm{c}}}^* \Vert \ge \Vert g_{x}(y_\mathrm{c}^*) - g_{x}(y_{\mathcal{Y}^{\mathrm{c}}}^*) \Vert = \Vert x_\mathrm{c}^* - x_{\mathcal{Y}^{\mathrm{c}}_x}^* \Vert 
    $
    where $L_{g, x} > 0$ is the Lipschitz constant of $g_{x} = [g_{x,1}^\top \: \dots \: g_{x,m}^\top]^\top$ on the compact set $\mathcal{Y}^\mr{c}_x$.
    Thus, from~\eqref{eq:cost_for_cooperation_set_indicator}, $V^\mr{c}(y_\mr{c}^*) \ge \underline{\alpha}(\vert y_\mathrm{c}^* \vert_{\mathcal{Y}^{\mathrm{c}}}) \ge \underline{\alpha}(L_{g, x}^{-1} \Vert x_\mathrm{c}^* - x_{\mathcal{Y}^{\mathrm{c}}_x}^* \Vert)$.
    Hence,
    $
        V(x) \ge \sum_{i=1}^m \Vert x_i - x^*_{\mr{c},i}\Vert_{Q_i}^2
        + V^\mr{c}(y_\mr{c}^*)
        \ge \lambda_{\min}(Q_i)\Vert x - x^*_{\mr{c}}\Vert^2
        + \underline{\alpha}(\vert y_\mr{c}^* \vert_{\mathcal{Y}^{\mathrm{c}}})
        \ge \tilde{\alpha}_\mr{lo}(\Vert x - x^*_{\mr{c}}\Vert)
        + \tilde{\alpha}_\mr{lo}(\Vert x_\mr{c}^* - x_{\mathcal{Y}^{\mathrm{c}}_x}^* \Vert)
        \ge \tilde{\alpha}_\mr{lo}(0.5(\Vert x - x^*_{\mr{c}}\Vert
        + \Vert x_\mr{c}^* - x_{\mathcal{Y}^{\mathrm{c}}_x}^* \Vert))
        \ge \tilde{\alpha}_\mr{lo}(0.5\Vert x - x_{\mathcal{Y}^{\mathrm{c}}_x}^* \Vert)
        \eqqcolon
        \alpha_\mr{lo}(\Vert x - x_{\mathcal{Y}^{\mathrm{c}}_x}^* \Vert) \ge \alpha_\mr{lo}(\vert x \vert_{\mathcal{Y}^{\mathrm{c}}_x})
    $
    where $\tilde{\alpha}_\mr{lo}(s) = \min(\lambda_{\min}(Q_i)s^2, L_{g, x}^{-1}\underline{\alpha}(s)) \in \mathcal{K}_\infty$.
    The fifth inequality holds since for any $\alpha \in \mathcal{K}, a, b \ge 0$: $\alpha(a + b) \le \alpha(2a) + \alpha(2b)$~\cite{Sontag.1989}.
    Note that $\Vert x - x_{\mathcal{Y}^{\mathrm{c}}_x}^* \Vert \ge \min_{\tilde{x}\in{\mathcal{Y}^{\mathrm{c}}_x}} \Vert x - \tilde{x} \Vert = \vert x \vert_{\mathcal{Y}^{\mathrm{c}}_x}$.
\end{proof}
 
Next, we show an upper bound on $V(x(t))$.
\begin{lemma}\label{lm:Lyapunov_function_upper_bound}
    Let Assumptions~\ref{asm:reference}--\ref{asm:output_state_input_relation} hold.
    Then, there exists $\alpha_\mr{up}\in\mathcal{K}_\infty $ such that $V(x) \le \alpha_\mr{up}(\vert x \vert_{\mathcal{Y}^{\mathrm{c}}_x})$ for all $x \in \mathcal{X}_N$.
\end{lemma}
\begin{proof}
    Define $Q = \mr{diag}_{i=1}^m(Q_i)$ and $\epsilon = \min_{i\in\mathcal{V}}\epsilon_i$.
    Let 
    $x \in \mathcal{X}_N$ 
    with 
    $\Vert x - \tilde{x}_{\mathcal{Y}^{\mathrm{c}}_x} \Vert^2_Q \le \epsilon^2$ 
    where 
    $\tilde{x}_{\mathcal{Y}^{\mathrm{c}}_x} = \argmin_{\tilde{x}\in{\mathcal{Y}^{\mathrm{c}}_x}} \Vert x - \tilde{x} \Vert$ 
    and 
    $\tilde{y}_{\mathcal{Y}^{\mathrm{c}}} = g_x^{-1}(\tilde{x}_{\mathcal{Y}^{\mathrm{c}}_x})$.
    Then, from Assumption~\ref{asm:stabilising_terminal_ingredients}, for all $i \in \mathcal{V}$, the cooperation output $\tilde{y}_{\mathcal{Y}^{\mathrm{c}},i}$ together with a feasible input sequence $\tilde{u}_{i}$ that minimizes the tracking cost part is a feasible candidate in~\eqref{eq:MPC_for_cooperation}, and
    $
        V(x)
        \le \sum_{i=1}^{m} W_i(x_i, \tilde{r}_{\mathcal{Y}^{\mathrm{c}},i})
        \stackrel{\eqref{eq:tracking_value_function_upper_bound}}{\le}
        \sum_{i=1}^{m} c_{\mr{u},i} \Vert x_i - \tilde{x}_{\mathcal{Y}^{\mathrm{c}}_x,i} \Vert_{Q_i}^2
        \eqqcolon \tilde{\alpha}_\mr{up}(\vert x \vert_{\mathcal{Y}^{\mathrm{c}}_x}).
    $
    The first inequality follows from the cooperation cost vanishing on the cooperation set.
    The existence of the local upper bound with $\tilde{\alpha}_\mr{up} \in \mathcal{K}_\infty$ on a compact subset of $\mathcal{X}_{N}$, together with compactness of $\mathcal{X}_{N}$, establishes the claim (cf.~\cite[Proposition 2.16]{Rawlings.2020}).
\end{proof}

As an intermediate step, we prove that there exists a suitable upper bound on $V$ based on the candidates from Theorem~\ref{thm:recursive_feasibility} and Lemma~\ref{lm:cooperation_candidate}.
\begin{lemma}\label{lm:value_function_upper_bound}
    Let Assumptions~\ref{asm:reference}--\ref{asm:cost_gradient_Lipschitz} hold.
    Define
    \begin{equation}\label{eq:standard_candidate_set}
        \mathcal{V}_\mr{a}
        =
        \left\{
            i \mid \Vert x_i(t) - x_{\mr{c}, i}^*(t) \Vert_{Q_i}^2 \ge \gamma_i \Vert T(y_{\mr{c},i}^*(t)) - y_{\mr{c}, i}^*(t) \Vert^2
        \right\}
    \end{equation}
    and consider again the complementary set from~\eqref{eq:incremental_candidate_set},
    \begin{equation*}
        \mathcal{V}_\mr{b}
        =
        \left\{
            i \mid \Vert x_i(t) - x_{\mr{c}, i}^*(t) \Vert_{Q_i}^2 \le \gamma_i \Vert T(y_{\mr{c},i}^*(t)) - y_{\mr{c}, i}^*(t) \Vert^2
        \right\}.
    \end{equation*}
    Furthermore, let~\eqref{eq:MPC_for_cooperation} be feasible at time $t$.
    Take $(\hat{u}_i^\mr{a}(\cdot \vert t + 1), \hat{y}_{\mr{c},i}^\mr{a}(t+1))$ from Theorem~\ref{thm:recursive_feasibility} and $(\hat{u}_i^\mr{b}(\cdot \vert t + 1), \hat{y}_{\mr{c},i}^\mr{b}(t+1))$ from Lemma~\ref{lm:cooperation_candidate}.
    Define the shorthands
    \begin{align*}
        {J}_i^{\mr{tr},*}(t) &=  J_i^{\mr{tr}}(x^*_i(\cdot \vert t), u^*_i(\cdot \vert t), r^*_{\mr{c},i}(t)),
        \\
        \hat{J}_i^{\mr{tr,b}}(t+1) &=  J_i^\mr{tr}(\hat{x}^\mr{b}_i(\cdot \vert t+1), \hat{u}^\mr{b}_i(\cdot \vert t+1), \hat{r}^\mr{b}_i(t+1)).
    \end{align*}
    Then,
    \begin{align}\label{eq:candidate_upper_bound}
        &V(x(t+1)) - V(x(t))
        \notag
        \\
        &\le \sum_{i \in \mathcal{V}_\mr{b} } \bigl( \hat{J}_i^\mr{tr,b}(t+1) - {J}_i^{\mr{tr},*}(t) - \kappa_i \Vert T(y_{\mr{c},i}^*(t)) - y^*_{\mr{c},i}(t) \Vert^2 \bigr)
        \notag
        \\
        &\phantom{={}} - \sum_{i \in \mathcal{V}_\mr{a} } \ell_i(x_i(t), u_i^*(0 \vert t), r_{\mr{c},i}^*(t)). 
    \end{align}
\end{lemma}
The notationally extensive, but straightforward proof is given in the Appendix.

Finally, we prove that the state cooperation set is asymptotically stable for the global closed-loop system~\eqref{eq:global_closed_loop}, where the region of attraction is the set of feasible initial states of~\eqref{eq:MPC_for_cooperation}.
Hence, the global output converges to the output cooperation set, i.e., the cooperation goal is asymptotically achieved.
\begin{theorem}\label{thm:asymptotic_stability}
    Let Assumptions~\ref{asm:reference}--\ref{asm:cost_gradient_Lipschitz} hold.
    Then, if Algorithm~\ref{alg:sequential_MPC} is applied for any $x_0 \in \mathcal{X}_N$, the state cooperation set $\mathcal{Y}^{\mathrm{c}}_x$ is asymptotically stable for the closed-loop system~\eqref{eq:global_closed_loop} with region of attraction $\mathcal{X}_N$.
\end{theorem}
\begin{proof}
    Recall again~\eqref{eq:candidate_upper_bound}, where we neglect some non-positive terms and arrive at
    \begin{align}\label{eq:Lyapunov_decrease_upper_bound_1}
        &V(x(t+1)) - V(x(t)) \le - \sum_{i \in \mathcal{V}} \Vert x_i(t) - x_{\mr{c},i}^*(t) \Vert_{Q_i}^2 
        \notag\\
        &\phantom{\le{}} + \sum_{i \in \mathcal{V}_\mr{b} } \big( \hat{J}_i^\mr{tr,b}(t+1) - \kappa_i \Vert T(y_{\mr{c},i}^*(t)) - y^*_{\mr{c},i}(t) \Vert^2  \big).
    \end{align}
    In the following, we abbreviate $\star^+ = \star(t+1)$ and $\star = \star(t)$, e.g. $x^+ = x(t+1)$, $x = x(t)$.
    With the upper bound on the tracking cost~\eqref{eq:candidate_tracking_cost_upper_bound},
    \begin{align*}
        &V(x^+) - V(x) \stackrel{\eqref{eq:candidate_tracking_cost_upper_bound}}{\le} - \sum_{i \in \mathcal{V}} \Vert x_i - x_{\mr{c},i}^* \Vert_{Q_i}^2 
        \\
        &\phantom{\le{}} + \sum_{i \in \mathcal{V}_\mr{b} } \big( -\kappa_i \Vert T(y_{\mr{c},i}^*) - y^*_{\mr{c},i} \Vert^2 + 2c_{\mr{u},i} \Vert x_i^+ - x^*_{\mr{c},i} \Vert_{Q_i}^2\big)
        \\
        &\phantom{\le{}} + \sum_{i \in \mathcal{V}_\mr{b} } 2L_{g_{x,i}}^2\lambda_{\max}(Q_i)c_{\mr{u},i} \Vert \hat{y}_{\mr{c},i}^{\mr{b},+} - y_{\mr{c},i}^* \Vert^2.
        \\
        &\stackrel{\eqref{eq:candidate_distance_bound},\eqref{eq:bound_future_target_distance}}{\le}
        - \sum_{i \in \mathcal{V}} \Vert x_i - x_{\mr{c},i}^* \Vert_{Q_i}^2 - \sum_{i \in \mathcal{V}_\mr{b} } \frac{c_{\theta_i}}{2} \Vert T(y_{\mr{c},i}^*) - y^*_{\mr{c},i} \Vert^2
        \\
        &\phantom{={}} + \sum_{i \in \mathcal{V}_\mr{b} }  (2c_{\mr{u},i}^2 \gamma_i - \frac{c_{\theta_i}}{2})  \Vert T(y_\ci^*) - y_\ci^* \Vert^2,
    \end{align*}
    where $c_{\theta_i} \coloneqq \kappa_i - \theta_i^2 2L_{g_{x,i}}^2\lambda_{\max}(Q_i)c_{\mr{u},i}$.
    Recall $\kappa_i = \theta_i(\tilde{\theta}_i)^{-1} - 0.5L_i\theta_i^2$.
    Then,
    $
        c_{\theta_i} = \theta_i( \tilde{\theta}_i^{-1} - \theta_i(0.5L_i + 2L_{g_{x,i}}^2\lambda_{\max}(Q_i)c_{\mr{u},i})) 
    $
    is positive if $\tilde{\theta}_i < 4(L_i + 4L_{g_{x,i}}^2\lambda_{\max}(Q_i)c_{\mr{u},i})^{-1}$.
    If $\gamma_i \le c_{\theta_i}(4c_{\mr{u},i}^2)^{-1}$, and since $-\Vert x_i - x_{\mr{c},i}^* \Vert_{Q_i}^2 \le - \gamma_i\Vert T(y_\ci^*) - y_\ci^* \Vert^2$ for all $i \in \mathcal{V}_\mr{a}$ by~\eqref{eq:standard_candidate_set}, we arrive at
    \begin{align}\label{eq:Lyapunov_function_decrease}
        &V(x^+) - V(x) \le - \sum_{i \in \mathcal{V}} \frac{1}{2}\Vert x_i - x_{\mr{c},i}^* \Vert_{Q_i}^2
        \notag\\
        &\phantom{\le{}} - \sum_{i \in \mathcal{V}_\mr{a} } \frac{1}{2} \gamma_i \Vert T(y_\ci^*) - y_\ci^* \Vert^2 
        - \sum_{i \in \mathcal{V}_\mr{b} } \frac{c_{\theta_i}}{2} \Vert T(y_{\mr{c},i}^*) - y^*_{\mr{c},i}\Vert^2
        \notag\\
        &\le - \frac{1}{2} \sum_{i \in \mathcal{V}} ( \Vert x_i - x_{\mr{c},i}^* \Vert_{Q_i}^2 + \tilde{\gamma} \Vert T(y_\ci^*) - y_\ci^* \Vert^2 )
    \end{align}
    with $\tilde{\gamma} = \min_{i\in\mathcal{V}}\{{c_{\theta_i}, \gamma_i}\}$.
    Now, recall Lemmas~\ref{lm:Lyapunov_function_lower_bound} and~\ref{lm:Lyapunov_function_upper_bound}, i.e., $V(x(t)) \ge \alpha_\mr{lo}(\vert x(t) \vert_{\mathcal{Y}^{\mathrm{c}}_x})$ and $V(x(t)) \le \alpha_\mr{up}(\vert x(t) \vert_{\mathcal{Y}^{\mathrm{c}}_x})$.
    Since $V(x(t))$ is non-increasing and bounded from below, it converges.
    Thus, $\lim_{t \to \infty} T(y_\ci^*(t)) - y_\ci^*(t) = 0$ for all $i \in \mathcal{V}$.
    Consider an accumulation point $y_{\mr{c},\infty}^*$ of $y_\mathrm{c}^*(t)$, i.e., there exists a subsequence $t_k$ such that $\lim_{k\to\infty} y_\mathrm{c}^*(t_k) = y_{\mr{c},\infty}^*$.
    By continuity of $T$, $y_{\mr{c},\infty,i}^* = \lim_{k\to\infty}T(y_{\mr{c},i}^*(t_k)) = T(y_{\mr{c},\infty,i}^*)$ (cf.~\cite[Proposition~3.8]{Bertsekas.1997}),
    which implies that $y_{\mr{c},\infty}^*$ minimizes $V^\mr{c}$ over $\prod_{i=1}^m\mathcal{Y}_i$~\cite[Proposition~3.3]{Bertsekas.1997} and $y_{\mr{c},\infty}^* \in \mathcal{Y}^\mr{c}$.
    Since $g_{x,i}$ is continuous for all $i\in\mathcal{V}$, $\lim_{k\to\infty} x_\mr{c}^*(t_k) = \lim_{k\to\infty} g_{x}(y_\mathrm{c}^*(t_k)) = g_{x}(y_{\mr{c},\infty}^*) \in \mathcal{Y}_x^\mr{c}$.
    In addition, again from~\eqref{eq:Lyapunov_function_decrease} it follows that $\lim_{t \to \infty} x_i(t) - x_{\mr{c},i}^*(t) = 0$ for all $i \in \mathcal{V}$.
    Thus, $\lim_{k\to\infty}x(t_k) \in \mathcal{Y}_x^\mr{c}$, and from Lemma~\ref{lm:Lyapunov_function_upper_bound} with continuity of $\alpha_{\mr{up}}$ and of the point-to-set distance, we have $\lim_{k\to\infty} V(x(t_k)) = 0$.
    Hence, since $V(x(t))$ is non-increasing, $\lim_{t\to\infty} V(x(t)) = 0$, and $\lim_{t\to\infty} \vert x(t) \vert_{\mathcal{Y}_x^\mathrm{c}} = 0$ as well as $\lim_{t\to\infty} \vert y(t) \vert_{\mathcal{Y}^\mathrm{c}} = 0$ from Lemma~\ref{lm:Lyapunov_function_lower_bound}.
    Standard Lyapunov arguments imply stability (cf.~\cite[Theorem.~B.13]{Rawlings.2020}).
\end{proof}

We want to stress that the second candidate $\hat{y}_{\mr{c},i}^\mr{b}(t+1)$ from Lemma~\ref{lm:cooperation_candidate} is only used for the analysis of the proposed scheme and its computation is not necessary, although useful as a warm start in the optimization.
The main idea in the analysis is that each agent is able to incrementally move its cooperation output while reducing the over-all cost for cooperation, offsetting an increase in the tracking cost part.

%===============================================================================
\section{Consensus and formation control}\label{sec:example}
In this section, we illustrate how output cooperation sets and associated costs for cooperation can be formulated for output consensus and formation control as two examples.

\subsection{Consensus}\label{ssec:consensus}
In this example, we demonstrate the formulation of the consensus problem in the proposed framework.
The control goal is to steer the outputs of the agents to a common value.
The output cooperation set is thus defined as 
$
    \mathcal{Y}^\mathrm{c} = \{y \in \prod_{i=1}^m \mathcal{Y}_i \mid \Vert y_i - y_j \Vert = 0 \; \forall (i,j) \in \mathcal{E} \}
$, where $\mathcal{Y}_i$ is a convex set of outputs that correspond to admissible equilibria of agent $i$. 
Hence, $\mathcal{Y}^\mathrm{c}$ is the set of all outputs, which are contained in local convex sets, correspond to equilibria satisfying the local constraints, and are equal.
A natural choice for the cost for cooperation is 
$
    V^\mathrm{c}(y) = \sum_{i=1}^m \sum_{j \in \mathcal{N}_i} \Vert y_i - y_j \Vert^2 = \Vert y \Vert_L^2
$
where $L = [l_{ij}] \in \mathbb{R}^{m\times m}$ is the symmetric Laplacian of the communication graph $\mathcal{G}$ with
$l_{ij} = -1$ if $j \in \mathcal{N}_i$, $l_{ij} = \vert \mathcal{N}_i \vert$ if $j=i$, and $l_{ij} = 0$ otherwise.
Clearly, the cost is separable as in~\eqref{eq:cost_for_cooperation_distributable} and convex.
In addition, since for all $i \in \mathcal{V}$ the cost $\bar{V}_i^\mathrm{c}(y_i) = \sum_{j\in\mathcal{N}_i} \Vert y_i -\bar{y}_j \Vert^2$, with $\bar{y}_j$ fixed, is smooth and $\mathcal{Y}_i$ is compact, Assumption~\ref{asm:cost_gradient_Lipschitz} holds.
To show the upper and lower bound in~\eqref{eq:cost_for_cooperation_set_indicator}, we define the projection $y_{\mathcal{Y}_\mr{c}} = \argmin_{\tilde{y} \in \mathcal{Y}_\mr{c}} \Vert y - \tilde{y} \Vert$, which is not needed in Algorithm 1 and is only introduced here for analysis.
In particular, $y_{\mathcal{Y}_\mr{c}} \in \ker L$ and thus $\Vert y \Vert_L = \Vert y - y_{\mathcal{Y}_\mr{c}} \Vert_L$.
Hence, $V^\mathrm{c}(y) = \Vert y \Vert_L^2 \le \Vert y - y_{\mathcal{Y}_\mr{c}} \Vert_L^2 \le \lambda_{\max}(L) \Vert y - y_{\mathcal{Y}_\mr{c}} \Vert^2 \eqqcolon \bar{\alpha}(\vert y \vert_{\mathcal{Y}^\mr{c}}) $ showing the right-hand side in~\eqref{eq:cost_for_cooperation_set_indicator}.
For the left-hand side, note that $\Vert y \Vert_L^2 \ge \lambda_2(L) \Vert y - y_{P}^\mr{c} \Vert^2$ (cf.~\cite[Theorem 3]{OlfatiSaber2007}) where $y_{P}^\mr{c}$ is the projection of $y$ onto the kernel of $L$ and $\lambda_2(L)$ is the second-smallest eigenvalue of $L$, which is positive if the graph is connected.
Since, by definition, $\Vert y - y_{\mathcal{Y}_\mr{c}} \Vert \le \Vert y - y_{P}^c \Vert$, the left-hand side in~\eqref{eq:cost_for_cooperation_set_indicator} follows with $\underline{\alpha}(\vert y \vert_{\mathcal{Y}^\mr{c}}) \coloneqq \lambda_2(L) \Vert y - y_{\mathcal{Y}_\mr{c}} \Vert^2$.

Hence, using this cost for cooperation, an implementation of the proposed scheme in Algorithm~\ref{alg:sequential_MPC} will asymptotically steer the system to consensus, if Assumptions~\ref{asm:reference}--\ref{asm:output_state_input_relation} are satisfied.
A numerical example can be found in the Appendix.

\subsection{Formation control}\label{ssec:formation}
In this example, we consider the problem of steering quadcopters into a configuration in which their positions have a specified distance while their altitude is the same.
The optimization problem~\eqref{eq:MPC_for_cooperation} is implemented with a terminal equality constraint using CasADi~\cite{Andersson2019} and IPOPT~\cite{Waechter2005}.

We apply the proposed scheme to a multiagent system comprising three $10$-state-quadcopters with dynamics
\begin{align*}
    \dot{x}_{i,1}(t) &= x_{i,4}(t),
    &
    \dot{x}_{i,6}(t) &= -g + 0.91 u_{i,3}(t) ,
    \\
    \dot{x}_{i,2}(t) &= x_{i,5}(t),
    &
    \dot{x}_{i,7}(t) &= -8 x_{i,7}(t) + x_{i,9}(t),
    \\
    \dot{x}_{i,3}(t) &= x_{i,6}(t),
    & 
    \dot{x}_{i,8}(t) &= -8 x_{i,8}(t) + x_{i,10}(t),
    \\
    \dot{x}_{i,4}(t) &= g \tan(x_{i,7}(t)),
    & 
    \dot{x}_{i,9}(t) &= 10 (-x_{i,7}(t) + u_{i,1}(t)),
    \\
    \dot{x}_{i,5}(t) &= g \tan(x_{i,8}(t)),
    &
    \dot{x}_{i,10}(t) &= 10 (-x_{i,8}(t) + u_{i,2}(t)),
\end{align*}
adapted from~\cite{Hu2018} and with $g = 9.81$.
The dynamics are discretized using the Euler method with a step-size of $h = 0.1$.
We choose as outputs the position of the quadcopters, i.e., $y_{i,1} = x_{i,1}$, $y_{i,2} = x_{i,2}$, and $y_{i,3} = x_{i,3}$, where $y_{i,3}$ is the quadcopter's altitude.
For simplicity, the state constraint sets of all quadcopters are set to $\mathbb{X}_i = \{ \Vert x_i \Vert_\infty \le 10 \}$ and
$
\mathcal{Y}_i = 
\{ 
    \begin{bsmallmatrix}
        -8 & -8 & 0
    \end{bsmallmatrix}^\top 
    \le
    y_i
    \le
    \begin{bsmallmatrix}
        8 & 8 & 8
    \end{bsmallmatrix}^\top
\}.
$
The input constraints are chosen as in~\cite{Hu2018} to 
$
\mathbb{U}_i = 
\{ 
    \begin{bsmallmatrix}
        -\frac{\pi}{2} & -\frac{\pi}{2} & 0
    \end{bsmallmatrix}^\top 
    \le
    u_i
    \le
    \begin{bsmallmatrix}
        \frac{\pi}{2} & \frac{\pi}{2} & 2g
    \end{bsmallmatrix}^\top
\}
$
for all agents $i \in \mathcal{V}$.

\begin{figure}[t]
    \setlength\axisheight{0.45\linewidth}
    \setlength\axiswidth{0.95\linewidth}
    \centering
    \input{formation_x1_manual.tex}
    \input{formation_x2_manual.tex}
    \input{formation_x3_manual.tex}
    \caption{Evolution over time of the first three states (the outputs) of the quadcopters in Section~\ref{ssec:formation}.}
    \label{fig:quadcopter_altitude}
\end{figure}
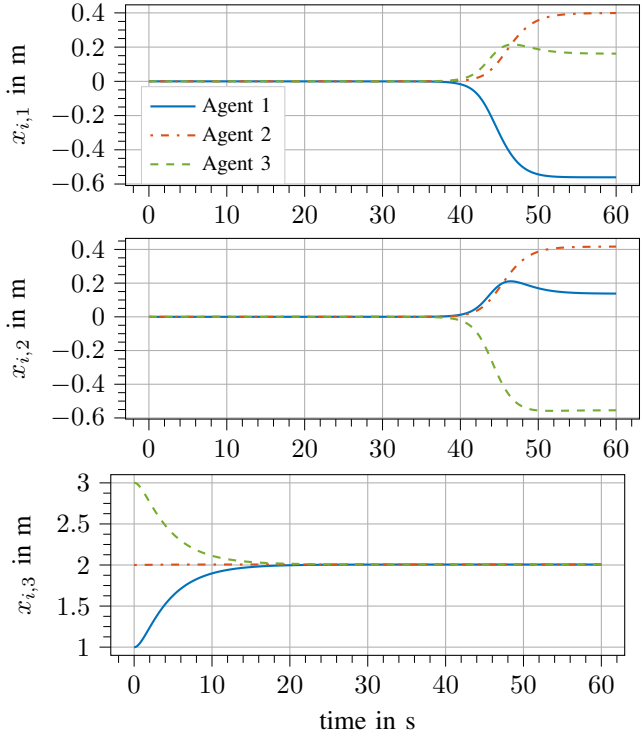
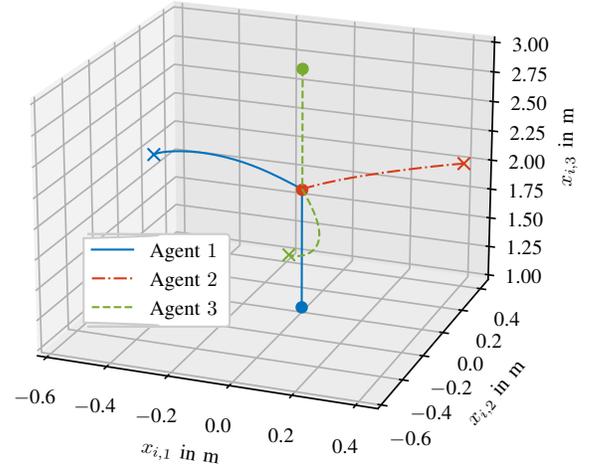
\begin{figure}[t]
    \centering
    \resizebox{0.8\linewidth}{!}  % {width}{height}{content}
    {
    \centerline{\input{formation_outputs_manual.pgf}}
    }
    \caption{Evolution of the first three states of the quadcopters in Section~\ref{ssec:formation} into an equilateral triangle. The initial condition of each trajectory is marked by a point and the final value by an X.}
    \label{fig:quadcopter_outputs}
\end{figure}     

Since we want to enforce a specified distance $d_{ij}$ between the position of neighbors and want consensus in the altitude, we formulate the output cooperation set
\begin{align}\label{eq:formation_control_set}
    \mathcal{Y}^\mathrm{c} = \{y \in {\textstyle \prod_{i=1}^m }\mathcal{Y}_i \mid &\Vert\begin{bsmallmatrix}y_{i,1} \\ y_{i,2}\end{bsmallmatrix} - \begin{bsmallmatrix}y_{j,1} \\ y_{j,2}\end{bsmallmatrix} \Vert = d_{ij} \; \forall (i,j) \in \mathcal{E}, 
    \notag
    \\
    & \Vert y_{i,3} - y_{j,3} \Vert = 0 \; \forall (i,j) \in \mathcal{E} \}.
\end{align}
Directly derived from the constraint in~\eqref{eq:formation_control_set}, a possible choice for the cost for cooperation is given by
$
    V^\mathrm{c}(y) = \sum_{i=1}^m \sum_{j \in \mathcal{N}_i} (\Vert \begin{bsmallmatrix}y_{i,1} \\ y_{i,2}\end{bsmallmatrix} - \begin{bsmallmatrix}y_{j,1} \\ y_{j,2}\end{bsmallmatrix} \Vert^2 - d_{ij}^2)^2 + \Vert y_{i,3} - y_{j,3} \Vert^2.
$

We choose an all-to-all graph, i.e., all agents are neighbors of each other, and want to reach a distance of $d_{ij} = 1$ for all $i,j \in \mathcal{V}$ with $i\neq j$.
Hence, the agents should converge to a common altitude and form an equilateral triangle.

Note that $V^\mathrm{c}(y)$ is not convex for the chosen $\mathcal{Y}_i$ and the projected gradient-descent update, which was used in the analysis of the proposed scheme, vanishes if the outputs of all agents are the same (corresponding to a local maximum).
Nevertheless, the following simulations illustrate that the proposed scheme still works well in this case. 
This is because it is enough that an agent can reduce the cost for cooperation offsetting the increase in its tracking part by incrementally moving its cooperation output if the ones of the neighbors are fixed.
Moreover, all local minima of $V^\mathrm{c}$ (with domain $\prod_{i=1}^m \mathcal{Y}_i$) are zero.
A detailed study of such more general settings with weaker conditions on the cost for cooperation is an interesting topic for future research.

We consider as critical initial conditions,
$x_1^\mr{ic} = \begin{bsmallmatrix} 0 & 0 & 1 & 0 & 0 & 0 & 0 & 0 & 0 & 0 \end{bsmallmatrix}^\top$, $x_2^\mr{ic} = \begin{bsmallmatrix} 0 & 0 & 2 & 0 & 0 & 0 & 0 & 0 & 0 & 0 \end{bsmallmatrix}^\top$, and $x_3^\mr{ic} = \begin{bsmallmatrix} 0 & 0 & 3 & 0 & 0 & 0 & 0 & 0 & 0 & 0 \end{bsmallmatrix}^\top$,
which are such that $\nabla\bar{V}_i^\mr{c}(y_{\mr{c},i}^*(t))$ vanishes in the first two outputs when selecting $y_{\mathrm{c},i}^*(0) = [x_{i,1}^\mr{ic} \; x_{i,2}^\mr{ic} \; x_{i,3}^\mr{ic}]^\top$ as the initial cooperation output.
However, Algorithm~\ref{alg:sequential_MPC} is applied unchanged by providing a suitable warm start to the gradient-based numerical solver.
This warm start is the solution of the optimization problem~\eqref{eq:MPC_for_cooperation} with slightly perturbed cooperation outputs.
The simulation results are depicted in Figures~\ref{fig:quadcopter_altitude}--\ref{fig:quadcopter_outputs}.
The agents achieve the cooperative goal by first converging to a common altitude and then the final formation.
In the Appendix, we provide a simulation for slightly perturbed initial conditions.
%===============================================================================
\section{Conclusion}\label{sec:conclusion}
We propose a framework for cooperative multiagent systems based on a sequential distributed MPC scheme.
The cooperative task is formulated by a suitably chosen cost associated with the set of outputs that solve it.
This cost is combined with a tracking MPC formulation using artificial references.
Cooperation is achieved by agents sharing their currently optimal artificial reference with neighbors.
The local tracking of the artificial reference adds flexibility for each agent.
We show that the closed-loop system is asymptotically stable with respect to a set in the states that corresponds to cooperation in the outputs, thus solving the cooperation problem.
The use of tracking MPC with artificial references allows us to combine the advantages of MPC with well established results from distributed optimization.

Dynamic tasks such as platooning or synchronization, and inclusion of coupling constraints, is subject of future work.
First steps for periodic tasks have been obtained in~\cite{Koehler2023}.

%\section*{Acknowledgment}

\bibliographystyle{unsrt}
\bibliography{references_adapted}

%===============================================================================
\section{Appendix}
\subsection{Proof of Lemma~\ref{lm:value_function_upper_bound}}
\begin{proof}
    We denote with 
    $z^i_{\mr{c}}(t+1) = (y_{\mr{c}, 1}^*(t+1), \dots, y_{\mr{c}, i-1}^*(t+1), y_{\mr{c}, i}^*(t), \dots, y^*_{\mr{c}, m}(t))$ 
    the current cooperative output vector at the beginning of the $i$-th step of the sequence in Step 1 of Algorithm~\ref{alg:sequential_MPC}.
    That is, agents~1 to $i-1$ have updated their cooperative output, whereas agents~$i$ to $m$ have yet to optimize and their communicated previously optimal cooperation output from time $t-1$ is available.
    In addition, $z^{m+1}_{\mr{c}}(t+1) = (y_{\mr{c}, 1}^*(t+1), \dots, y^*_{\mr{c}, m}(t+1))$.
    This enables the following convenient abuse of notation in the argument of $V_{ij}^\mr{c}$, which, of course, only depends on the $i$-th and $j$-th components of $z^i_{\mr{c}}(\cdot)$.
    Since a feasible candidate is an upper bound on the optimal cost, both candidates $(\hat{u}_i^\mathrm{a}(\cdot \vert t+1), \hat{y}_\ci^\mathrm{a}(t+1))$ and $(\hat{u}_i^\mathrm{b}(\cdot \vert t+1), \hat{y}_\ci^\mathrm{b}(t+1))$ provide an upper bound on the optimal cost of agent~$i$ in~\eqref{eq:MPC_for_cooperation}.

    Define the additional shorthands
    \begin{align*}
        \bar{J}_i^{\mr{tr,a}}(t+1) &= {J}_i^{\mr{tr},*}(t) - \ell_i(x_i(t), u_i^*(0 \vert t), r_{\mr{c},i}^*(t)),
        \\
        \bar{J}_i^{\mr{tr,b}}(t+1) &= \hat{J}_i^{\mr{tr,b}}(t+1) - \kappa_i \Vert T(y_{\mr{c},i}^*(t)) - y_{\mr{c},i}^*(t) \Vert^2.
    \end{align*}
    If $(\hat{u}_i^\mathrm{a}(\cdot \vert t+1), \hat{y}_i^\mathrm{a}(t+1))$ is used, the upper bound is
    \begin{align}\label{eq:candidate_upper_bound_for_a}
        &\hspace*{-0.5em}J_i^{\mr{tr},*}(t+1) + \sum_{\mathclap{j \in \mathcal{N}_i}} \left( V_{ij}^\mr{c}(z_\mr{c}^{i+1}(t+1)) + V_{ji}^\mr{c}(z_\mr{c}^{i+1}(t+1)) \right)
        \notag
        \\
        &\hspace*{-0.5em}\le {J}_i^{\mr{tr},*}(t) - \ell_i(x_i(t), u_i^*(0 \vert t), r_{\mr{c},i}^*(t)) 
        \notag\\
        &\hspace*{-0.5em}\phantom{le{}} + \sum_{\mathclap{j \in \mathcal{N}_i}} \left( V_{ij}^\mr{c}(z_\mr{c}^{i}(t+1)) + V_{ji}^\mr{c}(z_\mr{c}^{i}(t+1)) \right)
        \notag\\
        &\hspace*{-0.5em}= \bar{J}_i^{\mr{tr,a}}(t+1) + \sum_{\mathclap{j \in \mathcal{N}_i}} \left( V_{ij}^\mr{c}(z_\mr{c}^{i}(t+1)) + V_{ji}^\mr{c}(z_\mr{c}^{i}(t+1)) \right)
    \end{align}
    where the inequality follows from standard arguments in MPC (cf.~\cite{Rawlings.2020}).

    If $(\hat{u}_i^\mathrm{b}(\cdot \vert t+1), \hat{y}_i^\mathrm{b}(t+1))$ is used, the upper bound is 
    \begin{align}\label{eq:candidate_upper_bound_for_b}
        &\hspace*{-0.6em}J_i^{\mr{tr},*}(t+1) + \sum_{\mathclap{j \in \mathcal{N}_i}} \left( V_{ij}^\mr{c}(z_\mr{c}^{i+1}(t+1)) + V_{ji}^\mr{c}(z_\mr{c}^{i+1}(t+1)) \right)
        \notag
        \\
        &\hspace*{-0.6em}\stackrel{\eqref{eq:cooperation_cost_decrease}}{\le} \hat{J}_i^{\mr{tr,b}}(t+1) 
        + \sum_{\mathclap{j \in \mathcal{N}_i}} 
        \left( V_{ij}^\mr{c}(z_\mr{c}^{i}(t+1)) + V_{ji}^\mr{c}(z_\mr{c}^{i}(t+1)) \right)
        \notag\\
        &\hspace*{-0.6em}\phantom{\le{}} - \kappa_i \Vert T(y_{\mr{c},i}^*(t)) - y_{\mr{c},i}^*(t) \Vert^2
        \notag
        \\
        &\hspace*{-0.6em}= \bar{J}_i^{\mr{tr,b}}(t+1) 
        + \sum_{\mathclap{j \in \mathcal{N}_i}} \big( V_{ij}^\mr{c}(z_\mr{c}^{i}(t+1)) + V_{ji}^\mr{c}(z_\mr{c}^{i}(t+1)) \big).
    \end{align}

    Next, as in~\cite{Muller.2012}, we split the value function into
    \begin{align}\label{eq:candidate_upper_bound_step_1}
        &V(x(t+1)) = \sum_{i=1}^m \big( J_i^{\mr{tr},*}(t+1) + \sum_{\mathclap{j \in \mathcal{N}_i}} V_{ij}^\mr{c}(z_\mr{c}^{m+1}(t+1))\big)
        \notag
        \\
        &= \sum_{{i \in \mathcal{V} \backslash \{ \{m\}\cup \mathcal{N}_m \}}} \big( J_i^{\mr{tr},*}(t+1) + \sum_{j \in \mathcal{N}_i} V_{ij}^\mr{c}(z_\mr{c}^{m}(t + 1)) \big)
        \notag
        \\
        &\phantom{={}} + \sum_{{i \in \mathcal{N}_m}} \big( J_i^{\mr{tr},*}(t+1) 
        + \sum_{j \in \mathcal{N}_i\backslash\{m\}} V_{ij}^\mr{c}(z_\mr{c}^{m}(t+1)) \big)
        \notag
        \\
        &\phantom{={}} + J_m^{\mr{tr},*}(t+1) + \sum_{j \in \mathcal{N}_m} V_{mj}^\mr{c}(z_\mr{c}^{m+1}(t+1)) 
        \notag
        \\
        &\phantom{={}} + \sum_{j \in \mathcal{N}_m} V_{jm}^\mr{c}(z_\mr{c}^{m+1}(t+1))
    \end{align} 
    where equality holds since the coupling costs $V_{ij}^\mr{c}$ with changed arguments do not depend on the cooperation output of agent~$m$.
    Then, either from~\eqref{eq:candidate_upper_bound_for_a} with $\bar{J}_m^{\mr{tr}}(t+1) = \bar{J}_m^{\mr{tr,a}}(t+1)$ if $m \in \mathcal{V}_\mr{a}$ or from~\eqref{eq:candidate_upper_bound_for_b} with $\bar{J}_m^{\mr{tr}}(t+1) = \bar{J}_m^{\mr{tr,b}}(t+1)$ if $m \in \mathcal{V}_\mr{b}$,
    \begin{align}\label{eq:candidate_upper_bound_step_2}
        &V(x(t+1))
        \notag
        \\
        &\le \sum_{{i \in \mathcal{V} \backslash \{\{m\}\cup \mathcal{N}_m\} }} 
        \big( 
            J_i^{\mr{tr},*}(t+1) + \sum_{j \in \mathcal{N}_i} V_{ij}^\mr{c}(z_\mr{c}^{m}(t + 1)) 
        \big)
        \notag
        \\
        &\phantom{={}} + \sum_{{i \in \mathcal{N}_m}} 
        \big( 
            J_i^{\mr{tr},*}(t+1) + \sum_{j \in \mathcal{N}_i} V_{ij}^\mr{c}(z_\mr{c}^{m}(t+1)) 
        \big)
        \notag
        \\
        &\phantom{={}} + \bar{J}_m^{\mr{tr}}(t+1) 
        + \sum_{j \in \mathcal{N}_m} V_{mj}^\mr{c}(z_\mr{c}^{m}(t+1)). 
    \end{align}
    In the following, we use a similar argument as in~\eqref{eq:candidate_upper_bound_step_2} to establish via induction the bound~\eqref{eq:candidate_upper_bound} when going from $k=m$ to $k=1$.
    For this, assume the following upper bound for some $k=2,\dots,m$
    \begin{align}\label{eq:candidate_upper_bound_step_k}
        &V(x(t+1)) 
        \notag
        \\
        &\le \sum_{{i \in \mathcal{V}\backslash \{\{k\}\cup \mathcal{N}_k\} }, i < k} \big( J_i^{\mr{tr},*}(t+1) + \sum_{j \in \mathcal{N}_i} V_{ij}^\mr{c}(z_\mr{c}^{k}(t + 1)) \big)
        \notag
        \\
        &\phantom{\le{}} + \sum_{{i \in \mathcal{V}\backslash \{\{k\}\cup \mathcal{N}_k\} }, i > k} 
        \big( \bar{J}_i^{\mr{tr}}(t+1) + \sum_{j \in \mathcal{N}_i} V_{ij}^\mr{c}(z_\mr{c}^{k}(t + 1)) \big)
        \notag
        \\
        &\phantom{\le{}} + \sum_{{i \in \mathcal{N}_k}, i < k} 
        \big( J_i^{\mr{tr},*}(t+1) 
        + \sum_{j \in \mathcal{N}_i\backslash\{k\}} V_{ij}^\mr{c}(z_\mr{c}^{k}(t+1)) \big)
        \notag
        \\
        &\phantom{\le{}} + \sum_{{i \in \mathcal{N}_k}, i > k} \big( \bar{J}_i^{\mr{tr}}(t+1) 
        + \sum_{j \in \mathcal{N}_i\backslash\{k\}} V_{ij}^\mr{c}(z_\mr{c}^{k}(t+1)) \big)
        \notag
        \\
        &\phantom{\le{}} + J_k^{\mr{tr},*}(t+1) + \sum_{j \in \mathcal{N}_k} V_{kj}^\mr{c}(z_\mr{c}^{k+1}(t+1))
        \notag
        \\
        &\phantom{\le{}} + \sum_{j \in \mathcal{N}_k} V_{jk}^\mr{c}(z_\mr{c}^{k+1}(t+1)).
    \end{align} 
    For the start of the induction, note that~\eqref{eq:candidate_upper_bound_step_k} equals~\eqref{eq:candidate_upper_bound_step_1} for $k = m$.
    We identify again the optimal cost in~\eqref{eq:MPC_for_cooperation} of agent $k$ in the last three terms of~\eqref{eq:candidate_upper_bound_step_k} and by the same arguments as before
    \begin{align*}
        &V(x(t+1)) 
        \notag
        \\
        &\le \sum_{{i \in \mathcal{V}\backslash \{\{k\}\cup \mathcal{N}_k\} }, i < k} \big( J_i^{\mr{tr},*}(t+1) + \sum_{j \in \mathcal{N}_i} V_{ij}^\mr{c}(z_\mr{c}^{k}(t + 1)) \big)
        \notag
        \\
        &\phantom{\le{}} + \sum_{{i \in \mathcal{V}\backslash \{\{k\}\cup \mathcal{N}_k\} }, i > k} 
        \big( \bar{J}_i^{\mr{tr}}(t+1) + \sum_{j \in \mathcal{N}_i} V_{ij}^\mr{c}(z_\mr{c}^{k}(t + 1)) \big)
        \notag
        \\
        &\phantom{\le{}} + \sum_{{i \in \mathcal{N}_k}, i < k} 
        \big( J_i^{\mr{tr},*}(t+1) 
        + \sum_{j \in \mathcal{N}_i} V_{ij}^\mr{c}(z_\mr{c}^{k}(t+1)) \big)
        \notag
        \\
        &\phantom{\le{}} + \sum_{{i \in \mathcal{N}_k}, i > k} \big( \bar{J}_i^{\mr{tr}}(t+1) 
        + \sum_{j \in \mathcal{N}_i} V_{ij}^\mr{c}(z_\mr{c}^{k}(t+1)) \big)
        \notag
        \\
        &\phantom{\le{}} + \bar{J}_k^{\mr{tr}}(t+1) 
        + \sum_{j \in \mathcal{N}_k} V_{kj}^\mr{c}(z_\mr{c}^{k}(t+1))
    \end{align*} 
    where again $\bar{J}_k^{\mr{tr}}(t+1) = \bar{J}_k^{\mr{tr,b}}(t+1)$ if $k \in \mathcal{V}_\mr{b}$ or $\bar{J}_k^{\mr{tr}}(t+1) = \bar{J}_k^{\mr{tr,a}}(t+1)$ if $k \in \mathcal{V}_\mr{a}$.

    Reordering this and highlighting $k-1$ instead of $k$ reveals
    \begin{align*}
        &V(x(t+1)) 
        \notag
        \\
        &\le 
        \sum_{\substack{{i \in \mathcal{V}\backslash \{\{k-1\}\cup \mathcal{N}_{k-1}\} } \\ i < k-1 }} \big( J_i^{\mr{tr},*}(t+1) + \sum_{j \in \mathcal{N}_i} V_{ij}^\mr{c}(z_\mr{c}^{k}(t + 1)) \big)
        \notag
        \\
        &\phantom{\le{}} 
        + \sum_{\substack{{i \in \mathcal{V}\backslash \{\{k-1\}\cup \mathcal{N}_{k-1}\} } \\ i > k-1 }}
        \big( \bar{J}_i^{\mr{tr}}(t+1) + \sum_{j \in \mathcal{N}_i} V_{ij}^\mr{c}(z_\mr{c}^{k}(t + 1)) \big)
        \notag
        \\
        &\phantom{\le{}} + \sum_{\substack{{i \in \mathcal{N}_{k-1}}\\ i < k-1}} 
        \big( J_i^{\mr{tr},*}(t+1) 
        + \sum_{j \in \mathcal{N}_i  \backslash \{k-1\}} V_{ij}^\mr{c}(z_\mr{c}^{k}(t+1)) \big)
        \notag
        \\
        &\phantom{\le{}} + \sum_{\substack{{i \in \mathcal{N}_{k-1}} \\ i > k-1}} 
        \big( \bar{J}_i^{\mr{tr}}(t+1) 
        + \sum_{j \in \mathcal{N}_i  \backslash \{k-1\}} V_{ij}^\mr{c}(z_\mr{c}^{k}(t+1)) \big)
        \notag
        \\
        &\phantom{\le{}} + {J}_{k-1}^{\mr{tr},*}(t+1) 
        + \sum_{j \in \mathcal{N}_{k-1}} V_{(k-1)j}^\mr{c}(z_\mr{c}^{k}(t+1))
        \\
        &\phantom{\le{}} 
        + \sum_{j \in \mathcal{N}_{k-1}} V_{j(k-1)}^\mr{c}(z_\mr{c}^{k}(t+1)).
    \end{align*} 
    Since again the first four appearing $V_{ij}^\mr{c}$ do not depend on the cooperation output of agent $k-1$, the bound is
    \begin{align*}
        &V(x(t+1)) 
        \notag
        \\
        &\le 
        \sum_{\substack{{i \in \mathcal{V}\backslash \{\{k-1\}\cup \mathcal{N}_{k-1}\} } \\ i < k-1 }} \big( J_i^{\mr{tr},*}(t+1) + \sum_{j \in \mathcal{N}_i} V_{ij}^\mr{c}(z_\mr{c}^{k-1}(t+1)) \big)
        \notag
        \\
        &\phantom{\le{}} 
        + \sum_{\substack{{i \in \mathcal{V}\backslash \{\{k-1\}\cup \mathcal{N}_{k-1}\} } \\ i > k-1 }}
        \big( \bar{J}_i^{\mr{tr}}(t+1) + \sum_{j \in \mathcal{N}_i} V_{ij}^\mr{c}(z_\mr{c}^{k-1}(t+1)) \big)
        \notag
        \\
        &\phantom{\le{}} + \sum_{\substack{{i \in \mathcal{N}_{k-1}}\\ i < k-1}} 
        \big( J_i^{\mr{tr},*}(t+1) 
        + \sum_{j \in \mathcal{N}_i  \backslash \{k-1\}} V_{ij}^\mr{c}(z_\mr{c}^{k-1}(t+1)) \big)
        \notag
        \\
        &\phantom{\le{}} + \sum_{\substack{{i \in \mathcal{N}_{k-1}} \\ i > k-1}} 
        \big( \bar{J}_i^{\mr{tr}}(t+1) 
        + \sum_{j \in \mathcal{N}_i  \backslash \{k-1\}} V_{ij}^\mr{c}(z_\mr{c}^{k-1}(t+1)) \big)
        \notag
        \\
        &\phantom{\le{}} + {J}_{k-1}^{\mr{tr},*}(t+1) 
        + \sum_{j \in \mathcal{N}_{k-1}} V_{(k-1)j}^\mr{c}(z_\mr{c}^{k}(t+1))
        \\
        &\phantom{\le{}} 
        + \sum_{j \in \mathcal{N}_{k-1}} V_{j(k-1)}^\mr{c}(z_\mr{c}^{k}(t+1)),
    \end{align*} 
    which is~\eqref{eq:candidate_upper_bound_step_k} with $k-1$ instead of $k$.
    Hence, by induction until $k=1$ and inserting~\eqref{eq:candidate_upper_bound_for_a} or~\eqref{eq:candidate_upper_bound_for_b} one last time, the candidates provide the following upper bound on the value function:
    $%\begin{equation*}
        V(x(t+1)) \le \sum_{i\in\mathcal{V}} \bar{J}_i^{\mr{tr}}(t+1) 
        + \sum_{j \in \mathcal{N}_i} V_{ij}^\mr{c}(z_\mr{c}^{1}(t+1))
    $ %\end{equation*}
    with $V_{ij}^\mr{c}(z_\mr{c}^{1}(t+1)) = V_{ij}^\mr{c}(y_{\mr{c}, i}^*(t), y_{\mr{c}, j}^*(t))$ for all $i,j \in \mathcal{V}$.
    
    Inserting $\bar{J}_i^{\mr{tr}}(t+1) = \bar{J}_i^{\mr{tr,a}}(t+1)$ if $i \in \mathcal{V}_\mr{a}$ or $\bar{J}_i^{\mr{tr}}(t+1) = \bar{J}_i^{\mr{tr,b}}(t+1)$ if $i \in \mathcal{V}_\mr{b}$ and the shorthands back in, yields the claimed upper bound
    \begin{align*}
        &V(x(t+1)) 
        \\
        &\le \sum_{i \in \mathcal{V}_\mr{a} } \big( \bar{J}_i^\mr{tr,a}(t+1) + \sum_{j \in \mathcal{N}_i} V_{ij}^\mr{c}(y_{\mr{c},i}^*(t), y_{\mr{c},j}^*(t)) \big)
        \\
        &\phantom{\le{}} +  \sum_{i \in \mathcal{V}_\mr{b} } \big( \bar{J}_i^\mr{tr,b}(t+1) + \sum_{j \in \mathcal{N}_i} V_{ij}^\mr{c}(y_{\mr{c},i}^*(t), y_{\mr{c},j}^*(t)) \big)
        \\
        &\le \sum_{i \in \mathcal{V}_\mr{a} } \big( {J}_i^{\mr{tr},*}(t)  + \sum_{j \in \mathcal{N}_i} V_{ij}^\mr{c}(y_{\mr{c},i}^*(t), y_{\mr{c},j}^*(t)) \big)
        \\
        &\phantom{\le{}} + \sum_{i \in \mathcal{V}_\mr{b}} \big(  \hat{J}_i^\mr{tr,b}(t+1) + \sum_{j \in \mathcal{N}_i} V_{ij}^\mr{c}(y_{\mr{c},i}^*(t), y_{\mr{c},j}^*(t)) \big)
        \\
        &\phantom{\le{}} - \sum_{i \in \mathcal{V}_\mr{b} } \kappa_i \Vert T(y_{\mr{c},i}^*(t)) - y^*_{\mr{c},i}(t) \Vert^2
        \\
        &\phantom{\le{}} - \sum_{i \in \mathcal{V}_\mr{a} } \ell_i(x_i(t), u_i^*(0 \vert t), r_{\mr{c},i}^*(t))
        \\
        &= V(x(t))  - \sum_{i \in \mathcal{V}_\mr{a} } \ell_i(x_i(t), u_i^*(0 \vert t), r_{\mr{c},i}^*(t))
        \\
        &\phantom{={}} + \sum_{i \in \mathcal{V}_\mr{b} } \hat{J}_i^\mr{tr,b}(t+1) - {J}_i^{\mr{tr},*}(t) 
        - \kappa_i \Vert T(y_{\mr{c},i}^*(t)) - y^*_{\mr{c},i}(t) \Vert^2.
    \end{align*}
\end{proof}

%-------------------------------------------------------------------------------
\subsection{Consensus: Numerical example}\label{ssec:consensus_appendix}
This section supplements Section~\ref{ssec:consensus} with a numerical example.
We consider a multiagent system consisting of double integrators 
\begin{align*}
    x_i(t+1) = \begin{bsmallmatrix} 1 & 0 & 1 & 0 \\ 0 & 1 & 0 & 1 \\ 0 & 0 & 1 & 0 \\ 0 & 0 & 0 & 1 \end{bsmallmatrix} x_i(t) + \begin{bsmallmatrix} 0 & 0 \\ 0 & 0 \\ 1 & 0\\ 0 & 1 \end{bsmallmatrix} u_i(t), \quad
    y_i(t) = \begin{bsmallmatrix}
        x_{i,1}(t) \\ x_{i,2}(t)
    \end{bsmallmatrix}.
\end{align*}
The initial conditions are set to 
$x_1^\mr{ic} = \begin{bsmallmatrix} -1 & 4 & 0 & 0 \end{bsmallmatrix}^\top$,
$x_2^\mr{ic} = \begin{bsmallmatrix} 2 & 1.8 & 0 & 0 \end{bsmallmatrix}^\top$,
$x_3^\mr{ic} = \begin{bsmallmatrix} 3 & -1.5 & 0 & 0 \end{bsmallmatrix}^\top$,
$x_4^\mr{ic} = \begin{bsmallmatrix} -2 & 0 & 0 & 0 \end{bsmallmatrix}^\top$,
and 
$x_5^\mr{ic} = \begin{bsmallmatrix} 0 & -2 & 0 & 0 \end{bsmallmatrix}^\top$.
The constraint sets are given by $\mathbb{X}_1 = \tilde{\mathbb{X}}_\mr{a} \cap \bar{\mathbb{X}}$, as well as $\mathbb{X}_2 = \mathbb{X}_3 = \tilde{\mathbb{X}}_\mr{b} \cap \bar{\mathbb{X}}$, and $\mathbb{X}_4 = \mathbb{X}_5 = \tilde{\mathbb{X}}_\mr{c} \cap \bar{\mathbb{X}}$ where 
\begin{align*}
    \tilde{\mathbb{X}}_\mr{a} &= \{ x_i \mid \begin{bsmallmatrix} x_{i,1} \\ x_{i,2} \end{bsmallmatrix} \in \mathrm{co}(\begin{bsmallmatrix} 1.1 \\ -2.1 \end{bsmallmatrix}, \begin{bsmallmatrix} 1.1 \\ 4.1 \end{bsmallmatrix}, \begin{bsmallmatrix} -1.1 \\ 4.1 \end{bsmallmatrix}, \begin{bsmallmatrix} -1.1 \\ -2.1 \end{bsmallmatrix}) \},
    \\
    \tilde{\mathbb{X}}_\mr{b} &= \{ x_i \mid \begin{bsmallmatrix} x_{i,1} \\ x_{i,2} \end{bsmallmatrix} \in \mathrm{co}(\begin{bsmallmatrix} 4.1 \\ -2.1 \end{bsmallmatrix}, \begin{bsmallmatrix} 4.1 \\ 2.1 \end{bsmallmatrix}, \begin{bsmallmatrix} -1.1 \\ 2.1 \end{bsmallmatrix}, \begin{bsmallmatrix} -1.1 \\ -2.1 \end{bsmallmatrix}) \},
    \\
    \tilde{\mathbb{X}}_\mr{c} &= \{ x_i \mid \begin{bsmallmatrix} x_{i,1} \\ x_{i,2} \end{bsmallmatrix} \in \mathrm{co}(\begin{bsmallmatrix} 3.1 \\ -0.1 \end{bsmallmatrix}, \begin{bsmallmatrix} -0.1 \\ 3.1 \end{bsmallmatrix}, \begin{bsmallmatrix} -3.1 \\ -0.1 \end{bsmallmatrix}, \begin{bsmallmatrix} -0.1 \\ -3.1 \end{bsmallmatrix}) \},
    \\
    \bar{\mathbb{X}} &= \{ x_i \mid \Vert \begin{bsmallmatrix} x_{i,3} \\ x_{i,4} \end{bsmallmatrix} \Vert_\infty \le 0.25 \},
\end{align*}
as well as
$ %\begin{equation*}
    \mathbb{U}_i = \{ u_i \in \mathbb{R}^2 \mid \Vert u_i \Vert_\infty \le 0.25 \}
    % \quad \forall i \in \mathcal{V}.
$ %\end{equation*}
for all $i \in \mathcal{V}$.
The output cooperation sets $\mathcal{Y}_i$ are tightened compared to $\mathbb{X}_i$ to satisfy Assumption~\ref{asm:reference} and are depicted in Figure~\ref{fig:consensus_positions}.

In the beginning, the multiagent system comprises $m = 4$ systems that are connected according to the graph on the left-hand side in Figure~\ref{fig:graphs}.
After $t = 19$ time steps and agent $4$ completed its optimization, an additional system joins the multiagent system and the graph changes to the one on the right-hand side of Figure~\ref{fig:graphs}, i.e., agent $5$ joins the system.
The additional agent $5$ initializes at $t = 19$ with ${y}_{\mr{c},5}^*(19) = \begin{bsmallmatrix} x_{5,1}^\mr{ic} & x_{5,2}^\mr{ic} \end{bsmallmatrix}^\top$ and sends this to its neighbors.
Thus, all agents consider the changed system in Step 1) of Algorithm~\ref{alg:sequential_MPC} at $t=20$.
The simulation results are depicted in Figures~\ref{fig:consensus_time_evolution} and~\ref{fig:consensus_positions}.
In addition, each agent is only allowed to track cooperation outputs in their output cooperation set.
\begin{figure}[t]
    \centerline{\input{graphs.tex}}
    \caption{Communication topologies of the multiagent system in Section~\ref{ssec:consensus_appendix}. After some time, an agent joins the system and the topology changes from the one on the left-hand side to the other one.}
    \label{fig:graphs}
\end{figure}
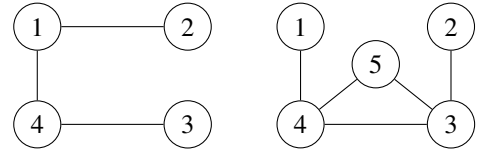

\begin{figure}[t]
    \setlength\axisheight{0.64\linewidth}
    \setlength\axiswidth{0.95\linewidth}
    \centering
    \input{consensus_x1_manual.tex}
    \input{consensus_x2_manual.tex}
    \caption{Evolution over time of the first two states (the outputs) of the multiagent system in Section~\ref{ssec:consensus_appendix}. The eventual consensus value changes after an additional agent joins the system.}
    \label{fig:consensus_time_evolution}
\end{figure}
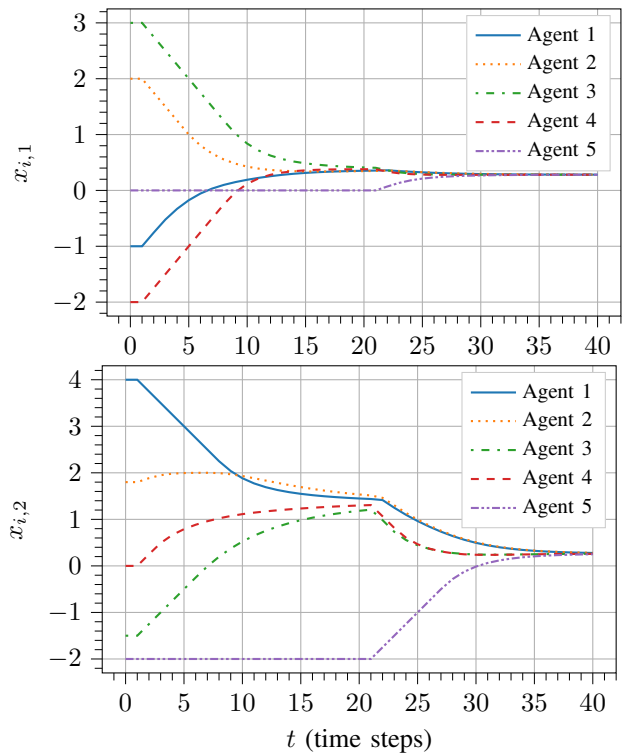

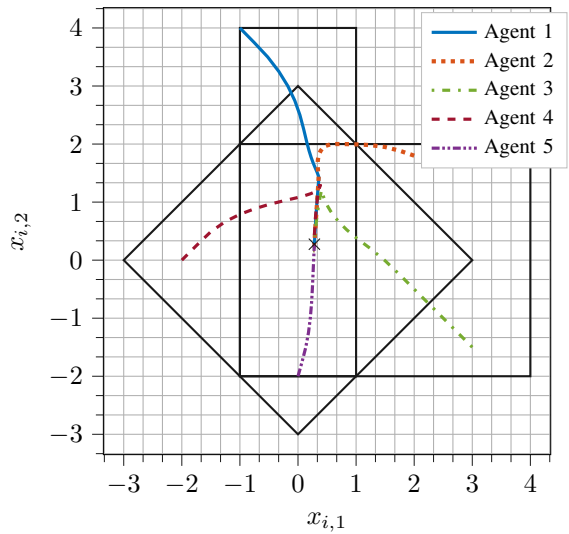
\begin{figure}[b]
    \setlength\axisheight{0.85\linewidth}
    \setlength\axiswidth{0.85\linewidth}
    \centering
    \input{consensus_pos_manual.tex}
    \caption{Evolution of the first two states (the outputs) of the multiagent system in Section~\ref{ssec:consensus_appendix}. The eventual consensus value (black X) changes after an additional agent joins the system.
    In addition, the output constraint sets are depicted.}
    \label{fig:consensus_positions}
\end{figure}%

%-------------------------------------------------------------------------------
\subsection{Formation control with perturbed initial conditions}\label{ssec:formation_appendix}
We revisit the simulation example in Section~\ref{ssec:formation} to highlight that the deliberate choice of initial conditions led to difficulties for the solver.
Here, we choose the slightly perturbed initial conditions
\begin{align*}
    x_1^\mr{ic} &= \begin{bsmallmatrix} 10^{-5} & 0 & 1 & 0 & 0 & 0 & 0 & 0 & 0 & 0 \end{bsmallmatrix}^\top,
    \\
    x_2^\mr{ic} &= \begin{bsmallmatrix} -10^{-5} & 10^{-5} & 2 & 0 & 0 & 0 & 0 & 0 & 0 & 0 \end{bsmallmatrix}^\top,
    \\
    x_3^\mr{ic} &= \begin{bsmallmatrix} -10^{-5} & -10^{-5} & 3 & 0 & 0 & 0 & 0 & 0 & 0 & 0 \end{bsmallmatrix}^\top.
\end{align*}
The simulation results are depicted in Figures~\ref{fig:quadcopter_altitude2}--\ref{fig:quadcopter_outputs2}.
Note that the final configuration is achieved much quicker than in Section~\ref{ssec:formation} (cf. Figure~\ref{fig:quadcopter_altitude}).
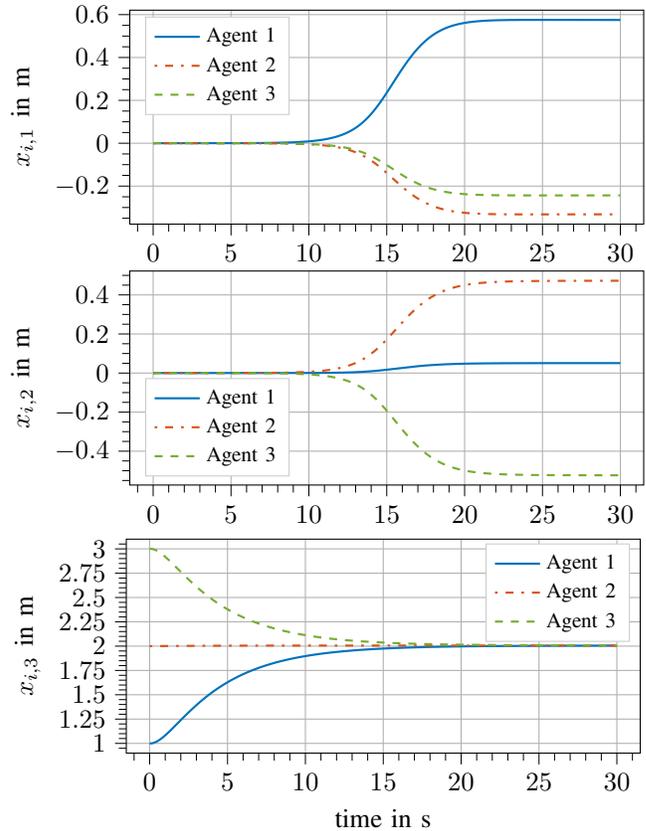
\begin{figure}[t]
    \setlength\axisheight{0.5\linewidth}
    \setlength\axiswidth{0.95\linewidth}
    \centerline{\input{alt_x1_manual.tex}}
    \centerline{\input{alt_x2_manual.tex}}
    \centerline{\input{alt_x3_manual.tex}}
    \caption{Evolution over time of the first three states (the outputs) of the quadcopters in Section~\ref{ssec:formation_appendix} with the perturbed choice of initial conditions.}
    \label{fig:quadcopter_altitude2}
\end{figure}
\begin{figure}[t]
    \resizebox{0.85\linewidth}{!}
    {
    \centerline{\input{alt_outputs.pgf}}
    }
    \caption{Evolution of the first three states of the quadcopters in Section~\ref{ssec:formation_appendix} into an equilateral triangle. The initial condition of each trajectory is marked by a point and the final value by an X.}
    \label{fig:quadcopter_outputs2}
\end{figure}
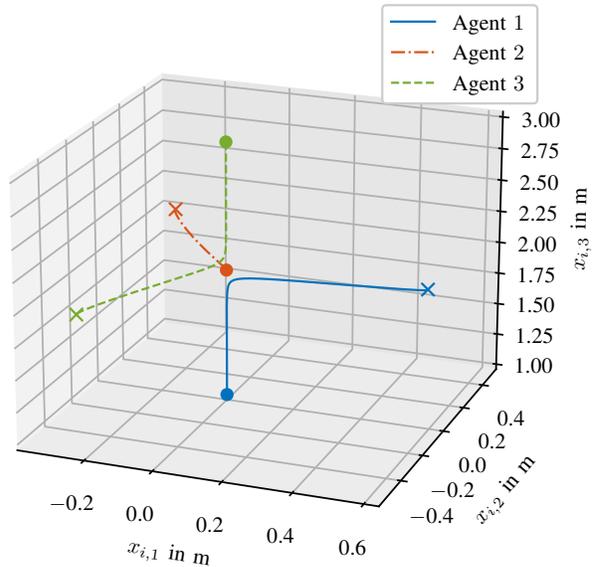
%\appendices{}

\end{document}

%% file: formation_x1_manual.tex
% Change size of legend entries to footnotesize.
% Remove x-axis label.
% Moved legend.
\begin{tikzpicture}

\definecolor{grey}{RGB}{176,176,176}
\definecolor{orange}{RGB}{217,83,25}
\definecolor{green}{RGB}{119,172,48}
\definecolor{lightgrey}{RGB}{204,204,204}
\definecolor{blue}{RGB}{0,114,189}

\begin{axis}[
height=\axisheight,
legend cell align={left},
legend style={
  fill opacity=0.8,
  draw opacity=1,
  text opacity=1,
  at={(0.03,0.55)},
  anchor=north west,
  draw=lightgrey
},
minor tick num=2,
minor xtick={-2,2,4,6,8,12,14,16,18,22,24,26,28,32,34,36,38,42,44,46,48,52,54,56,58,62},
minor ytick={-0.55,-0.5,-0.45,-0.35,-0.3,-0.25,-0.15,-0.1,-0.0500000000000003,0.0499999999999998,0.0999999999999998,0.15,0.25,0.3,0.35},
tick align=outside,
tick pos=left,
width=\axiswidth,
x grid style={grey},
xmajorgrids,
xmin=-3, xmax=63,
xtick distance=10,
xtick style={color=black},
xtick={-10,0,10,20,30,40,50,60,70},
% xticklabels ={,,},  % Hide x tick labels.
y grid style={grey},
ylabel={\(\displaystyle x_{i,1}\) in m},
ymajorgrids,
ymin=-0.608726101007325, ymax=0.446980528185849,
ytick distance=0.2,
ytick style={color=black},
ytick={-0.8,-0.6,-0.4,-0.2,0,0.2,0.4,0.6}
]
\addplot [thick, blue]
table {%
0 0
34.5999984741211 -0.000362753868103027
36.2000007629395 -0.00112855434417725
37.2000007629395 -0.00229370594024658
37.9000015258789 -0.0037682056427002
38.5 -0.00576627254486084
39 -0.00821828842163086
39.4000015258789 -0.0109097957611084
39.7000007629395 -0.0134897232055664
40 -0.016675591468811
40.2999992370605 -0.0206058025360107
40.5999984741211 -0.0254473686218262
40.7999992370605 -0.0292781591415405
41 -0.0336686372756958
41.2000007629395 -0.0386918783187866
41.4000015258789 -0.0444269180297852
41.5999984741211 -0.0509566068649292
41.7999992370605 -0.0583653450012207
42 -0.0667357444763184
42.2000007629395 -0.0761435031890869
42.4000015258789 -0.0866519212722778
42.5999984741211 -0.0983049869537354
42.7999992370605 -0.111121535301208
43 -0.125089645385742
43.2000007629395 -0.14016318321228
43.4000015258789 -0.156261444091797
43.5999984741211 -0.173271417617798
43.7999992370605 -0.191051959991455
44 -0.209441542625427
44.2999992370605 -0.237783432006836
45 -0.304352760314941
45.2999992370605 -0.331770420074463
45.5 -0.349349975585938
45.7000007629395 -0.366255283355713
45.9000015258789 -0.382408738136292
46.0999984741211 -0.397751331329346
46.2999992370605 -0.412241220474243
46.5 -0.425852417945862
46.7000007629395 -0.438573479652405
46.9000015258789 -0.450405359268188
47.0999984741211 -0.461360096931458
47.2999992370605 -0.471458911895752
47.5 -0.480730414390564
47.7000007629395 -0.489209294319153
47.9000015258789 -0.49693489074707
48.0999984741211 -0.503949403762817
48.2999992370605 -0.510297536849976
48.5 -0.516024351119995
48.7000007629395 -0.521175384521484
48.9000015258789 -0.525795817375183
49.2000007629395 -0.531826496124268
49.5 -0.536901473999023
49.7999992370605 -0.541150689125061
50.0999984741211 -0.544691801071167
50.5 -0.548492193222046
50.9000015258789 -0.551435947418213
51.2999992370605 -0.553702116012573
51.7999992370605 -0.5558021068573
52.4000015258789 -0.557536482810974
53.0999984741211 -0.55882465839386
54.0999984741211 -0.559837341308594
55.5 -0.560437560081482
58.2000007629395 -0.560712575912476
60 -0.560739517211914
};
\addlegendentry{\footnotesize Agent 1}
\addplot [thick, orange, dash pattern=on 1pt off 3pt on 3pt off 3pt]
table {%
0 0
36.5999984741211 0.000352263450622559
38.2000007629395 0.00114238262176514
39.0999984741211 0.00221192836761475
39.7999992370605 0.00369560718536377
40.4000015258789 0.00573444366455078
40.9000015258789 0.00826334953308105
41.2999992370605 0.0110585689544678
41.5999984741211 0.0137479305267334
41.9000015258789 0.0170725584030151
42.2000007629395 0.021167516708374
42.5 0.0261855125427246
42.7000007629395 0.0301250219345093
42.9000015258789 0.0345984697341919
43.0999984741211 0.0396550893783569
43.2999992370605 0.0453402996063232
43.5 0.0516939163208008
43.7000007629395 0.058746337890625
43.9000015258789 0.066516637802124
44.0999984741211 0.0750095844268799
44.2999992370605 0.084214448928833
44.5 0.0941028594970703
44.7000007629395 0.104629755020142
44.9000015258789 0.115733623504639
45.2000007629395 0.133301019668579
45.5 0.151687622070312
46.0999984741211 0.189540982246399
46.5 0.21447491645813
46.7999992370605 0.232502818107605
47.0999984741211 0.249685764312744
47.4000015258789 0.265844106674194
47.5999984741211 0.275984168052673
47.7999992370605 0.285590171813965
48 0.294649004936218
48.2000007629395 0.303155541419983
48.4000015258789 0.311112761497498
48.5999984741211 0.318529486656189
48.7999992370605 0.325420022010803
49 0.331802606582642
49.2999992370605 0.340471744537354
49.5999984741211 0.348128318786621
49.9000015258789 0.354859232902527
50.2000007629395 0.360753536224365
50.5 0.365898609161377
50.7999992370605 0.370377421379089
51.2000007629395 0.375446677207947
51.5999984741211 0.37963604927063
52 0.38309121131897
52.5 0.386566400527954
53 0.389291048049927
53.5999984741211 0.391795635223389
54.2999992370605 0.393930792808533
55.0999984741211 0.395634174346924
56.0999984741211 0.397026538848877
57.4000015258789 0.398088574409485
59.2999992370605 0.398843765258789
60 0.398993849754333
};
\addlegendentry{\footnotesize Agent 2}
\addplot [thick, green, dashed]
table {%
0 0
34.7999992370605 0.000342369079589844
36.5 0.00112831592559814
37.5 0.00227499008178711
38.2000007629395 0.00371646881103516
38.7999992370605 0.00565874576568604
39.2999992370605 0.00803053379058838
39.7000007629395 0.0106216669082642
40.0999984741211 0.0140399932861328
40.4000015258789 0.0172960758209229
40.7000007629395 0.0212876796722412
41 0.0261647701263428
41.2000007629395 0.0299899578094482
41.4000015258789 0.0343348979949951
41.5999984741211 0.0392513275146484
41.7999992370605 0.0447882413864136
42 0.0509874820709229
42.2000007629395 0.0578786134719849
42.4000015258789 0.0654735565185547
42.5999984741211 0.0737608671188354
42.7999992370605 0.0827004909515381
43 0.0922207832336426
43.2999992370605 0.107355117797852
44.2000007629395 0.154060006141663
44.4000015258789 0.163610577583313
44.5999984741211 0.172526121139526
44.7999992370605 0.18069314956665
45 0.188026905059814
45.2000007629395 0.194472193717957
45.4000015258789 0.200002670288086
45.5999984741211 0.204617857933044
45.7999992370605 0.208340167999268
46 0.21121084690094
46.2000007629395 0.213285446166992
46.4000015258789 0.214630246162415
46.5999984741211 0.215318083763123
46.7999992370605 0.215425133705139
47.0999984741211 0.214664220809937
47.4000015258789 0.213018536567688
47.7999992370605 0.209839701652527
48.2999992370605 0.204918503761292
49.9000015258789 0.188392877578735
50.5 0.183221459388733
51.0999984741211 0.178840756416321
51.7000007629395 0.175219774246216
52.2999992370605 0.172278761863708
53 0.16957676410675
53.7999992370605 0.16726291179657
54.7000007629395 0.165408372879028
55.7999992370605 0.163901686668396
57.2000007629395 0.162761330604553
59.0999984741211 0.161986827850342
60 0.161790490150452
};
\addlegendentry{\footnotesize Agent 3}
\end{axis}

\end{tikzpicture}

%% file: formation_x2_manual.tex
% Change size of legend entries to footnotesize.
% Remove x axis label.
% Moved legend.
\begin{tikzpicture}

\definecolor{grey}{RGB}{176,176,176}
\definecolor{orange}{RGB}{217,83,25}
\definecolor{green}{RGB}{119,172,48}
\definecolor{lightgrey}{RGB}{204,204,204}
\definecolor{blue}{RGB}{0,114,189}

\begin{axis}[
height=\axisheight,
legend cell align={left},
legend style={
  fill opacity=0.8,
  draw opacity=1,
  text opacity=1,
  at={(0.03,0.5)},
  anchor=north west,
  draw=lightgrey
},
minor tick num=2,
minor xtick={-2,2,4,6,8,12,14,16,18,22,24,26,28,32,34,36,38,42,44,46,48,52,54,56,58,62},
minor ytick={-0.55,-0.5,-0.45,-0.35,-0.3,-0.25,-0.15,-0.1,-0.0500000000000003,0.0499999999999998,0.0999999999999998,0.15,0.25,0.3,0.35,0.45,0.5},
tick align=outside,
tick pos=left,
width=\axiswidth,
x grid style={grey},
xmajorgrids,
xmin=-3, xmax=63,
xtick distance=10,
xtick style={color=black},
xtick={-10,0,10,20,30,40,50,60,70},
% xticklabels ={,,},  % Hide x tick labels.
y grid style={grey},
ylabel={\(\displaystyle x_{i,2}\) in m},
ymajorgrids,
ymin=-0.606642183368995, ymax=0.465399539329772,
ytick distance=0.2,
ytick style={color=black},
ytick={-0.8,-0.6,-0.4,-0.2,0,0.2,0.4,0.6}
]
\addplot [thick, blue]
table {%
0 0
35.0999984741211 0.000370979309082031
36.7000007629395 0.0011674165725708
37.7000007629395 0.00238955020904541
38.4000015258789 0.00394439697265625
39 0.00605952739715576
39.5 0.00866222381591797
39.9000015258789 0.0115227699279785
40.2000007629395 0.0142650604248047
40.5 0.0176477432250977
40.7999992370605 0.0218095779418945
41.0999984741211 0.0269111394882202
41.2999992370605 0.030921459197998
41.5 0.0354832410812378
41.7000007629395 0.0406506061553955
41.9000015258789 0.0464733839035034
42.0999984741211 0.052992582321167
42.2999992370605 0.0602344274520874
42.5 0.068204402923584
42.7000007629395 0.0768810510635376
42.9000015258789 0.0862109661102295
43.2000007629395 0.101227998733521
43.5999984741211 0.122436761856079
44 0.143712520599365
44.2999992370605 0.158799529075623
44.5 0.168113708496094
44.7000007629395 0.176663398742676
44.9000015258789 0.184336304664612
45.0999984741211 0.19105339050293
45.2999992370605 0.196768522262573
45.5 0.201466917991638
45.7000007629395 0.205162048339844
45.9000015258789 0.20789110660553
46.0999984741211 0.209711074829102
46.2999992370605 0.210693120956421
46.5 0.210918664932251
46.7000007629395 0.210474491119385
46.9000015258789 0.209449768066406
47.2000007629395 0.207016348838806
47.5 0.203756093978882
47.9000015258789 0.19855535030365
48.5999984741211 0.188428401947021
49.2999992370605 0.178475975990295
49.7999992370605 0.172022581100464
50.2999992370605 0.166292905807495
50.7999992370605 0.161324381828308
51.2999992370605 0.157090187072754
51.7999992370605 0.153528690338135
52.4000015258789 0.150033354759216
53 0.147256135940552
53.7000007629395 0.144750475883484
54.5 0.142638921737671
55.5 0.140819787979126
56.7000007629395 0.139445304870605
58.2999992370605 0.138431549072266
60 0.137903213500977
};
%\addlegendentry{\footnotesize Agent 1}
\addplot [thick, orange, dash pattern=on 1pt off 3pt on 3pt off 3pt]
table {%
0 0
35.4000015258789 0.000355720520019531
37.0999984741211 0.00115656852722168
38.0999984741211 0.00231325626373291
38.7999992370605 0.00375699996948242
39.4000015258789 0.0056917667388916
39.9000015258789 0.00804281234741211
40.2999992370605 0.0106004476547241
40.7000007629395 0.0139611959457397
41 0.0171509981155396
41.2999992370605 0.0210481882095337
41.5999984741211 0.0257930755615234
41.9000015258789 0.0315409898757935
42.0999984741211 0.0360110998153687
42.2999992370605 0.0410481691360474
42.5 0.0466976165771484
42.7000007629395 0.05299973487854
42.9000015258789 0.0599862337112427
43.0999984741211 0.0676783323287964
43.2999992370605 0.0760840177536011
43.5 0.0851967334747314
43.7000007629395 0.0949946641921997
43.9000015258789 0.105440616607666
44.0999984741211 0.116482973098755
44.4000015258789 0.134020209312439
44.7000007629395 0.152488589286804
45.0999984741211 0.178036689758301
45.9000015258789 0.229444742202759
46.2000007629395 0.247909188270569
46.5 0.265539407730103
46.7999992370605 0.282150983810425
47 0.29259204864502
47.2000007629395 0.302493691444397
47.4000015258789 0.31183910369873
47.5999984741211 0.320619583129883
47.7999992370605 0.328834891319275
48 0.336490988731384
48.2000007629395 0.343600034713745
48.4000015258789 0.350178718566895
48.7000007629395 0.359098196029663
49 0.366952180862427
49.2999992370605 0.373829483985901
49.5999984741211 0.379822731018066
49.9000015258789 0.385024309158325
50.2000007629395 0.389522790908813
50.5 0.393402338027954
50.9000015258789 0.397744417190552
51.2999992370605 0.401286125183105
51.7000007629395 0.404168725013733
52.2000007629395 0.407025218009949
52.7999992370605 0.409604072570801
53.5 0.411754965782166
54.2999992370605 0.413429498672485
55.2999992370605 0.414761543273926
56.7000007629395 0.415801644325256
58.7000007629395 0.41647207736969
60 0.416670322418213
};
%\addlegendentry{\footnotesize Agent 2}
\addplot [thick, green, dashed]
table {%
0 0
34.2000007629395 -0.000362396240234375
35.7999992370605 -0.00112295150756836
36.7999992370605 -0.00227677822113037
37.5 -0.00373375415802002
38.0999984741211 -0.00570487976074219
38.5999984741211 -0.00812125205993652
39 -0.0107712745666504
39.4000015258789 -0.0142829418182373
39.7000007629395 -0.0176455974578857
40 -0.0217932462692261
40.2999992370605 -0.0269033908843994
40.5 -0.0309485197067261
40.7000007629395 -0.0355877876281738
40.9000015258789 -0.0409011840820312
41.0999984741211 -0.0469760894775391
41.2999992370605 -0.0539063215255737
41.5 -0.0617897510528564
41.7000007629395 -0.0707252025604248
41.9000015258789 -0.080808162689209
42.0999984741211 -0.0921239852905273
42.2999992370605 -0.104740381240845
42.5 -0.118698596954346
42.7000007629395 -0.134005665779114
42.9000015258789 -0.150626063346863
43.0999984741211 -0.16847836971283
43.2999992370605 -0.187434315681458
43.5 -0.207322359085083
43.7000007629395 -0.227935671806335
44 -0.259703636169434
44.5 -0.3128741979599
44.7000007629395 -0.333544135093689
44.9000015258789 -0.353570222854614
45.0999984741211 -0.372790336608887
45.2999992370605 -0.391073822975159
45.5 -0.408321619033813
45.7000007629395 -0.42446506023407
45.9000015258789 -0.439463496208191
46.0999984741211 -0.453300833702087
46.2999992370605 -0.465982794761658
46.5 -0.477532386779785
46.7000007629395 -0.487987637519836
46.9000015258789 -0.497397184371948
47.0999984741211 -0.505818247795105
47.2999992370605 -0.513313293457031
47.5 -0.519948482513428
47.7000007629395 -0.525791168212891
47.9000015258789 -0.530908465385437
48.0999984741211 -0.535366415977478
48.2999992370605 -0.539228558540344
48.5 -0.542555451393127
48.7999992370605 -0.546665906906128
49.0999984741211 -0.549875736236572
49.4000015258789 -0.552340745925903
49.7999992370605 -0.554700255393982
50.2000007629395 -0.556243300437927
50.7000007629395 -0.557346343994141
51.2999992370605 -0.55786657333374
52.2000007629395 -0.557750225067139
54.2999992370605 -0.556351661682129
56.5999984741211 -0.555209755897522
59.5 -0.554603576660156
60 -0.554553031921387
};
%\addlegendentry{\footnotesize Agent 3}
\end{axis}

\end{tikzpicture}

%% file: formation_x3_manual.tex
% Change size of legend entries to footnotesize.
\begin{tikzpicture}

\definecolor{grey}{RGB}{176,176,176}
\definecolor{orange}{RGB}{217,83,25}
\definecolor{green}{RGB}{119,172,48}
\definecolor{lightgrey}{RGB}{204,204,204}
\definecolor{blue}{RGB}{0,114,189}

\begin{axis}[
height=\axisheight,
legend cell align={left},
legend style={
    fill opacity=0.8, 
    draw opacity=1, 
    text opacity=1, 
    draw=lightgrey, 
    at={(0.70,0.99)},
    anchor=north west
    },
minor tick num=3,
minor xtick={-2,2,4,6,8,12,14,16,18,22,24,26,28,32,34,36,38,42,44,46,48,52,54,56,58,62},
%minor ytick={0.95,1.05,1.1,1.15,1.2,1.3,1.35,1.4,1.45,1.55,1.6,1.65,1.7,1.8,1.85,1.9,1.95,2.05,2.1,2.15,2.2,2.3,2.35,2.4,2.45,2.55,2.6,2.65,2.7,2.8,2.85,2.9,2.95,3.05},
tick align=outside,
tick pos=left,
width=\axiswidth,
x grid style={grey},
xlabel={time in s},
xmajorgrids,
xmin=-3, xmax=63,
xtick distance=10,
xtick style={color=black},
xtick={-10,0,10,20,30,40,50,60,70},
y grid style={grey},
ylabel={\(\displaystyle x_{i,3}\) in m},
ymajorgrids,
ymin=0.9, ymax=3.1,
ytick distance=0.2,
ytick style={color=black},
ytick={1,1.5,2,2.5,3}
]
\addplot [thick, blue]
table {%
0 1
0.100000023841858 1
0.200000047683716 1.00266444683075
0.299999952316284 1.00753903388977
0.399999976158142 1.01433002948761
0.5 1.02278172969818
0.600000023841858 1.03266632556915
0.700000047683716 1.04378092288971
0.799999952316284 1.05594432353973
1 1.0827910900116
1.20000004768372 1.11212027072906
1.39999997615814 1.14307522773743
1.79999995231628 1.2073210477829
2.29999995231628 1.28760290145874
2.59999990463257 1.33414173126221
2.90000009536743 1.37882161140442
3.20000004768372 1.42136013507843
3.40000009536743 1.44845283031464
3.59999990463257 1.47450757026672
3.79999995231628 1.4995197057724
4 1.52349555492401
4.30000019073486 1.55754995346069
4.59999990463257 1.58938109874725
4.90000009536743 1.61908292770386
5.19999980926514 1.64676141738892
5.5 1.67252790927887
5.80000019073486 1.6964955329895
6.09999990463257 1.7187762260437
6.40000009536743 1.73947882652283
6.69999980926514 1.75870788097382
7 1.77656328678131
7.30000019073486 1.79313921928406
7.69999980926514 1.81340336799622
8.10000038146973 1.83174669742584
8.5 1.84834825992584
8.89999961853027 1.86337172985077
9.30000019073486 1.87696599960327
9.69999980926514 1.88926613330841
10.1999998092651 1.90300738811493
10.6999998092651 1.91513228416443
11.1999998092651 1.92583060264587
11.8000001907349 1.93702006340027
12.3999996185303 1.946648478508
13.1000003814697 1.9561972618103
13.8000001907349 1.96421051025391
14.6000003814697 1.97180306911469
15.5 1.97870981693268
16.5 1.98476183414459
17.6000003814697 1.98988234996796
18.8999996185303 1.99436843395233
20.5 1.99821698665619
22.3999996185303 2.00117325782776
24.8999996185303 2.00343060493469
28.3999996185303 2.00494456291199
34.4000015258789 2.00578355789185
51.0999984741211 2.00602006912231
60 2.00602340698242
};
%\addlegendentry{\footnotesize Agent 1}
\addplot [thick, orange, dash pattern=on 1pt off 3pt on 3pt off 3pt]
table {%
0 2
0.600000023841858 2.00071001052856
3.5 2.00394320487976
6.30000019073486 2.00544548034668
10.8000001907349 2.00609493255615
29.2000007629395 2.0060350894928
60 2.006023645401
};
%\addlegendentry{\footnotesize Agent 2}
\addplot [thick, green, dashed]
table {%
0 3
0.100000023841858 3
0.200000047683716 2.99748969078064
0.299999952316284 2.99277591705322
0.399999976158142 2.98614501953125
0.5 2.97785258293152
0.600000023841858 2.96812701225281
0.700000047683716 2.95717120170593
0.799999952316284 2.94516658782959
1 2.9186372756958
1.20000004768372 2.88962531089783
1.39999997615814 2.8589870929718
1.79999995231628 2.79536986351013
2.29999995231628 2.71585536003113
2.59999990463257 2.66976261138916
2.90000009536743 2.62551522254944
3.20000004768372 2.58339405059814
3.5 2.54354524612427
3.70000004768372 2.51827025413513
3.90000009536743 2.49402832984924
4.19999980926514 2.45957684516907
4.5 2.42735886573792
4.80000019073486 2.39728617668152
5.09999990463257 2.36925625801086
5.40000009536743 2.34315967559814
5.69999980926514 2.31888389587402
6 2.29631686210632
6.30000019073486 2.27534961700439
6.59999990463257 2.25587630271912
6.90000009536743 2.23779606819153
7.19999980926514 2.22101354598999
7.5 2.2054386138916
7.90000009536743 2.18640446662903
8.30000019073486 2.16918087005615
8.69999980926514 2.15359783172607
9.10000038146973 2.13950037956238
9.5 2.12674784660339
10 2.11250519752502
10.5 2.09994173049927
11 2.08886003494263
11.6000003814697 2.07727360725403
12.1999998092651 2.06730771064758
12.8999996185303 2.05742812156677
13.6000003814697 2.04914116859436
14.3999996185303 2.04129338264465
15.3000001907349 2.03415822982788
16.2999992370605 2.02790999412537
17.3999996185303 2.02262711524963
18.7000007629395 2.01800227165222
20.2000007629395 2.01424241065979
22.1000003814697 2.01112365722656
24.5 2.00881505012512
27.8999996185303 2.00721216201782
33.5 2.00631499290466
47.5999984741211 2.00603222846985
60 2.00602412223816
};
%\addlegendentry{\footnotesize Agent 3}
\end{axis}

\end{tikzpicture}

%% file: formation_outputs_manual.pgf
\begingroup%
\makeatletter%
\begin{pgfpicture}%
\pgfpathrectangle{\pgfpointorigin}{\pgfqpoint{4.169636in}{3.896000in}}%
\pgfusepath{use as bounding box, clip}%
\begin{pgfscope}%
\pgfsetbuttcap%
\pgfsetmiterjoin%
\definecolor{currentfill}{rgb}{1.000000,1.000000,1.000000}%
\pgfsetfillcolor{currentfill}%
\pgfsetlinewidth{0.000000pt}%
\definecolor{currentstroke}{rgb}{1.000000,1.000000,1.000000}%
\pgfsetstrokecolor{currentstroke}%
\pgfsetdash{}{0pt}%
\pgfpathmoveto{\pgfqpoint{0.000000in}{0.000000in}}%
\pgfpathlineto{\pgfqpoint{4.169636in}{0.000000in}}%
\pgfpathlineto{\pgfqpoint{4.169636in}{3.896000in}}%
\pgfpathlineto{\pgfqpoint{0.000000in}{3.896000in}}%
\pgfpathlineto{\pgfqpoint{0.000000in}{0.000000in}}%
\pgfpathclose%
\pgfusepath{fill}%
\end{pgfscope}%
\begin{pgfscope}%
\pgfsetbuttcap%
\pgfsetmiterjoin%
\definecolor{currentfill}{rgb}{1.000000,1.000000,1.000000}%
\pgfsetfillcolor{currentfill}%
\pgfsetlinewidth{0.000000pt}%
\definecolor{currentstroke}{rgb}{0.000000,0.000000,0.000000}%
\pgfsetstrokecolor{currentstroke}%
\pgfsetstrokeopacity{0.000000}%
\pgfsetdash{}{0pt}%
\pgfpathmoveto{\pgfqpoint{0.100000in}{0.100000in}}%
\pgfpathlineto{\pgfqpoint{3.796000in}{0.100000in}}%
\pgfpathlineto{\pgfqpoint{3.796000in}{3.796000in}}%
\pgfpathlineto{\pgfqpoint{0.100000in}{3.796000in}}%
\pgfpathlineto{\pgfqpoint{0.100000in}{0.100000in}}%
\pgfpathclose%
\pgfusepath{fill}%
\end{pgfscope}%
\begin{pgfscope}%
\pgfsetbuttcap%
\pgfsetmiterjoin%
\definecolor{currentfill}{rgb}{0.950000,0.950000,0.950000}%
\pgfsetfillcolor{currentfill}%
\pgfsetfillopacity{0.500000}%
\pgfsetlinewidth{1.003750pt}%
\definecolor{currentstroke}{rgb}{0.950000,0.950000,0.950000}%
\pgfsetstrokecolor{currentstroke}%
\pgfsetstrokeopacity{0.500000}%
\pgfsetdash{}{0pt}%
\pgfpathmoveto{\pgfqpoint{0.435156in}{0.907452in}}%
\pgfpathlineto{\pgfqpoint{1.368966in}{1.710475in}}%
\pgfpathlineto{\pgfqpoint{1.350794in}{3.287994in}}%
\pgfpathlineto{\pgfqpoint{0.385055in}{2.632186in}}%
\pgfusepath{stroke,fill}%
\end{pgfscope}%
\begin{pgfscope}%
\pgfsetbuttcap%
\pgfsetmiterjoin%
\definecolor{currentfill}{rgb}{0.900000,0.900000,0.900000}%
\pgfsetfillcolor{currentfill}%
\pgfsetfillopacity{0.500000}%
\pgfsetlinewidth{1.003750pt}%
\definecolor{currentstroke}{rgb}{0.900000,0.900000,0.900000}%
\pgfsetstrokecolor{currentstroke}%
\pgfsetstrokeopacity{0.500000}%
\pgfsetdash{}{0pt}%
\pgfpathmoveto{\pgfqpoint{1.368966in}{1.710475in}}%
\pgfpathlineto{\pgfqpoint{3.464729in}{1.421160in}}%
\pgfpathlineto{\pgfqpoint{3.508776in}{3.052182in}}%
\pgfpathlineto{\pgfqpoint{1.350794in}{3.287994in}}%
\pgfusepath{stroke,fill}%
\end{pgfscope}%
\begin{pgfscope}%
\pgfsetbuttcap%
\pgfsetmiterjoin%
\definecolor{currentfill}{rgb}{0.925000,0.925000,0.925000}%
\pgfsetfillcolor{currentfill}%
\pgfsetfillopacity{0.500000}%
\pgfsetlinewidth{1.003750pt}%
\definecolor{currentstroke}{rgb}{0.925000,0.925000,0.925000}%
\pgfsetstrokecolor{currentstroke}%
\pgfsetstrokeopacity{0.500000}%
\pgfsetdash{}{0pt}%
\pgfpathmoveto{\pgfqpoint{0.435156in}{0.907452in}}%
\pgfpathlineto{\pgfqpoint{2.723326in}{0.551957in}}%
\pgfpathlineto{\pgfqpoint{3.464729in}{1.421160in}}%
\pgfpathlineto{\pgfqpoint{1.368966in}{1.710475in}}%
\pgfusepath{stroke,fill}%
\end{pgfscope}%
\begin{pgfscope}%
\pgfsetrectcap%
\pgfsetroundjoin%
\pgfsetlinewidth{0.803000pt}%
\definecolor{currentstroke}{rgb}{0.000000,0.000000,0.000000}%
\pgfsetstrokecolor{currentstroke}%
\pgfsetdash{}{0pt}%
\pgfpathmoveto{\pgfqpoint{0.435156in}{0.907452in}}%
\pgfpathlineto{\pgfqpoint{2.723326in}{0.551957in}}%
\pgfusepath{stroke}%
\end{pgfscope}%
\begin{pgfscope}%
\definecolor{textcolor}{rgb}{0.000000,0.000000,0.000000}%
\pgfsetstrokecolor{textcolor}%
\pgfsetfillcolor{textcolor}%
\pgftext[x=1.130034in, y=0.262435in, left, base,rotate=351.169007]{\color{textcolor}\rmfamily\fontsize{10.000000}{12.000000}\selectfont \(\displaystyle x_{i,1}\) in m}%
\end{pgfscope}%
\begin{pgfscope}%
\pgfsetbuttcap%
\pgfsetroundjoin%
\pgfsetlinewidth{0.803000pt}%
\definecolor{currentstroke}{rgb}{0.690196,0.690196,0.690196}%
\pgfsetstrokecolor{currentstroke}%
\pgfsetdash{}{0pt}%
\pgfpathmoveto{\pgfqpoint{0.496539in}{0.897916in}}%
\pgfpathlineto{\pgfqpoint{1.425411in}{1.702683in}}%
\pgfpathlineto{\pgfqpoint{1.408852in}{3.281649in}}%
\pgfusepath{stroke}%
\end{pgfscope}%
\begin{pgfscope}%
\pgfsetbuttcap%
\pgfsetroundjoin%
\pgfsetlinewidth{0.803000pt}%
\definecolor{currentstroke}{rgb}{0.690196,0.690196,0.690196}%
\pgfsetstrokecolor{currentstroke}%
\pgfsetdash{}{0pt}%
\pgfpathmoveto{\pgfqpoint{0.899619in}{0.835292in}}%
\pgfpathlineto{\pgfqpoint{1.795757in}{1.651557in}}%
\pgfpathlineto{\pgfqpoint{1.789869in}{3.240014in}}%
\pgfusepath{stroke}%
\end{pgfscope}%
\begin{pgfscope}%
\pgfsetbuttcap%
\pgfsetroundjoin%
\pgfsetlinewidth{0.803000pt}%
\definecolor{currentstroke}{rgb}{0.690196,0.690196,0.690196}%
\pgfsetstrokecolor{currentstroke}%
\pgfsetdash{}{0pt}%
\pgfpathmoveto{\pgfqpoint{1.308775in}{0.771725in}}%
\pgfpathlineto{\pgfqpoint{2.171141in}{1.599737in}}%
\pgfpathlineto{\pgfqpoint{2.176220in}{3.197796in}}%
\pgfusepath{stroke}%
\end{pgfscope}%
\begin{pgfscope}%
\pgfsetbuttcap%
\pgfsetroundjoin%
\pgfsetlinewidth{0.803000pt}%
\definecolor{currentstroke}{rgb}{0.690196,0.690196,0.690196}%
\pgfsetstrokecolor{currentstroke}%
\pgfsetdash{}{0pt}%
\pgfpathmoveto{\pgfqpoint{1.724146in}{0.707192in}}%
\pgfpathlineto{\pgfqpoint{2.551668in}{1.547206in}}%
\pgfpathlineto{\pgfqpoint{2.568021in}{3.154982in}}%
\pgfusepath{stroke}%
\end{pgfscope}%
\begin{pgfscope}%
\pgfsetbuttcap%
\pgfsetroundjoin%
\pgfsetlinewidth{0.803000pt}%
\definecolor{currentstroke}{rgb}{0.690196,0.690196,0.690196}%
\pgfsetstrokecolor{currentstroke}%
\pgfsetdash{}{0pt}%
\pgfpathmoveto{\pgfqpoint{2.145873in}{0.641672in}}%
\pgfpathlineto{\pgfqpoint{2.937443in}{1.493951in}}%
\pgfpathlineto{\pgfqpoint{2.965386in}{3.111561in}}%
\pgfusepath{stroke}%
\end{pgfscope}%
\begin{pgfscope}%
\pgfsetbuttcap%
\pgfsetroundjoin%
\pgfsetlinewidth{0.803000pt}%
\definecolor{currentstroke}{rgb}{0.690196,0.690196,0.690196}%
\pgfsetstrokecolor{currentstroke}%
\pgfsetdash{}{0pt}%
\pgfpathmoveto{\pgfqpoint{2.574105in}{0.575141in}}%
\pgfpathlineto{\pgfqpoint{3.328575in}{1.439956in}}%
\pgfpathlineto{\pgfqpoint{3.368435in}{3.067518in}}%
\pgfusepath{stroke}%
\end{pgfscope}%
\begin{pgfscope}%
\pgfsetrectcap%
\pgfsetroundjoin%
\pgfsetlinewidth{0.803000pt}%
\definecolor{currentstroke}{rgb}{0.000000,0.000000,0.000000}%
\pgfsetstrokecolor{currentstroke}%
\pgfsetdash{}{0pt}%
\pgfpathmoveto{\pgfqpoint{0.504753in}{0.905033in}}%
\pgfpathlineto{\pgfqpoint{0.480068in}{0.883646in}}%
\pgfusepath{stroke}%
\end{pgfscope}%
\begin{pgfscope}%
\definecolor{textcolor}{rgb}{0.000000,0.000000,0.000000}%
\pgfsetstrokecolor{textcolor}%
\pgfsetfillcolor{textcolor}%
\pgftext[x=0.414634in,y=0.683910in,,top]{\color{textcolor}\rmfamily\fontsize{10.000000}{12.000000}\selectfont \ensuremath{-}0.6}%
\end{pgfscope}%
\begin{pgfscope}%
\pgfsetrectcap%
\pgfsetroundjoin%
\pgfsetlinewidth{0.803000pt}%
\definecolor{currentstroke}{rgb}{0.000000,0.000000,0.000000}%
\pgfsetstrokecolor{currentstroke}%
\pgfsetdash{}{0pt}%
\pgfpathmoveto{\pgfqpoint{0.907550in}{0.842516in}}%
\pgfpathlineto{\pgfqpoint{0.883717in}{0.820808in}}%
\pgfusepath{stroke}%
\end{pgfscope}%
\begin{pgfscope}%
\definecolor{textcolor}{rgb}{0.000000,0.000000,0.000000}%
\pgfsetstrokecolor{textcolor}%
\pgfsetfillcolor{textcolor}%
\pgftext[x=0.819649in,y=0.619279in,,top]{\color{textcolor}\rmfamily\fontsize{10.000000}{12.000000}\selectfont \ensuremath{-}0.4}%
\end{pgfscope}%
\begin{pgfscope}%
\pgfsetrectcap%
\pgfsetroundjoin%
\pgfsetlinewidth{0.803000pt}%
\definecolor{currentstroke}{rgb}{0.000000,0.000000,0.000000}%
\pgfsetstrokecolor{currentstroke}%
\pgfsetdash{}{0pt}%
\pgfpathmoveto{\pgfqpoint{1.316412in}{0.779058in}}%
\pgfpathlineto{\pgfqpoint{1.293461in}{0.757021in}}%
\pgfusepath{stroke}%
\end{pgfscope}%
\begin{pgfscope}%
\definecolor{textcolor}{rgb}{0.000000,0.000000,0.000000}%
\pgfsetstrokecolor{textcolor}%
\pgfsetfillcolor{textcolor}%
\pgftext[x=1.230809in,y=0.553668in,,top]{\color{textcolor}\rmfamily\fontsize{10.000000}{12.000000}\selectfont \ensuremath{-}0.2}%
\end{pgfscope}%
\begin{pgfscope}%
\pgfsetrectcap%
\pgfsetroundjoin%
\pgfsetlinewidth{0.803000pt}%
\definecolor{currentstroke}{rgb}{0.000000,0.000000,0.000000}%
\pgfsetstrokecolor{currentstroke}%
\pgfsetdash{}{0pt}%
\pgfpathmoveto{\pgfqpoint{1.731480in}{0.714637in}}%
\pgfpathlineto{\pgfqpoint{1.709440in}{0.692264in}}%
\pgfusepath{stroke}%
\end{pgfscope}%
\begin{pgfscope}%
\definecolor{textcolor}{rgb}{0.000000,0.000000,0.000000}%
\pgfsetstrokecolor{textcolor}%
\pgfsetfillcolor{textcolor}%
\pgftext[x=1.648253in,y=0.487054in,,top]{\color{textcolor}\rmfamily\fontsize{10.000000}{12.000000}\selectfont 0.0}%
\end{pgfscope}%
\begin{pgfscope}%
\pgfsetrectcap%
\pgfsetroundjoin%
\pgfsetlinewidth{0.803000pt}%
\definecolor{currentstroke}{rgb}{0.000000,0.000000,0.000000}%
\pgfsetstrokecolor{currentstroke}%
\pgfsetdash{}{0pt}%
\pgfpathmoveto{\pgfqpoint{2.152894in}{0.649231in}}%
\pgfpathlineto{\pgfqpoint{2.131795in}{0.626514in}}%
\pgfusepath{stroke}%
\end{pgfscope}%
\begin{pgfscope}%
\definecolor{textcolor}{rgb}{0.000000,0.000000,0.000000}%
\pgfsetstrokecolor{textcolor}%
\pgfsetfillcolor{textcolor}%
\pgftext[x=2.072128in,y=0.419414in,,top]{\color{textcolor}\rmfamily\fontsize{10.000000}{12.000000}\selectfont 0.2}%
\end{pgfscope}%
\begin{pgfscope}%
\pgfsetrectcap%
\pgfsetroundjoin%
\pgfsetlinewidth{0.803000pt}%
\definecolor{currentstroke}{rgb}{0.000000,0.000000,0.000000}%
\pgfsetstrokecolor{currentstroke}%
\pgfsetdash{}{0pt}%
\pgfpathmoveto{\pgfqpoint{2.580801in}{0.582817in}}%
\pgfpathlineto{\pgfqpoint{2.560676in}{0.559748in}}%
\pgfusepath{stroke}%
\end{pgfscope}%
\begin{pgfscope}%
\definecolor{textcolor}{rgb}{0.000000,0.000000,0.000000}%
\pgfsetstrokecolor{textcolor}%
\pgfsetfillcolor{textcolor}%
\pgftext[x=2.502582in,y=0.350723in,,top]{\color{textcolor}\rmfamily\fontsize{10.000000}{12.000000}\selectfont 0.4}%
\end{pgfscope}%
\begin{pgfscope}%
\pgfsetrectcap%
\pgfsetroundjoin%
\pgfsetlinewidth{0.803000pt}%
\definecolor{currentstroke}{rgb}{0.000000,0.000000,0.000000}%
\pgfsetstrokecolor{currentstroke}%
\pgfsetdash{}{0pt}%
\pgfpathmoveto{\pgfqpoint{3.464729in}{1.421160in}}%
\pgfpathlineto{\pgfqpoint{2.723326in}{0.551957in}}%
\pgfusepath{stroke}%
\end{pgfscope}%
\begin{pgfscope}%
\definecolor{textcolor}{rgb}{0.000000,0.000000,0.000000}%
\pgfsetstrokecolor{textcolor}%
\pgfsetfillcolor{textcolor}%
\pgftext[x=3.375618in, y=0.436233in, left, base,rotate=49.536842]{\color{textcolor}\rmfamily\fontsize{10.000000}{12.000000}\selectfont \(\displaystyle x_{i,2}\) in m}%
\end{pgfscope}%
\begin{pgfscope}%
\pgfsetbuttcap%
\pgfsetroundjoin%
\pgfsetlinewidth{0.803000pt}%
\definecolor{currentstroke}{rgb}{0.690196,0.690196,0.690196}%
\pgfsetstrokecolor{currentstroke}%
\pgfsetdash{}{0pt}%
\pgfpathmoveto{\pgfqpoint{0.412782in}{2.651015in}}%
\pgfpathlineto{\pgfqpoint{0.461886in}{0.930439in}}%
\pgfpathlineto{\pgfqpoint{2.744633in}{0.576938in}}%
\pgfusepath{stroke}%
\end{pgfscope}%
\begin{pgfscope}%
\pgfsetbuttcap%
\pgfsetroundjoin%
\pgfsetlinewidth{0.803000pt}%
\definecolor{currentstroke}{rgb}{0.690196,0.690196,0.690196}%
\pgfsetstrokecolor{currentstroke}%
\pgfsetdash{}{0pt}%
\pgfpathmoveto{\pgfqpoint{0.599955in}{2.778119in}}%
\pgfpathlineto{\pgfqpoint{0.642455in}{1.085717in}}%
\pgfpathlineto{\pgfqpoint{2.888434in}{0.745526in}}%
\pgfusepath{stroke}%
\end{pgfscope}%
\begin{pgfscope}%
\pgfsetbuttcap%
\pgfsetroundjoin%
\pgfsetlinewidth{0.803000pt}%
\definecolor{currentstroke}{rgb}{0.690196,0.690196,0.690196}%
\pgfsetstrokecolor{currentstroke}%
\pgfsetdash{}{0pt}%
\pgfpathmoveto{\pgfqpoint{0.780074in}{2.900433in}}%
\pgfpathlineto{\pgfqpoint{0.816421in}{1.235318in}}%
\pgfpathlineto{\pgfqpoint{3.026763in}{0.907700in}}%
\pgfusepath{stroke}%
\end{pgfscope}%
\begin{pgfscope}%
\pgfsetbuttcap%
\pgfsetroundjoin%
\pgfsetlinewidth{0.803000pt}%
\definecolor{currentstroke}{rgb}{0.690196,0.690196,0.690196}%
\pgfsetstrokecolor{currentstroke}%
\pgfsetdash{}{0pt}%
\pgfpathmoveto{\pgfqpoint{0.953531in}{3.018223in}}%
\pgfpathlineto{\pgfqpoint{0.984141in}{1.379548in}}%
\pgfpathlineto{\pgfqpoint{3.159927in}{1.063818in}}%
\pgfusepath{stroke}%
\end{pgfscope}%
\begin{pgfscope}%
\pgfsetbuttcap%
\pgfsetroundjoin%
\pgfsetlinewidth{0.803000pt}%
\definecolor{currentstroke}{rgb}{0.690196,0.690196,0.690196}%
\pgfsetstrokecolor{currentstroke}%
\pgfsetdash{}{0pt}%
\pgfpathmoveto{\pgfqpoint{1.120688in}{3.131735in}}%
\pgfpathlineto{\pgfqpoint{1.145946in}{1.518691in}}%
\pgfpathlineto{\pgfqpoint{3.288210in}{1.214214in}}%
\pgfusepath{stroke}%
\end{pgfscope}%
\begin{pgfscope}%
\pgfsetbuttcap%
\pgfsetroundjoin%
\pgfsetlinewidth{0.803000pt}%
\definecolor{currentstroke}{rgb}{0.690196,0.690196,0.690196}%
\pgfsetstrokecolor{currentstroke}%
\pgfsetdash{}{0pt}%
\pgfpathmoveto{\pgfqpoint{1.281883in}{3.241198in}}%
\pgfpathlineto{\pgfqpoint{1.302143in}{1.653011in}}%
\pgfpathlineto{\pgfqpoint{3.411875in}{1.359195in}}%
\pgfusepath{stroke}%
\end{pgfscope}%
\begin{pgfscope}%
\pgfsetrectcap%
\pgfsetroundjoin%
\pgfsetlinewidth{0.803000pt}%
\definecolor{currentstroke}{rgb}{0.000000,0.000000,0.000000}%
\pgfsetstrokecolor{currentstroke}%
\pgfsetdash{}{0pt}%
\pgfpathmoveto{\pgfqpoint{2.725603in}{0.579885in}}%
\pgfpathlineto{\pgfqpoint{2.782731in}{0.571038in}}%
\pgfusepath{stroke}%
\end{pgfscope}%
\begin{pgfscope}%
\definecolor{textcolor}{rgb}{0.000000,0.000000,0.000000}%
\pgfsetstrokecolor{textcolor}%
\pgfsetfillcolor{textcolor}%
\pgftext[x=2.940231in,y=0.402962in,,top]{\color{textcolor}\rmfamily\fontsize{10.000000}{12.000000}\selectfont \ensuremath{-}0.6}%
\end{pgfscope}%
\begin{pgfscope}%
\pgfsetrectcap%
\pgfsetroundjoin%
\pgfsetlinewidth{0.803000pt}%
\definecolor{currentstroke}{rgb}{0.000000,0.000000,0.000000}%
\pgfsetstrokecolor{currentstroke}%
\pgfsetdash{}{0pt}%
\pgfpathmoveto{\pgfqpoint{2.869725in}{0.748360in}}%
\pgfpathlineto{\pgfqpoint{2.925888in}{0.739853in}}%
\pgfusepath{stroke}%
\end{pgfscope}%
\begin{pgfscope}%
\definecolor{textcolor}{rgb}{0.000000,0.000000,0.000000}%
\pgfsetstrokecolor{textcolor}%
\pgfsetfillcolor{textcolor}%
\pgftext[x=3.080353in,y=0.575066in,,top]{\color{textcolor}\rmfamily\fontsize{10.000000}{12.000000}\selectfont \ensuremath{-}0.4}%
\end{pgfscope}%
\begin{pgfscope}%
\pgfsetrectcap%
\pgfsetroundjoin%
\pgfsetlinewidth{0.803000pt}%
\definecolor{currentstroke}{rgb}{0.000000,0.000000,0.000000}%
\pgfsetstrokecolor{currentstroke}%
\pgfsetdash{}{0pt}%
\pgfpathmoveto{\pgfqpoint{3.008365in}{0.910427in}}%
\pgfpathlineto{\pgfqpoint{3.063595in}{0.902241in}}%
\pgfusepath{stroke}%
\end{pgfscope}%
\begin{pgfscope}%
\definecolor{textcolor}{rgb}{0.000000,0.000000,0.000000}%
\pgfsetstrokecolor{textcolor}%
\pgfsetfillcolor{textcolor}%
\pgftext[x=3.215140in,y=0.740617in,,top]{\color{textcolor}\rmfamily\fontsize{10.000000}{12.000000}\selectfont \ensuremath{-}0.2}%
\end{pgfscope}%
\begin{pgfscope}%
\pgfsetrectcap%
\pgfsetroundjoin%
\pgfsetlinewidth{0.803000pt}%
\definecolor{currentstroke}{rgb}{0.000000,0.000000,0.000000}%
\pgfsetstrokecolor{currentstroke}%
\pgfsetdash{}{0pt}%
\pgfpathmoveto{\pgfqpoint{3.141830in}{1.066444in}}%
\pgfpathlineto{\pgfqpoint{3.196156in}{1.058561in}}%
\pgfusepath{stroke}%
\end{pgfscope}%
\begin{pgfscope}%
\definecolor{textcolor}{rgb}{0.000000,0.000000,0.000000}%
\pgfsetstrokecolor{textcolor}%
\pgfsetfillcolor{textcolor}%
\pgftext[x=3.344889in,y=0.899981in,,top]{\color{textcolor}\rmfamily\fontsize{10.000000}{12.000000}\selectfont 0.0}%
\end{pgfscope}%
\begin{pgfscope}%
\pgfsetrectcap%
\pgfsetroundjoin%
\pgfsetlinewidth{0.803000pt}%
\definecolor{currentstroke}{rgb}{0.000000,0.000000,0.000000}%
\pgfsetstrokecolor{currentstroke}%
\pgfsetdash{}{0pt}%
\pgfpathmoveto{\pgfqpoint{3.270404in}{1.216744in}}%
\pgfpathlineto{\pgfqpoint{3.323855in}{1.209147in}}%
\pgfusepath{stroke}%
\end{pgfscope}%
\begin{pgfscope}%
\definecolor{textcolor}{rgb}{0.000000,0.000000,0.000000}%
\pgfsetstrokecolor{textcolor}%
\pgfsetfillcolor{textcolor}%
\pgftext[x=3.469879in,y=1.053498in,,top]{\color{textcolor}\rmfamily\fontsize{10.000000}{12.000000}\selectfont 0.2}%
\end{pgfscope}%
\begin{pgfscope}%
\pgfsetrectcap%
\pgfsetroundjoin%
\pgfsetlinewidth{0.803000pt}%
\definecolor{currentstroke}{rgb}{0.000000,0.000000,0.000000}%
\pgfsetstrokecolor{currentstroke}%
\pgfsetdash{}{0pt}%
\pgfpathmoveto{\pgfqpoint{3.394351in}{1.361636in}}%
\pgfpathlineto{\pgfqpoint{3.446954in}{1.354310in}}%
\pgfusepath{stroke}%
\end{pgfscope}%
\begin{pgfscope}%
\definecolor{textcolor}{rgb}{0.000000,0.000000,0.000000}%
\pgfsetstrokecolor{textcolor}%
\pgfsetfillcolor{textcolor}%
\pgftext[x=3.590366in,y=1.201485in,,top]{\color{textcolor}\rmfamily\fontsize{10.000000}{12.000000}\selectfont 0.4}%
\end{pgfscope}%
\begin{pgfscope}%
\pgfsetrectcap%
\pgfsetroundjoin%
\pgfsetlinewidth{0.803000pt}%
\definecolor{currentstroke}{rgb}{0.000000,0.000000,0.000000}%
\pgfsetstrokecolor{currentstroke}%
\pgfsetdash{}{0pt}%
\pgfpathmoveto{\pgfqpoint{3.464729in}{1.421160in}}%
\pgfpathlineto{\pgfqpoint{3.508776in}{3.052182in}}%
\pgfusepath{stroke}%
\end{pgfscope}%
\begin{pgfscope}%
\definecolor{textcolor}{rgb}{0.000000,0.000000,0.000000}%
\pgfsetstrokecolor{textcolor}%
\pgfsetfillcolor{textcolor}%
\pgftext[x=4.024207in, y=2.046532in, left, base,rotate=88.453050]{\color{textcolor}\rmfamily\fontsize{10.000000}{12.000000}\selectfont \(\displaystyle x_{i,3}\) in m}%
\end{pgfscope}%
\begin{pgfscope}%
\pgfsetbuttcap%
\pgfsetroundjoin%
\pgfsetlinewidth{0.803000pt}%
\definecolor{currentstroke}{rgb}{0.690196,0.690196,0.690196}%
\pgfsetstrokecolor{currentstroke}%
\pgfsetdash{}{0pt}%
\pgfpathmoveto{\pgfqpoint{3.465585in}{1.452848in}}%
\pgfpathlineto{\pgfqpoint{1.368612in}{1.741156in}}%
\pgfpathlineto{\pgfqpoint{0.434184in}{0.940896in}}%
\pgfusepath{stroke}%
\end{pgfscope}%
\begin{pgfscope}%
\pgfsetbuttcap%
\pgfsetroundjoin%
\pgfsetlinewidth{0.803000pt}%
\definecolor{currentstroke}{rgb}{0.690196,0.690196,0.690196}%
\pgfsetstrokecolor{currentstroke}%
\pgfsetdash{}{0pt}%
\pgfpathmoveto{\pgfqpoint{3.470740in}{1.643755in}}%
\pgfpathlineto{\pgfqpoint{1.366483in}{1.925973in}}%
\pgfpathlineto{\pgfqpoint{0.428330in}{1.142437in}}%
\pgfusepath{stroke}%
\end{pgfscope}%
\begin{pgfscope}%
\pgfsetbuttcap%
\pgfsetroundjoin%
\pgfsetlinewidth{0.803000pt}%
\definecolor{currentstroke}{rgb}{0.690196,0.690196,0.690196}%
\pgfsetstrokecolor{currentstroke}%
\pgfsetdash{}{0pt}%
\pgfpathmoveto{\pgfqpoint{3.475932in}{1.836007in}}%
\pgfpathlineto{\pgfqpoint{1.364340in}{2.112044in}}%
\pgfpathlineto{\pgfqpoint{0.422431in}{1.345491in}}%
\pgfusepath{stroke}%
\end{pgfscope}%
\begin{pgfscope}%
\pgfsetbuttcap%
\pgfsetroundjoin%
\pgfsetlinewidth{0.803000pt}%
\definecolor{currentstroke}{rgb}{0.690196,0.690196,0.690196}%
\pgfsetstrokecolor{currentstroke}%
\pgfsetdash{}{0pt}%
\pgfpathmoveto{\pgfqpoint{3.481161in}{2.029620in}}%
\pgfpathlineto{\pgfqpoint{1.362182in}{2.299383in}}%
\pgfpathlineto{\pgfqpoint{0.416488in}{1.550078in}}%
\pgfusepath{stroke}%
\end{pgfscope}%
\begin{pgfscope}%
\pgfsetbuttcap%
\pgfsetroundjoin%
\pgfsetlinewidth{0.803000pt}%
\definecolor{currentstroke}{rgb}{0.690196,0.690196,0.690196}%
\pgfsetstrokecolor{currentstroke}%
\pgfsetdash{}{0pt}%
\pgfpathmoveto{\pgfqpoint{3.486427in}{2.224608in}}%
\pgfpathlineto{\pgfqpoint{1.360009in}{2.488002in}}%
\pgfpathlineto{\pgfqpoint{0.410500in}{1.756214in}}%
\pgfusepath{stroke}%
\end{pgfscope}%
\begin{pgfscope}%
\pgfsetbuttcap%
\pgfsetroundjoin%
\pgfsetlinewidth{0.803000pt}%
\definecolor{currentstroke}{rgb}{0.690196,0.690196,0.690196}%
\pgfsetstrokecolor{currentstroke}%
\pgfsetdash{}{0pt}%
\pgfpathmoveto{\pgfqpoint{3.491730in}{2.420985in}}%
\pgfpathlineto{\pgfqpoint{1.357822in}{2.677915in}}%
\pgfpathlineto{\pgfqpoint{0.404467in}{1.963917in}}%
\pgfusepath{stroke}%
\end{pgfscope}%
\begin{pgfscope}%
\pgfsetbuttcap%
\pgfsetroundjoin%
\pgfsetlinewidth{0.803000pt}%
\definecolor{currentstroke}{rgb}{0.690196,0.690196,0.690196}%
\pgfsetstrokecolor{currentstroke}%
\pgfsetdash{}{0pt}%
\pgfpathmoveto{\pgfqpoint{3.497071in}{2.618766in}}%
\pgfpathlineto{\pgfqpoint{1.355619in}{2.869135in}}%
\pgfpathlineto{\pgfqpoint{0.398388in}{2.173205in}}%
\pgfusepath{stroke}%
\end{pgfscope}%
\begin{pgfscope}%
\pgfsetbuttcap%
\pgfsetroundjoin%
\pgfsetlinewidth{0.803000pt}%
\definecolor{currentstroke}{rgb}{0.690196,0.690196,0.690196}%
\pgfsetstrokecolor{currentstroke}%
\pgfsetdash{}{0pt}%
\pgfpathmoveto{\pgfqpoint{3.502451in}{2.817967in}}%
\pgfpathlineto{\pgfqpoint{1.353401in}{3.061676in}}%
\pgfpathlineto{\pgfqpoint{0.392262in}{2.384096in}}%
\pgfusepath{stroke}%
\end{pgfscope}%
\begin{pgfscope}%
\pgfsetbuttcap%
\pgfsetroundjoin%
\pgfsetlinewidth{0.803000pt}%
\definecolor{currentstroke}{rgb}{0.690196,0.690196,0.690196}%
\pgfsetstrokecolor{currentstroke}%
\pgfsetdash{}{0pt}%
\pgfpathmoveto{\pgfqpoint{3.507869in}{3.018602in}}%
\pgfpathlineto{\pgfqpoint{1.351168in}{3.255550in}}%
\pgfpathlineto{\pgfqpoint{0.386088in}{2.596608in}}%
\pgfusepath{stroke}%
\end{pgfscope}%
\begin{pgfscope}%
\pgfsetrectcap%
\pgfsetroundjoin%
\pgfsetlinewidth{0.803000pt}%
\definecolor{currentstroke}{rgb}{0.000000,0.000000,0.000000}%
\pgfsetstrokecolor{currentstroke}%
\pgfsetdash{}{0pt}%
\pgfpathmoveto{\pgfqpoint{3.448172in}{1.455242in}}%
\pgfpathlineto{\pgfqpoint{3.500442in}{1.448056in}}%
\pgfusepath{stroke}%
\end{pgfscope}%
\begin{pgfscope}%
\definecolor{textcolor}{rgb}{0.000000,0.000000,0.000000}%
\pgfsetstrokecolor{textcolor}%
\pgfsetfillcolor{textcolor}%
\pgftext[x=3.700449in,y=1.492513in,,top]{\color{textcolor}\rmfamily\fontsize{10.000000}{12.000000}\selectfont 1.00}%
\end{pgfscope}%
\begin{pgfscope}%
\pgfsetrectcap%
\pgfsetroundjoin%
\pgfsetlinewidth{0.803000pt}%
\definecolor{currentstroke}{rgb}{0.000000,0.000000,0.000000}%
\pgfsetstrokecolor{currentstroke}%
\pgfsetdash{}{0pt}%
\pgfpathmoveto{\pgfqpoint{3.453265in}{1.646098in}}%
\pgfpathlineto{\pgfqpoint{3.505724in}{1.639063in}}%
\pgfusepath{stroke}%
\end{pgfscope}%
\begin{pgfscope}%
\definecolor{textcolor}{rgb}{0.000000,0.000000,0.000000}%
\pgfsetstrokecolor{textcolor}%
\pgfsetfillcolor{textcolor}%
\pgftext[x=3.706399in,y=1.682586in,,top]{\color{textcolor}\rmfamily\fontsize{10.000000}{12.000000}\selectfont 1.25}%
\end{pgfscope}%
\begin{pgfscope}%
\pgfsetrectcap%
\pgfsetroundjoin%
\pgfsetlinewidth{0.803000pt}%
\definecolor{currentstroke}{rgb}{0.000000,0.000000,0.000000}%
\pgfsetstrokecolor{currentstroke}%
\pgfsetdash{}{0pt}%
\pgfpathmoveto{\pgfqpoint{3.458393in}{1.838300in}}%
\pgfpathlineto{\pgfqpoint{3.511042in}{1.831417in}}%
\pgfusepath{stroke}%
\end{pgfscope}%
\begin{pgfscope}%
\definecolor{textcolor}{rgb}{0.000000,0.000000,0.000000}%
\pgfsetstrokecolor{textcolor}%
\pgfsetfillcolor{textcolor}%
\pgftext[x=3.712391in,y=1.873992in,,top]{\color{textcolor}\rmfamily\fontsize{10.000000}{12.000000}\selectfont 1.50}%
\end{pgfscope}%
\begin{pgfscope}%
\pgfsetrectcap%
\pgfsetroundjoin%
\pgfsetlinewidth{0.803000pt}%
\definecolor{currentstroke}{rgb}{0.000000,0.000000,0.000000}%
\pgfsetstrokecolor{currentstroke}%
\pgfsetdash{}{0pt}%
\pgfpathmoveto{\pgfqpoint{3.463558in}{2.031861in}}%
\pgfpathlineto{\pgfqpoint{3.516398in}{2.025134in}}%
\pgfusepath{stroke}%
\end{pgfscope}%
\begin{pgfscope}%
\definecolor{textcolor}{rgb}{0.000000,0.000000,0.000000}%
\pgfsetstrokecolor{textcolor}%
\pgfsetfillcolor{textcolor}%
\pgftext[x=3.718425in,y=2.066746in,,top]{\color{textcolor}\rmfamily\fontsize{10.000000}{12.000000}\selectfont 1.75}%
\end{pgfscope}%
\begin{pgfscope}%
\pgfsetrectcap%
\pgfsetroundjoin%
\pgfsetlinewidth{0.803000pt}%
\definecolor{currentstroke}{rgb}{0.000000,0.000000,0.000000}%
\pgfsetstrokecolor{currentstroke}%
\pgfsetdash{}{0pt}%
\pgfpathmoveto{\pgfqpoint{3.468760in}{2.226796in}}%
\pgfpathlineto{\pgfqpoint{3.521793in}{2.220227in}}%
\pgfusepath{stroke}%
\end{pgfscope}%
\begin{pgfscope}%
\definecolor{textcolor}{rgb}{0.000000,0.000000,0.000000}%
\pgfsetstrokecolor{textcolor}%
\pgfsetfillcolor{textcolor}%
\pgftext[x=3.724502in,y=2.260861in,,top]{\color{textcolor}\rmfamily\fontsize{10.000000}{12.000000}\selectfont 2.00}%
\end{pgfscope}%
\begin{pgfscope}%
\pgfsetrectcap%
\pgfsetroundjoin%
\pgfsetlinewidth{0.803000pt}%
\definecolor{currentstroke}{rgb}{0.000000,0.000000,0.000000}%
\pgfsetstrokecolor{currentstroke}%
\pgfsetdash{}{0pt}%
\pgfpathmoveto{\pgfqpoint{3.473999in}{2.423119in}}%
\pgfpathlineto{\pgfqpoint{3.527225in}{2.416711in}}%
\pgfusepath{stroke}%
\end{pgfscope}%
\begin{pgfscope}%
\definecolor{textcolor}{rgb}{0.000000,0.000000,0.000000}%
\pgfsetstrokecolor{textcolor}%
\pgfsetfillcolor{textcolor}%
\pgftext[x=3.730622in,y=2.456352in,,top]{\color{textcolor}\rmfamily\fontsize{10.000000}{12.000000}\selectfont 2.25}%
\end{pgfscope}%
\begin{pgfscope}%
\pgfsetrectcap%
\pgfsetroundjoin%
\pgfsetlinewidth{0.803000pt}%
\definecolor{currentstroke}{rgb}{0.000000,0.000000,0.000000}%
\pgfsetstrokecolor{currentstroke}%
\pgfsetdash{}{0pt}%
\pgfpathmoveto{\pgfqpoint{3.479275in}{2.620847in}}%
\pgfpathlineto{\pgfqpoint{3.532697in}{2.614601in}}%
\pgfusepath{stroke}%
\end{pgfscope}%
\begin{pgfscope}%
\definecolor{textcolor}{rgb}{0.000000,0.000000,0.000000}%
\pgfsetstrokecolor{textcolor}%
\pgfsetfillcolor{textcolor}%
\pgftext[x=3.736785in,y=2.653235in,,top]{\color{textcolor}\rmfamily\fontsize{10.000000}{12.000000}\selectfont 2.50}%
\end{pgfscope}%
\begin{pgfscope}%
\pgfsetrectcap%
\pgfsetroundjoin%
\pgfsetlinewidth{0.803000pt}%
\definecolor{currentstroke}{rgb}{0.000000,0.000000,0.000000}%
\pgfsetstrokecolor{currentstroke}%
\pgfsetdash{}{0pt}%
\pgfpathmoveto{\pgfqpoint{3.484589in}{2.819992in}}%
\pgfpathlineto{\pgfqpoint{3.538208in}{2.813912in}}%
\pgfusepath{stroke}%
\end{pgfscope}%
\begin{pgfscope}%
\definecolor{textcolor}{rgb}{0.000000,0.000000,0.000000}%
\pgfsetstrokecolor{textcolor}%
\pgfsetfillcolor{textcolor}%
\pgftext[x=3.742992in,y=2.851522in,,top]{\color{textcolor}\rmfamily\fontsize{10.000000}{12.000000}\selectfont 2.75}%
\end{pgfscope}%
\begin{pgfscope}%
\pgfsetrectcap%
\pgfsetroundjoin%
\pgfsetlinewidth{0.803000pt}%
\definecolor{currentstroke}{rgb}{0.000000,0.000000,0.000000}%
\pgfsetstrokecolor{currentstroke}%
\pgfsetdash{}{0pt}%
\pgfpathmoveto{\pgfqpoint{3.489941in}{3.020572in}}%
\pgfpathlineto{\pgfqpoint{3.543759in}{3.014659in}}%
\pgfusepath{stroke}%
\end{pgfscope}%
\begin{pgfscope}%
\definecolor{textcolor}{rgb}{0.000000,0.000000,0.000000}%
\pgfsetstrokecolor{textcolor}%
\pgfsetfillcolor{textcolor}%
\pgftext[x=3.749244in,y=3.051231in,,top]{\color{textcolor}\rmfamily\fontsize{10.000000}{12.000000}\selectfont 3.00}%
\end{pgfscope}%
\begin{pgfscope}%
\pgfpathrectangle{\pgfqpoint{0.100000in}{0.100000in}}{\pgfqpoint{3.696000in}{3.696000in}}%
\pgfusepath{clip}%
\pgfsetrectcap%
\pgfsetroundjoin%
\pgfsetlinewidth{1.003750pt}%
\definecolor{currentstroke}{rgb}{0,0.447,0.741}%
\pgfsetstrokecolor{currentstroke}%
\pgfsetdash{}{0pt}%
\pgfpathmoveto{\pgfqpoint{2.211233in}{1.233946in}}%
\pgfpathlineto{\pgfqpoint{2.213270in}{2.029678in}}%
\pgfpathlineto{\pgfqpoint{2.171511in}{2.051577in}}%
\pgfpathlineto{\pgfqpoint{2.070937in}{2.103223in}}%
\pgfpathlineto{\pgfqpoint{1.997382in}{2.138347in}}%
\pgfpathlineto{\pgfqpoint{1.932524in}{2.166686in}}%
\pgfpathlineto{\pgfqpoint{1.874943in}{2.189499in}}%
\pgfpathlineto{\pgfqpoint{1.814119in}{2.211045in}}%
\pgfpathlineto{\pgfqpoint{1.751899in}{2.230286in}}%
\pgfpathlineto{\pgfqpoint{1.690084in}{2.246558in}}%
\pgfpathlineto{\pgfqpoint{1.630281in}{2.259594in}}%
\pgfpathlineto{\pgfqpoint{1.573820in}{2.269457in}}%
\pgfpathlineto{\pgfqpoint{1.521696in}{2.276441in}}%
\pgfpathlineto{\pgfqpoint{1.463592in}{2.281773in}}%
\pgfpathlineto{\pgfqpoint{1.413796in}{2.284186in}}%
\pgfpathlineto{\pgfqpoint{1.364799in}{2.284430in}}%
\pgfpathlineto{\pgfqpoint{1.321033in}{2.282564in}}%
\pgfpathlineto{\pgfqpoint{1.281989in}{2.278754in}}%
\pgfpathlineto{\pgfqpoint{1.251082in}{2.273578in}}%
\pgfpathlineto{\pgfqpoint{1.229830in}{2.267915in}}%
\pgfpathlineto{\pgfqpoint{1.217831in}{2.262608in}}%
\pgfpathlineto{\pgfqpoint{1.217020in}{2.262059in}}%
\pgfpathlineto{\pgfqpoint{1.217020in}{2.262059in}}%
\pgfusepath{stroke}%
\end{pgfscope}%
\begin{pgfscope}%
\pgfpathrectangle{\pgfqpoint{0.100000in}{0.100000in}}{\pgfqpoint{3.696000in}{3.696000in}}%
\pgfusepath{clip}%
\pgfsetrectcap%
\pgfsetroundjoin%
\pgfsetlinewidth{1.003750pt}%
\definecolor{currentstroke}{rgb}{0,0.447,0.741}%
\pgfsetstrokecolor{currentstroke}%
\pgfsetdash{}{0pt}%
\pgfpathmoveto{\pgfqpoint{1.217020in}{2.262059in}}%
\pgfusepath{stroke}%
\end{pgfscope}%
\begin{pgfscope}%
\pgfpathrectangle{\pgfqpoint{0.100000in}{0.100000in}}{\pgfqpoint{3.696000in}{3.696000in}}%
\pgfusepath{clip}%
\pgfsetbuttcap%
\pgfsetroundjoin%
\definecolor{currentfill}{rgb}{0,0.447,0.741}%
\pgfsetfillcolor{currentfill}%
\pgfsetlinewidth{1.003750pt}%
\definecolor{currentstroke}{rgb}{0,0.447,0.741}%
\pgfsetstrokecolor{currentstroke}%
\pgfsetdash{}{0pt}%
\pgfsys@defobject{currentmarker}{\pgfqpoint{-0.041667in}{-0.041667in}}{\pgfqpoint{0.041667in}{0.041667in}}{%
\pgfpathmoveto{\pgfqpoint{-0.041667in}{-0.041667in}}%
\pgfpathlineto{\pgfqpoint{0.041667in}{0.041667in}}%
\pgfpathmoveto{\pgfqpoint{-0.041667in}{0.041667in}}%
\pgfpathlineto{\pgfqpoint{0.041667in}{-0.041667in}}%
\pgfusepath{stroke,fill}%
}%
\begin{pgfscope}%
\pgfsys@transformshift{1.217020in}{2.262059in}%
\pgfsys@useobject{currentmarker}{}%
\end{pgfscope}%
\end{pgfscope}%
\begin{pgfscope}%
\pgfpathrectangle{\pgfqpoint{0.100000in}{0.100000in}}{\pgfqpoint{3.696000in}{3.696000in}}%
\pgfusepath{clip}%
\pgfsetrectcap%
\pgfsetroundjoin%
\pgfsetlinewidth{1.003750pt}%
\definecolor{currentstroke}{rgb}{0,0.447,0.741}%
\pgfsetstrokecolor{currentstroke}%
\pgfsetdash{}{0pt}%
\pgfpathmoveto{\pgfqpoint{2.211233in}{1.233946in}}%
\pgfusepath{stroke}%
\end{pgfscope}%
\begin{pgfscope}%
\pgfpathrectangle{\pgfqpoint{0.100000in}{0.100000in}}{\pgfqpoint{3.696000in}{3.696000in}}%
\pgfusepath{clip}%
\pgfsetbuttcap%
\pgfsetroundjoin%
\definecolor{currentfill}{rgb}{0,0.447,0.741}%
\pgfsetfillcolor{currentfill}%
\pgfsetlinewidth{1.003750pt}%
\definecolor{currentstroke}{rgb}{0,0.447,0.741}%
\pgfsetstrokecolor{currentstroke}%
\pgfsetdash{}{0pt}%
\pgfsys@defobject{currentmarker}{\pgfqpoint{-0.034722in}{-0.034722in}}{\pgfqpoint{0.034722in}{0.034722in}}{%
\pgfpathmoveto{\pgfqpoint{0.000000in}{-0.034722in}}%
\pgfpathcurveto{\pgfqpoint{0.009208in}{-0.034722in}}{\pgfqpoint{0.018041in}{-0.031064in}}{\pgfqpoint{0.024552in}{-0.024552in}}%
\pgfpathcurveto{\pgfqpoint{0.031064in}{-0.018041in}}{\pgfqpoint{0.034722in}{-0.009208in}}{\pgfqpoint{0.034722in}{0.000000in}}%
\pgfpathcurveto{\pgfqpoint{0.034722in}{0.009208in}}{\pgfqpoint{0.031064in}{0.018041in}}{\pgfqpoint{0.024552in}{0.024552in}}%
\pgfpathcurveto{\pgfqpoint{0.018041in}{0.031064in}}{\pgfqpoint{0.009208in}{0.034722in}}{\pgfqpoint{0.000000in}{0.034722in}}%
\pgfpathcurveto{\pgfqpoint{-0.009208in}{0.034722in}}{\pgfqpoint{-0.018041in}{0.031064in}}{\pgfqpoint{-0.024552in}{0.024552in}}%
\pgfpathcurveto{\pgfqpoint{-0.031064in}{0.018041in}}{\pgfqpoint{-0.034722in}{0.009208in}}{\pgfqpoint{-0.034722in}{0.000000in}}%
\pgfpathcurveto{\pgfqpoint{-0.034722in}{-0.009208in}}{\pgfqpoint{-0.031064in}{-0.018041in}}{\pgfqpoint{-0.024552in}{-0.024552in}}%
\pgfpathcurveto{\pgfqpoint{-0.018041in}{-0.031064in}}{\pgfqpoint{-0.009208in}{-0.034722in}}{\pgfqpoint{0.000000in}{-0.034722in}}%
\pgfpathlineto{\pgfqpoint{0.000000in}{-0.034722in}}%
\pgfpathclose%
\pgfusepath{stroke,fill}%
}%
\begin{pgfscope}%
\pgfsys@transformshift{2.211233in}{1.233946in}%
\pgfsys@useobject{currentmarker}{}%
\end{pgfscope}%
\end{pgfscope}%
\begin{pgfscope}%
\pgfpathrectangle{\pgfqpoint{0.100000in}{0.100000in}}{\pgfqpoint{3.696000in}{3.696000in}}%
\pgfusepath{clip}%
\pgfsetbuttcap%
\pgfsetroundjoin%
\pgfsetlinewidth{1.003750pt}%
\definecolor{currentstroke}{rgb}{0.85,0.235,0.098}%
\pgfsetstrokecolor{currentstroke}%
\pgfsetdash{{6.400000pt}{1.600000pt}{1.000000pt}{1.600000pt}}{0.000000pt}%
\pgfpathmoveto{\pgfqpoint{2.214348in}{2.024437in}}%
\pgfpathlineto{\pgfqpoint{2.215429in}{2.029645in}}%
\pgfpathlineto{\pgfqpoint{2.244979in}{2.038492in}}%
\pgfpathlineto{\pgfqpoint{2.307614in}{2.055364in}}%
\pgfpathlineto{\pgfqpoint{2.380846in}{2.072893in}}%
\pgfpathlineto{\pgfqpoint{2.463353in}{2.090291in}}%
\pgfpathlineto{\pgfqpoint{2.550714in}{2.106618in}}%
\pgfpathlineto{\pgfqpoint{2.665646in}{2.125674in}}%
\pgfpathlineto{\pgfqpoint{2.800972in}{2.145521in}}%
\pgfpathlineto{\pgfqpoint{2.953981in}{2.165423in}}%
\pgfpathlineto{\pgfqpoint{3.111168in}{2.183503in}}%
\pgfpathlineto{\pgfqpoint{3.247560in}{2.197005in}}%
\pgfpathlineto{\pgfqpoint{3.302270in}{2.201279in}}%
\pgfpathlineto{\pgfqpoint{3.302270in}{2.201279in}}%
\pgfusepath{stroke}%
\end{pgfscope}%
\begin{pgfscope}%
\pgfpathrectangle{\pgfqpoint{0.100000in}{0.100000in}}{\pgfqpoint{3.696000in}{3.696000in}}%
\pgfusepath{clip}%
\pgfsetbuttcap%
\pgfsetroundjoin%
\pgfsetlinewidth{1.003750pt}%
\definecolor{currentstroke}{rgb}{0.85,0.235,0.098}%
\pgfsetstrokecolor{currentstroke}%
\pgfsetdash{{6.400000pt}{1.600000pt}{1.000000pt}{1.600000pt}}{0.000000pt}%
\pgfpathmoveto{\pgfqpoint{3.302270in}{2.201279in}}%
\pgfusepath{stroke}%
\end{pgfscope}%
\begin{pgfscope}%
\pgfpathrectangle{\pgfqpoint{0.100000in}{0.100000in}}{\pgfqpoint{3.696000in}{3.696000in}}%
\pgfusepath{clip}%
\pgfsetbuttcap%
\pgfsetroundjoin%
\definecolor{currentfill}{rgb}{0.85,0.235,0.098}%
\pgfsetfillcolor{currentfill}%
\pgfsetlinewidth{1.003750pt}%
\definecolor{currentstroke}{rgb}{0.85,0.235,0.098}%
\pgfsetstrokecolor{currentstroke}%
\pgfsetdash{}{0pt}%
\pgfsys@defobject{currentmarker}{\pgfqpoint{-0.041667in}{-0.041667in}}{\pgfqpoint{0.041667in}{0.041667in}}{%
\pgfpathmoveto{\pgfqpoint{-0.041667in}{-0.041667in}}%
\pgfpathlineto{\pgfqpoint{0.041667in}{0.041667in}}%
\pgfpathmoveto{\pgfqpoint{-0.041667in}{0.041667in}}%
\pgfpathlineto{\pgfqpoint{0.041667in}{-0.041667in}}%
\pgfusepath{stroke,fill}%
}%
\begin{pgfscope}%
\pgfsys@transformshift{3.302270in}{2.201279in}%
\pgfsys@useobject{currentmarker}{}%
\end{pgfscope}%
\end{pgfscope}%
\begin{pgfscope}%
\pgfpathrectangle{\pgfqpoint{0.100000in}{0.100000in}}{\pgfqpoint{3.696000in}{3.696000in}}%
\pgfusepath{clip}%
\pgfsetbuttcap%
\pgfsetroundjoin%
\pgfsetlinewidth{1.003750pt}%
\definecolor{currentstroke}{rgb}{0.85,0.235,0.098}%
\pgfsetstrokecolor{currentstroke}%
\pgfsetdash{{6.400000pt}{1.600000pt}{1.000000pt}{1.600000pt}}{0.000000pt}%
\pgfpathmoveto{\pgfqpoint{2.214348in}{2.024437in}}%
\pgfusepath{stroke}%
\end{pgfscope}%
\begin{pgfscope}%
\pgfpathrectangle{\pgfqpoint{0.100000in}{0.100000in}}{\pgfqpoint{3.696000in}{3.696000in}}%
\pgfusepath{clip}%
\pgfsetbuttcap%
\pgfsetroundjoin%
\definecolor{currentfill}{rgb}{0.85,0.235,0.098}%
\pgfsetfillcolor{currentfill}%
\pgfsetlinewidth{1.003750pt}%
\definecolor{currentstroke}{rgb}{0.85,0.235,0.098}%
\pgfsetstrokecolor{currentstroke}%
\pgfsetdash{}{0pt}%
\pgfsys@defobject{currentmarker}{\pgfqpoint{-0.034722in}{-0.034722in}}{\pgfqpoint{0.034722in}{0.034722in}}{%
\pgfpathmoveto{\pgfqpoint{0.000000in}{-0.034722in}}%
\pgfpathcurveto{\pgfqpoint{0.009208in}{-0.034722in}}{\pgfqpoint{0.018041in}{-0.031064in}}{\pgfqpoint{0.024552in}{-0.024552in}}%
\pgfpathcurveto{\pgfqpoint{0.031064in}{-0.018041in}}{\pgfqpoint{0.034722in}{-0.009208in}}{\pgfqpoint{0.034722in}{0.000000in}}%
\pgfpathcurveto{\pgfqpoint{0.034722in}{0.009208in}}{\pgfqpoint{0.031064in}{0.018041in}}{\pgfqpoint{0.024552in}{0.024552in}}%
\pgfpathcurveto{\pgfqpoint{0.018041in}{0.031064in}}{\pgfqpoint{0.009208in}{0.034722in}}{\pgfqpoint{0.000000in}{0.034722in}}%
\pgfpathcurveto{\pgfqpoint{-0.009208in}{0.034722in}}{\pgfqpoint{-0.018041in}{0.031064in}}{\pgfqpoint{-0.024552in}{0.024552in}}%
\pgfpathcurveto{\pgfqpoint{-0.031064in}{0.018041in}}{\pgfqpoint{-0.034722in}{0.009208in}}{\pgfqpoint{-0.034722in}{0.000000in}}%
\pgfpathcurveto{\pgfqpoint{-0.034722in}{-0.009208in}}{\pgfqpoint{-0.031064in}{-0.018041in}}{\pgfqpoint{-0.024552in}{-0.024552in}}%
\pgfpathcurveto{\pgfqpoint{-0.018041in}{-0.031064in}}{\pgfqpoint{-0.009208in}{-0.034722in}}{\pgfqpoint{0.000000in}{-0.034722in}}%
\pgfpathlineto{\pgfqpoint{0.000000in}{-0.034722in}}%
\pgfpathclose%
\pgfusepath{stroke,fill}%
}%
\begin{pgfscope}%
\pgfsys@transformshift{2.214348in}{2.024437in}%
\pgfsys@useobject{currentmarker}{}%
\end{pgfscope}%
\end{pgfscope}%
\begin{pgfscope}%
\pgfpathrectangle{\pgfqpoint{0.100000in}{0.100000in}}{\pgfqpoint{3.696000in}{3.696000in}}%
\pgfusepath{clip}%
\pgfsetbuttcap%
\pgfsetroundjoin%
\pgfsetlinewidth{1.003750pt}%
\definecolor{currentstroke}{rgb}{0.466,0.674,0.188}%
\pgfsetstrokecolor{currentstroke}%
\pgfsetdash{{3.700000pt}{1.600000pt}}{0.000000pt}%
\pgfpathmoveto{\pgfqpoint{2.217555in}{2.838360in}}%
\pgfpathlineto{\pgfqpoint{2.215419in}{2.027600in}}%
\pgfpathlineto{\pgfqpoint{2.244925in}{1.973944in}}%
\pgfpathlineto{\pgfqpoint{2.270739in}{1.923223in}}%
\pgfpathlineto{\pgfqpoint{2.290717in}{1.879381in}}%
\pgfpathlineto{\pgfqpoint{2.304131in}{1.845582in}}%
\pgfpathlineto{\pgfqpoint{2.315958in}{1.809720in}}%
\pgfpathlineto{\pgfqpoint{2.322898in}{1.782506in}}%
\pgfpathlineto{\pgfqpoint{2.327610in}{1.755834in}}%
\pgfpathlineto{\pgfqpoint{2.329757in}{1.730353in}}%
\pgfpathlineto{\pgfqpoint{2.329197in}{1.706585in}}%
\pgfpathlineto{\pgfqpoint{2.325979in}{1.684903in}}%
\pgfpathlineto{\pgfqpoint{2.320318in}{1.665533in}}%
\pgfpathlineto{\pgfqpoint{2.312551in}{1.648561in}}%
\pgfpathlineto{\pgfqpoint{2.303090in}{1.633965in}}%
\pgfpathlineto{\pgfqpoint{2.292376in}{1.621633in}}%
\pgfpathlineto{\pgfqpoint{2.280842in}{1.611396in}}%
\pgfpathlineto{\pgfqpoint{2.264868in}{1.600643in}}%
\pgfpathlineto{\pgfqpoint{2.248926in}{1.592714in}}%
\pgfpathlineto{\pgfqpoint{2.229928in}{1.585976in}}%
\pgfpathlineto{\pgfqpoint{2.209420in}{1.581338in}}%
\pgfpathlineto{\pgfqpoint{2.189307in}{1.579065in}}%
\pgfpathlineto{\pgfqpoint{2.167588in}{1.578982in}}%
\pgfpathlineto{\pgfqpoint{2.147290in}{1.581260in}}%
\pgfpathlineto{\pgfqpoint{2.130060in}{1.585455in}}%
\pgfpathlineto{\pgfqpoint{2.125458in}{1.587101in}}%
\pgfpathlineto{\pgfqpoint{2.125458in}{1.587101in}}%
\pgfusepath{stroke}%
\end{pgfscope}%
\begin{pgfscope}%
\pgfpathrectangle{\pgfqpoint{0.100000in}{0.100000in}}{\pgfqpoint{3.696000in}{3.696000in}}%
\pgfusepath{clip}%
\pgfsetbuttcap%
\pgfsetroundjoin%
\pgfsetlinewidth{1.003750pt}%
\definecolor{currentstroke}{rgb}{0.466,0.674,0.188}%
\pgfsetstrokecolor{currentstroke}%
\pgfsetdash{{3.700000pt}{1.600000pt}}{0.000000pt}%
\pgfpathmoveto{\pgfqpoint{2.125458in}{1.587101in}}%
\pgfusepath{stroke}%
\end{pgfscope}%
\begin{pgfscope}%
\pgfpathrectangle{\pgfqpoint{0.100000in}{0.100000in}}{\pgfqpoint{3.696000in}{3.696000in}}%
\pgfusepath{clip}%
\pgfsetbuttcap%
\pgfsetroundjoin%
\definecolor{currentfill}{rgb}{0.466,0.674,0.188}%
\pgfsetfillcolor{currentfill}%
\pgfsetlinewidth{1.003750pt}%
\definecolor{currentstroke}{rgb}{0.466,0.674,0.188}%
\pgfsetstrokecolor{currentstroke}%
\pgfsetdash{}{0pt}%
\pgfsys@defobject{currentmarker}{\pgfqpoint{-0.041667in}{-0.041667in}}{\pgfqpoint{0.041667in}{0.041667in}}{%
\pgfpathmoveto{\pgfqpoint{-0.041667in}{-0.041667in}}%
\pgfpathlineto{\pgfqpoint{0.041667in}{0.041667in}}%
\pgfpathmoveto{\pgfqpoint{-0.041667in}{0.041667in}}%
\pgfpathlineto{\pgfqpoint{0.041667in}{-0.041667in}}%
\pgfusepath{stroke,fill}%
}%
\begin{pgfscope}%
\pgfsys@transformshift{2.125458in}{1.587101in}%
\pgfsys@useobject{currentmarker}{}%
\end{pgfscope}%
\end{pgfscope}%
\begin{pgfscope}%
\pgfpathrectangle{\pgfqpoint{0.100000in}{0.100000in}}{\pgfqpoint{3.696000in}{3.696000in}}%
\pgfusepath{clip}%
\pgfsetbuttcap%
\pgfsetroundjoin%
\pgfsetlinewidth{1.003750pt}%
\definecolor{currentstroke}{rgb}{0.466,0.674,0.188}%
\pgfsetstrokecolor{currentstroke}%
\pgfsetdash{{3.700000pt}{1.600000pt}}{0.000000pt}%
\pgfpathmoveto{\pgfqpoint{2.217555in}{2.838360in}}%
\pgfusepath{stroke}%
\end{pgfscope}%
\begin{pgfscope}%
\pgfpathrectangle{\pgfqpoint{0.100000in}{0.100000in}}{\pgfqpoint{3.696000in}{3.696000in}}%
\pgfusepath{clip}%
\pgfsetbuttcap%
\pgfsetroundjoin%
\definecolor{currentfill}{rgb}{0.466,0.674,0.188}%
\pgfsetfillcolor{currentfill}%
\pgfsetlinewidth{1.003750pt}%
\definecolor{currentstroke}{rgb}{0.466,0.674,0.188}%
\pgfsetstrokecolor{currentstroke}%
\pgfsetdash{}{0pt}%
\pgfsys@defobject{currentmarker}{\pgfqpoint{-0.034722in}{-0.034722in}}{\pgfqpoint{0.034722in}{0.034722in}}{%
\pgfpathmoveto{\pgfqpoint{0.000000in}{-0.034722in}}%
\pgfpathcurveto{\pgfqpoint{0.009208in}{-0.034722in}}{\pgfqpoint{0.018041in}{-0.031064in}}{\pgfqpoint{0.024552in}{-0.024552in}}%
\pgfpathcurveto{\pgfqpoint{0.031064in}{-0.018041in}}{\pgfqpoint{0.034722in}{-0.009208in}}{\pgfqpoint{0.034722in}{0.000000in}}%
\pgfpathcurveto{\pgfqpoint{0.034722in}{0.009208in}}{\pgfqpoint{0.031064in}{0.018041in}}{\pgfqpoint{0.024552in}{0.024552in}}%
\pgfpathcurveto{\pgfqpoint{0.018041in}{0.031064in}}{\pgfqpoint{0.009208in}{0.034722in}}{\pgfqpoint{0.000000in}{0.034722in}}%
\pgfpathcurveto{\pgfqpoint{-0.009208in}{0.034722in}}{\pgfqpoint{-0.018041in}{0.031064in}}{\pgfqpoint{-0.024552in}{0.024552in}}%
\pgfpathcurveto{\pgfqpoint{-0.031064in}{0.018041in}}{\pgfqpoint{-0.034722in}{0.009208in}}{\pgfqpoint{-0.034722in}{0.000000in}}%
\pgfpathcurveto{\pgfqpoint{-0.034722in}{-0.009208in}}{\pgfqpoint{-0.031064in}{-0.018041in}}{\pgfqpoint{-0.024552in}{-0.024552in}}%
\pgfpathcurveto{\pgfqpoint{-0.018041in}{-0.031064in}}{\pgfqpoint{-0.009208in}{-0.034722in}}{\pgfqpoint{0.000000in}{-0.034722in}}%
\pgfpathlineto{\pgfqpoint{0.000000in}{-0.034722in}}%
\pgfpathclose%
\pgfusepath{stroke,fill}%
}%
\begin{pgfscope}%
\pgfsys@transformshift{2.217555in}{2.838360in}%
\pgfsys@useobject{currentmarker}{}%
\end{pgfscope}%
\end{pgfscope}%
%===================================================================================================
% Legend:
% Legend: Box
\begin{pgfscope}%
\pgfsetbuttcap%
\pgfsetmiterjoin%
\definecolor{currentfill}{rgb}{1.000000,1.000000,1.000000}%
\pgfsetfillcolor{currentfill}%
\pgfsetfillopacity{0.800000}%
\pgfsetlinewidth{1.003750pt}%
\definecolor{currentstroke}{rgb}{0.800000,0.800000,0.800000}%
\pgfsetstrokecolor{currentstroke}%
\pgfsetstrokeopacity{0.800000}%
\pgfsetdash{}{0pt}%
\pgfpathmoveto{\pgfqpoint{0.772567in}{1.103890in}}%
\pgfpathlineto{\pgfqpoint{1.698778in}{1.103890in}}%
\pgfpathquadraticcurveto{\pgfqpoint{1.726556in}{1.103890in}}{\pgfqpoint{1.726556in}{1.131668in}}%
\pgfpathlineto{\pgfqpoint{1.726556in}{1.698778in}}%
\pgfpathquadraticcurveto{\pgfqpoint{1.726556in}{1.726556in}}{\pgfqpoint{1.698778in}{1.726556in}}%
\pgfpathlineto{\pgfqpoint{0.772567in}{1.726556in}}%
\pgfpathquadraticcurveto{\pgfqpoint{1.744789in}{1.726556in}}{\pgfqpoint{0.744789in}{1.698778in}}%
\pgfpathlineto{\pgfqpoint{0.744789in}{1.131668in}}%
\pgfpathquadraticcurveto{\pgfqpoint{1.744789in}{1.103890in}}{\pgfqpoint{0.772567in}{1.103890in}}%
\pgfpathlineto{\pgfqpoint{0.772567in}{1.103890in}}%
\pgfpathclose%
\pgfusepath{stroke,fill}%
\end{pgfscope}%
% Legend: Agent 1 line
\begin{pgfscope}%
\pgfsetrectcap%
\pgfsetroundjoin%
\pgfsetlinewidth{1.003750pt}%
\definecolor{currentstroke}{rgb}{0,0.447,0.741}%
\pgfsetstrokecolor{currentstroke}%
\pgfsetdash{}{0pt}%
\pgfpathmoveto{\pgfqpoint{0.800344in}{1.622389in}}%
\pgfpathlineto{\pgfqpoint{0.939233in}{1.622389in}}%
\pgfpathlineto{\pgfqpoint{1.078122in}{1.622389in}}%
\pgfusepath{stroke}%
\end{pgfscope}%
% Legend: Agent 1 label
\begin{pgfscope}%
\definecolor{textcolor}{rgb}{0.000000,0.000000,0.000000}%
\pgfsetstrokecolor{textcolor}%
\pgfsetfillcolor{textcolor}%
\pgftext[x=1.189233in,y=1.573778in,left,base]{\color{textcolor}\rmfamily\fontsize{10.000000}{12.000000}\selectfont Agent \(\displaystyle 1\)}%
\end{pgfscope}%
% Legend: Agent 2 line
\begin{pgfscope}%
\pgfsetbuttcap%
\pgfsetroundjoin%
\pgfsetlinewidth{1.003750pt}%
\definecolor{currentstroke}{rgb}{0.85,0.235,0.098}%
\pgfsetstrokecolor{currentstroke}%
\pgfsetdash{{6.400000pt}{1.600000pt}{1.000000pt}{1.600000pt}}{0.000000pt}%
\pgfpathmoveto{\pgfqpoint{0.800344in}{1.428723in}}%
\pgfpathlineto{\pgfqpoint{0.939233in}{1.428723in}}%
\pgfpathlineto{\pgfqpoint{1.078122in}{1.428723in}}%
\pgfusepath{stroke}%
\end{pgfscope}%
% Legend: Agent 2 label
\begin{pgfscope}%
\definecolor{textcolor}{rgb}{0.000000,0.000000,0.000000}%
\pgfsetstrokecolor{textcolor}%
\pgfsetfillcolor{textcolor}%
\pgftext[x=1.189233in,y=1.380112in,left,base]{\color{textcolor}\rmfamily\fontsize{10.000000}{12.000000}\selectfont Agent \(\displaystyle 2\)}%
\end{pgfscope}%
% Legend: Agent 3 line
\begin{pgfscope}%
\pgfsetbuttcap%
\pgfsetroundjoin%
\pgfsetlinewidth{1.003750pt}%
\definecolor{currentstroke}{rgb}{0.466,0.674,0.188}%
\pgfsetstrokecolor{currentstroke}%
\pgfsetdash{{3.700000pt}{1.600000pt}}{0.000000pt}%
\pgfpathmoveto{\pgfqpoint{0.800344in}{1.235056in}}%
\pgfpathlineto{\pgfqpoint{0.939233in}{1.235056in}}%
\pgfpathlineto{\pgfqpoint{1.078122in}{1.235056in}}%
\pgfusepath{stroke}%
\end{pgfscope}%
% Legend: Agent 3 label
\begin{pgfscope}%
\definecolor{textcolor}{rgb}{0.000000,0.000000,0.000000}%
\pgfsetstrokecolor{textcolor}%
\pgfsetfillcolor{textcolor}%
\pgftext[x=1.189233in,y=1.186445in,left,base]{\color{textcolor}\rmfamily\fontsize{10.000000}{12.000000}\selectfont Agent \(\displaystyle 3\)}%
\end{pgfscope}%
\end{pgfpicture}%
\makeatother%
\endgroup%

%% file: graphs.tex
\begin{tikzpicture}[
    =1.5]
    % \draw[step=0.5, gray, thin] (0, 0) grid (8, 4);
    % \fill[red] (0, 0) circle (0.5pt);
    %
    \node[circle, draw] (agent1) at (0.5, 1.5) {1};
    \node[circle, draw] (agent4) at (0.5, 0.2) {4};
    \draw (agent1) -- (agent4);
    \node[circle, draw] (agent2) at (2.5, 1.5) {2};
    \draw (agent2) -- (agent1);
    \node[circle, draw] (agent3) at (2.5, 0.2) {3};
    \draw (agent3) -- (agent4);
    \node[circle, draw] (agent11) at (4, 1.5) {1};
    \node[circle, draw] (agent44) at (4, 0.2) {4};
    \draw (agent11) -- (agent44);
    \node[circle, draw] (agent22) at (6, 1.5) {2};
    \node[circle, draw] (agent33) at (6, 0.2) {3};
    \draw (agent33) -- (agent44);
    
    \draw (agent22) -- (agent33);
    \node[circle, draw] (agent55) at (5, 1) {5};
    
    \draw (agent55) -- (agent44);
    \draw (agent33) -- (agent55);
\end{tikzpicture}

%% file: consensus_x1_manual.tex
% Change size of legend entries to footnotesize.
\begin{tikzpicture}

\definecolor{crimson2143940}{RGB}{214,39,40}
\definecolor{darkgray176}{RGB}{176,176,176}
\definecolor{darkorange25512714}{RGB}{255,127,14}
\definecolor{forestgreen4416044}{RGB}{44,160,44}
\definecolor{lightgray204}{RGB}{204,204,204}
\definecolor{mediumpurple148103189}{RGB}{148,103,189}
\definecolor{steelblue31119180}{RGB}{31,119,180}

\begin{axis}[
height=\axisheight,
legend cell align={left},
legend style={fill opacity=0.8, draw opacity=1, text opacity=1, draw=lightgray204},
minor tick num=2,
minor xtick={-1,1,2,3,4,6,7,8,9,11,12,13,14,16,17,18,19,21,22,23,24,26,27,28,29,31,32,33,34,36,37,38,39,41,42},
minor ytick={-2.2,-1.8,-1.6,-1.4,-1.2,-0.8,-0.600000000000001,-0.4,-0.2,0.199999999999999,0.399999999999999,0.6,0.799999999999999,1.2,1.4,1.6,1.8,2.2,2.4,2.6,2.8,3.2},
tick align=outside,
tick pos=left,
width=\axiswidth,
x grid style={darkgray176},
xmajorgrids,
xmin=-2, xmax=42,
xtick distance=10,
xtick style={color=black},
xtick={-5,0,5,10,15,20,25,30,35,40,45},
y grid style={darkgray176},
ylabel={\(\displaystyle x_{i,1}\)},
ymajorgrids,
ymin=-2.25, ymax=3.25,
ytick distance=0.2,
ytick style={color=black},
ytick={-3,-2,-1,0,1,2,3,4}
]
\addplot [thick, steelblue31119180]
table {%
0 -1
1 -1
2 -0.755395779484107
3 -0.521629151612041
4 -0.32984821964809
5 -0.175494336756065
6 -0.0560050145068443
7 0.0328676422811684
8 0.0988662338663227
9 0.149477398977528
10 0.190660276477376
11 0.225410269011194
12 0.254753866319491
13 0.278973431372997
14 0.298458005443829
15 0.313831844161424
16 0.325816537652212
17 0.33510371848072
18 0.342291875202755
19 0.347869061701015
20 0.352219000111249
21 0.355636463791838
22 0.360091826642905
23 0.345830185428188
24 0.33416774238222
25 0.323902612565162
26 0.31497312931404
27 0.307521717089985
28 0.301505459298685
29 0.296754167555881
30 0.293058463546979
31 0.290215554732902
32 0.288047090341809
33 0.286404036724803
34 0.285165790884554
35 0.284236859306615
36 0.283542777266934
37 0.283026112030123
38 0.282642923468004
39 0.282359798746051
40 0.282151450520185
};
\addlegendentry{\footnotesize Agent 1}
\addplot [thick, darkorange25512714, dotted]
table {%
0 2
1 2
2 1.74999999300975
3 1.49999998496034
4 1.24999997638245
5 1.00000015208219
6 0.804104361924818
7 0.65925757369116
8 0.553880547213635
9 0.47829051180758
10 0.425392158878422
11 0.38956745561856
12 0.36627707275996
13 0.351911505461114
14 0.34374208183602
15 0.339773648272676
16 0.338575697315939
17 0.339139224552324
18 0.340764692386504
19 0.34297630119324
20 0.345457524310055
21 0.348003527580185
22 0.357143479029148
23 0.348630876152348
24 0.338184905420051
25 0.328504755604659
26 0.319680363699741
27 0.311990218964955
28 0.305561042298417
29 0.300335389963096
30 0.296165881335067
31 0.29288095397137
32 0.290315970356433
33 0.288325773509909
34 0.286788345147052
35 0.285604191187057
36 0.284693829228723
37 0.283994656745592
38 0.283457854827852
39 0.283045621379484
40 0.282728831444349
};
\addlegendentry{\footnotesize Agent 2}
\addplot [thick, forestgreen4416044, dash pattern=on 1pt off 3pt on 3pt off 3pt]
table {%
0 3
1 3
2 2.74999999045562
3 2.49999998087327
4 2.24999997149976
5 1.99999996167458
6 1.7499999531201
7 1.49999994399346
8 1.24999994019792
9 1.01513082692168
10 0.838509431856992
11 0.712415317724158
12 0.62371767559853
13 0.561188415486135
14 0.516639464680813
15 0.484443287035556
16 0.460799715235358
17 0.4431436845873
18 0.42973366076113
19 0.41937742602689
20 0.411250914244064
21 0.404778367062688
22 0.359490197532871
23 0.32374266546166
24 0.302914471758438
25 0.291314959271386
26 0.285051066665947
27 0.28188296638802
28 0.280457312776269
29 0.279961322717961
30 0.27992772233859
31 0.280100037477465
32 0.280343775609239
33 0.280592469641524
34 0.280816500540857
35 0.281005525775977
36 0.281158747217889
37 0.281279636003321
38 0.281373167891791
39 0.281444451744054
40 0.28149811383802
};
\addlegendentry{\footnotesize Agent 3}
\addplot [thick, crimson2143940, dashed]
table {%
0 -2
1 -2
2 -1.74999999026939
3 -1.49999998063923
4 -1.24999997074422
5 -0.999999961676214
6 -0.749999953068218
7 -0.499999945366214
8 -0.249999942498941
9 -0.0326777225641873
10 0.118443482813279
11 0.217008444225522
12 0.279928237471588
13 0.319658813444295
14 0.344494839015049
15 0.359812391191461
16 0.369065571459966
17 0.374468187917924
18 0.37743946891811
19 0.378890098836883
20 0.379404587538742
21 0.3793575774157
22 0.334478506232459
23 0.30701906601671
24 0.291971946966193
25 0.284000577671107
26 0.280103189332496
27 0.27850070702342
28 0.278108526335986
29 0.278294780503423
30 0.278714874850069
31 0.27919319881952
32 0.279647509530147
33 0.280044767422065
34 0.280376634870321
35 0.280646294612165
36 0.280861560809819
37 0.281031426512436
38 0.281164454168403
39 0.281268128634328
40 0.281348692061499
};
\addlegendentry{\footnotesize Agent 4}
\addplot [thick, mediumpurple148103189, dash pattern=on 3pt off 1pt on 1pt off 1pt on 1pt off 1pt]
table {%
0 0
1 0
2 0
3 0
4 0
5 0
6 0
7 0
8 0
9 0
10 0
11 0
12 0
13 0
14 0
15 0
16 0
17 0
18 0
19 0
20 0
21 2.08448641096549e-37
22 0.0702824323261579
23 0.134144818711088
24 0.180265588125721
25 0.211551943802292
26 0.23260425672023
27 0.24691064225779
28 0.256769071483239
29 0.263657053041329
30 0.268531347898642
31 0.2720199991278
32 0.274541460786155
33 0.276378905484745
34 0.277726961789156
35 0.278721390116492
36 0.279458164041008
37 0.28000592986808
38 0.280414280426967
39 0.280719343393222
40 0.280947617832453
};
\addlegendentry{\footnotesize Agent 5}
\end{axis}

\end{tikzpicture}

%% file: consensus_x2_manual.tex
% Change size of legend entries to footnotesize.
\begin{tikzpicture}

\definecolor{crimson2143940}{RGB}{214,39,40}
\definecolor{darkgray176}{RGB}{176,176,176}
\definecolor{darkorange25512714}{RGB}{255,127,14}
\definecolor{forestgreen4416044}{RGB}{44,160,44}
\definecolor{lightgray204}{RGB}{204,204,204}
\definecolor{mediumpurple148103189}{RGB}{148,103,189}
\definecolor{steelblue31119180}{RGB}{31,119,180}

\begin{axis}[
height=\axisheight,
legend cell align={left},
legend style={fill opacity=0.8, draw opacity=1, text opacity=1, draw=lightgray204},
minor tick num=2,
minor xtick={-1,1,2,3,4,6,7,8,9,11,12,13,14,16,17,18,19,21,22,23,24,26,27,28,29,31,32,33,34,36,37,38,39,41,42},
minor ytick={-2.2,-1.8,-1.6,-1.4,-1.2,-0.8,-0.600000000000001,-0.4,-0.2,0.199999999999999,0.399999999999999,0.6,0.799999999999999,1.2,1.4,1.6,1.8,2.2,2.4,2.6,2.8,3.2,3.4,3.6,3.8,4.2},
tick align=outside,
tick pos=left,
width=\axiswidth,
x grid style={darkgray176},
xlabel={\(\displaystyle t\) (time steps)},
xmajorgrids,
xmin=-2, xmax=42,
xtick distance=10,
xtick style={color=black},
xtick={-5,0,5,10,15,20,25,30,35,40,45},
y grid style={darkgray176},
ylabel={\(\displaystyle x_{i,2}\)},
ymajorgrids,
ymin=-2.3, ymax=4.3,
ytick distance=0.2,
ytick style={color=black},
ytick={-3,-2,-1,0,1,2,3,4,5}
]
\addplot [thick, steelblue31119180]
table {%
0 4
1 4
2 3.74999999042473
3 3.49999998061534
4 3.24999997131904
5 2.99999996226599
6 2.74999995279296
7 2.4999999437503
8 2.24999994784223
9 2.03875193755962
10 1.88395670145372
11 1.77252303788485
12 1.69122405198775
13 1.63064041232951
14 1.58444909677035
15 1.54844255493609
16 1.51980664280901
17 1.49663449688271
18 1.4776099868216
19 1.46180484138483
20 1.44854898859572
21 1.43734721215084
22 1.41495328201562
23 1.24968383463901
24 1.0984303061127
25 0.963990118994912
26 0.843520690674519
27 0.735622641131951
28 0.641308425536834
29 0.561150966835779
30 0.495121688843148
31 0.442112600349156
32 0.400365528346471
33 0.367890130630763
34 0.342842681566755
35 0.3236492540314
36 0.309015609219561
37 0.297901882726989
38 0.289486789238247
39 0.283130011709557
40 0.278337076824124
};
\addlegendentry{\footnotesize Agent 1}
\addplot [thick, darkorange25512714, dotted]
table {%
0 1.8
1 1.8
2 1.88500000438946
3 1.94875000755298
4 1.98168751170103
5 1.9952343858108
6 1.9996304812161
7 2.00055684097265
8 1.99371538743869
9 1.96918837278026
10 1.9304454638687
11 1.88351663395549
12 1.83359246744309
13 1.78414350025388
14 1.73719257783444
15 1.69378334682205
16 1.6543481288268
17 1.618953024326
18 1.58745386209719
19 1.55959375044671
20 1.53506324762382
21 1.51353677578989
22 1.46218377861693
23 1.30225907172149
24 1.13915244884653
25 0.995100183136813
26 0.868228862373277
27 0.756404814069108
28 0.659372043904482
29 0.577023055881943
30 0.508958883668848
31 0.454017169193515
32 0.410473856295288
33 0.376385964379935
34 0.349925614653111
35 0.329516932714439
36 0.313854011499549
37 0.301879231738672
38 0.292750514574149
39 0.285806245171978
40 0.280531833113138
};
\addlegendentry{\footnotesize Agent 2}
\addplot [thick, forestgreen4416044, dash pattern=on 1pt off 3pt on 3pt off 3pt]
table {%
0 -1.5
1 -1.5
2 -1.25000006005228
3 -1.00000005271286
4 -0.750000046223525
5 -0.500000038231121
6 -0.250002082926578
7 -0.0105677870198521
8 0.197087315150498
9 0.373842879934707
10 0.523239422916116
11 0.648885431118817
12 0.754453770498481
13 0.843309475942278
14 0.918310231159711
15 0.981817363271307
16 1.03576631096453
17 1.08173901730225
18 1.12102827138089
19 1.15469283656738
20 1.18360327037632
21 1.20847905530654
22 0.983354835403826
23 0.733354854103951
24 0.556255440832854
25 0.442556788658922
26 0.366254685131644
27 0.3133298035995
28 0.277493281372998
29 0.255668604991744
30 0.244680614137959
31 0.24103949059862
32 0.241583614968593
33 0.244079672903399
34 0.247217035201802
35 0.250323624406596
36 0.253100301598772
37 0.255447098491176
38 0.257363133530245
39 0.258891826094517
40 0.260091960553403
};
\addlegendentry{\footnotesize Agent 3}
\addplot [thick, crimson2143940, dashed]
table {%
0 0
1 4.01959415208098e-33
2 0.249999939886475
3 0.490330075142765
4 0.664779948131877
5 0.789475438726556
6 0.884385367719881
7 0.960136807255415
8 1.02199954198909
9 1.07215545616288
10 1.11216431424802
11 1.14448592699601
12 1.17147126124734
13 1.19472323696347
14 1.21517641237298
15 1.23335502114823
16 1.24956877528666
17 1.26402480265964
18 1.27688438005825
19 1.28828915001329
20 1.2983717175164
21 1.30725882633373
22 1.0661459267067
23 0.81614593892497
24 0.606657725863637
25 0.465504213142688
26 0.371870138201531
27 0.31020085216568
28 0.271065559541671
29 0.249004676233812
30 0.239077234962686
31 0.236788471908831
32 0.238512307063306
33 0.241904934799267
34 0.245681644236449
35 0.249228487289522
36 0.252302386867713
37 0.254848422405538
38 0.256898827751844
39 0.258519948054361
40 0.25978574472874
};
\addlegendentry{\footnotesize Agent 4}
\addplot [thick, mediumpurple148103189, dash pattern=on 3pt off 1pt on 1pt off 1pt on 1pt off 1pt]
table {%
0 -2
1 -2
2 -2
3 -2
4 -2
5 -2
6 -2
7 -2
8 -2
9 -2
10 -2
11 -2
12 -2
13 -2
14 -2
15 -2
16 -2
17 -2
18 -2
19 -2
20 -2
21 -2
22 -1.74999999070432
23 -1.49999998124439
24 -1.24999997212996
25 -0.999999963524316
26 -0.749999954398116
27 -0.499999951796539
28 -0.282282830731541
29 -0.124553528143585
30 -0.0141167617395606
31 0.0632910112750571
32 0.118024983185928
33 0.157106598270558
34 0.185272820695173
35 0.205742039083515
36 0.220723880995456
37 0.231754115127957
38 0.239913769016258
39 0.245972876116793
40 0.250485716483043
};
\addlegendentry{\footnotesize Agent 5}
\end{axis}

\end{tikzpicture}

%% file: consensus_pos_manual.tex
\begin{tikzpicture}
% Manually move the legend.
% Increase line width.
% Remove superfluous black marks.
% Changed size of legend entries to footnotesize.

\definecolor{black}{RGB}{25,25,25}
\definecolor{red}{RGB}{159,20,47}
\definecolor{grey}{RGB}{176,176,176}
\definecolor{orange}{RGB}{217,83,25}
\definecolor{green}{RGB}{119,172,48}
\definecolor{lightgrey}{RGB}{204,204,204}
\definecolor{purple}{RGB}{126,47,142}
\definecolor{blue}{RGB}{0,114,189}

\begin{axis}[
height=\axisheight,
legend cell align={left},
legend style={
  fill opacity=0.8,
  draw opacity=1,
  text opacity=1,
  at={(0.71,0.99)},
  anchor=north west,
  draw=lightgray
},
minor tick num=2,
minor xtick={},
minor ytick={},
tick align=outside,
tick pos=left,
width=\axiswidth,
x grid style={grey},
xlabel={\(\displaystyle x_{i,1}\)},
xmajorgrids,
xmin=-3.35, xmax=4.35,
xminorgrids,
xtick distance=10,
xtick style={color=black},
xtick={-4,-3,-2,-1,0,1,2,3,4,5},
y grid style={grey},
ylabel={\(\displaystyle x_{i,2}\)},
ymajorgrids,
ymin=-3.35, ymax=4.35,
yminorgrids,
ytick distance=0.2,
ytick style={color=black},
ytick={-4,-3,-2,-1,0,1,2,3,4,5}
]
\path [draw=black, thick]
(axis cs:1,-2)
--(axis cs:1,4)
--(axis cs:-1,4)
--(axis cs:-1,-2)
--cycle;
\path [draw=black, thick]
(axis cs:4,-2)
--(axis cs:4,2)
--(axis cs:-1,2)
--(axis cs:-1,-2)
--cycle;
\path [draw=black, thick]
(axis cs:3,0)
--(axis cs:0,3)
--(axis cs:-3,0)
--(axis cs:0,-3)
--cycle;
\addplot [very thick, blue]
table {%
-1 4
-0.755395779484107 3.74999999042473
-0.521629151612041 3.49999998061534
-0.32984821964809 3.24999997131904
-0.175494336756065 2.99999996226599
-0.0560050145068443 2.74999995279296
0.0328676422811684 2.4999999437503
0.0988662338663227 2.24999994784223
0.149477398977528 2.03875193755962
0.190660276477376 1.88395670145372
0.225410269011194 1.77252303788485
0.254753866319491 1.69122405198775
0.278973431372997 1.63064041232951
0.298458005443829 1.58444909677035
0.313831844161424 1.54844255493609
0.325816537652212 1.51980664280901
0.33510371848072 1.49663449688271
0.342291875202755 1.4776099868216
0.347869061701015 1.46180484138483
0.352219000111249 1.44854898859572
0.355636463791838 1.43734721215084
0.360091826642905 1.41495328201562
0.345830185428188 1.24968383463901
0.33416774238222 1.0984303061127
0.323902612565162 0.963990118994912
0.31497312931404 0.843520690674519
0.307521717089985 0.735622641131951
0.301505459298685 0.641308425536834
0.296754167555881 0.561150966835779
0.293058463546979 0.495121688843148
0.290215554732902 0.442112600349156
0.288047090341809 0.400365528346471
0.286404036724803 0.367890130630763
0.285165790884554 0.342842681566755
0.284236859306615 0.3236492540314
0.283542777266934 0.309015609219561
0.283026112030123 0.297901882726989
0.282642923468004 0.289486789238247
0.282359798746051 0.283130011709557
0.282151450520185 0.278337076824124
0.281998813902132 0.274728801333499
};
\addlegendentry{\footnotesize Agent 1}
\addplot [black, mark=x, mark size=3, mark options={solid}, forget plot]
table {%
0.281998813902132 0.274728801333499
};
\addplot [ultra thick, orange, dotted]
table {%
2 1.8
1.74999999300975 1.88500000438946
1.49999998496034 1.94875000755298
1.24999997638245 1.98168751170103
1.00000015208219 1.9952343858108
0.804104361924818 1.9996304812161
0.65925757369116 2.00055684097265
0.553880547213635 1.99371538743869
0.47829051180758 1.96918837278026
0.425392158878422 1.9304454638687
0.38956745561856 1.88351663395549
0.36627707275996 1.83359246744309
0.351911505461114 1.78414350025388
0.34374208183602 1.73719257783444
0.339773648272676 1.69378334682205
0.338575697315939 1.6543481288268
0.339139224552324 1.618953024326
0.340764692386504 1.58745386209719
0.34297630119324 1.55959375044671
0.345457524310055 1.53506324762382
0.348003527580185 1.51353677578989
0.357143479029148 1.46218377861693
0.348630876152348 1.30225907172149
0.338184905420051 1.13915244884653
0.328504755604659 0.995100183136813
0.319680363699741 0.868228862373277
0.311990218964955 0.756404814069108
0.305561042298417 0.659372043904482
0.300335389963096 0.577023055881943
0.296165881335067 0.508958883668848
0.29288095397137 0.454017169193515
0.290315970356433 0.410473856295288
0.288325773509909 0.376385964379935
0.286788345147052 0.349925614653111
0.285604191187057 0.329516932714439
0.284693829228723 0.313854011499549
0.283994656745592 0.301879231738672
0.283457854827852 0.292750514574149
0.283045621379484 0.285806245171978
0.282728831444349 0.280531833113138
0.282485124326775 0.276530054400088
};
\addlegendentry{\footnotesize Agent 2}
\addplot [very thick, green, dash pattern=on 1pt off 3pt on 3pt off 3pt]
table {%
3 -1.5
2.74999999045562 -1.25000006005228
2.49999998087327 -1.00000005271286
2.24999997149976 -0.750000046223525
1.99999996167458 -0.500000038231121
1.7499999531201 -0.250002082926578
1.49999994399346 -0.0105677870198521
1.24999994019792 0.197087315150498
1.01513082692168 0.373842879934707
0.838509431856992 0.523239422916116
0.712415317724158 0.648885431118817
0.62371767559853 0.754453770498481
0.561188415486135 0.843309475942278
0.516639464680813 0.918310231159711
0.484443287035556 0.981817363271307
0.460799715235358 1.03576631096453
0.4431436845873 1.08173901730225
0.42973366076113 1.12102827138089
0.41937742602689 1.15469283656738
0.411250914244064 1.18360327037632
0.404778367062688 1.20847905530654
0.359490197532871 0.983354835403826
0.32374266546166 0.733354854103951
0.302914471758438 0.556255440832854
0.291314959271386 0.442556788658922
0.285051066665947 0.366254685131644
0.28188296638802 0.3133298035995
0.280457312776269 0.277493281372998
0.279961322717961 0.255668604991744
0.27992772233859 0.244680614137959
0.280100037477465 0.24103949059862
0.280343775609239 0.241583614968593
0.280592469641524 0.244079672903399
0.280816500540857 0.247217035201802
0.281005525775977 0.250323624406596
0.281158747217889 0.253100301598772
0.281279636003321 0.255447098491176
0.281373167891791 0.257363133530245
0.281444451744054 0.258891826094517
0.28149811383802 0.260091960553403
0.281538078593401 0.261023197188321
};
\addlegendentry{\footnotesize Agent 3}
\addplot [very thick, red, dashed]
table {%
-2 4.01959415208098e-33
-1.74999999026939 0.249999939886475
-1.49999998063923 0.490330075142765
-1.24999997074422 0.664779948131877
-0.999999961676214 0.789475438726556
-0.749999953068218 0.884385367719881
-0.499999945366214 0.960136807255415
-0.249999942498941 1.02199954198909
-0.0326777225641873 1.07215545616288
0.118443482813279 1.11216431424802
0.217008444225522 1.14448592699601
0.279928237471588 1.17147126124734
0.319658813444295 1.19472323696347
0.344494839015049 1.21517641237298
0.359812391191461 1.23335502114823
0.369065571459966 1.24956877528666
0.374468187917924 1.26402480265964
0.37743946891811 1.27688438005825
0.378890098836883 1.28828915001329
0.379404587538742 1.2983717175164
0.3793575774157 1.30725882633373
0.334478506232459 1.0661459267067
0.30701906601671 0.81614593892497
0.291971946966193 0.606657725863637
0.284000577671107 0.465504213142688
0.280103189332496 0.371870138201531
0.27850070702342 0.31020085216568
0.278108526335986 0.271065559541671
0.278294780503423 0.249004676233812
0.278714874850069 0.239077234962686
0.27919319881952 0.236788471908831
0.279647509530147 0.238512307063306
0.280044767422065 0.241904934799267
0.280376634870321 0.245681644236449
0.280646294612165 0.249228487289522
0.280861560809819 0.252302386867713
0.281031426512436 0.254848422405538
0.281164454168403 0.256898827751844
0.281268128634328 0.258519948054361
0.281348692061499 0.25978574472874
0.281411203159034 0.260765588018302
};
\addlegendentry{\footnotesize Agent 4}
\addplot [very thick, purple, dash pattern=on 3pt off 1pt on 1pt off 1pt on 1pt off 1pt]
table {%
0 -2
0 -2
0 -2
0 -2
0 -2
0 -2
0 -2
0 -2
0 -2
0 -2
0 -2
0 -2
0 -2
0 -2
0 -2
0 -2
0 -2
0 -2
0 -2
0 -2
2.08448641096549e-37 -2
0.0702824323261579 -1.74999999070432
0.134144818711088 -1.49999998124439
0.180265588125721 -1.24999997212996
0.211551943802292 -0.999999963524316
0.23260425672023 -0.749999954398116
0.24691064225779 -0.499999951796539
0.256769071483239 -0.282282830731541
0.263657053041329 -0.124553528143585
0.268531347898642 -0.0141167617395606
0.2720199991278 0.0632910112750571
0.274541460786155 0.118024983185928
0.276378905484745 0.157106598270558
0.277726961789156 0.185272820695173
0.278721390116492 0.205742039083515
0.279458164041008 0.220723880995456
0.28000592986808 0.231754115127957
0.280414280426967 0.239913769016258
0.280719343393222 0.245972876116793
0.280947617832453 0.250485716483043
0.281118648162559 0.253854828950455
};
\addlegendentry{\footnotesize Agent 5}
\end{axis}

\end{tikzpicture}

%% file: alt_x1_manual.tex
% Change size of legend entries to footnotesize.
% Remove xlabel
\begin{tikzpicture}

\definecolor{grey}{RGB}{176,176,176}
\definecolor{orange}{RGB}{217,83,25}
\definecolor{green}{RGB}{119,172,48}
\definecolor{lightgrey}{RGB}{204,204,204}
\definecolor{blue}{RGB}{0,114,189}

\begin{axis}[
height=\axisheight,
legend cell align={left},
legend style={
  fill opacity=0.8,
  draw opacity=1,
  text opacity=1,
  at={(0.03,0.97)},
  anchor=north west,
  draw=lightgrey
},
minor tick num=2,
minor xtick={-1,1,2,3,4,6,7,8,9,11,12,13,14,16,17,18,19,21,22,23,24,26,27,28,29,31},
minor ytick={-0.35,-0.3,-0.25,-0.15,-0.1,-0.05,0.05,0.1,0.15,0.25,0.3,0.35,0.45,0.5,0.55},
tick align=outside,
tick pos=left,
width=\axiswidth,
x grid style={grey},
xmajorgrids,
xmin=-1.5, xmax=31.5,
xtick distance=10,
xtick style={color=black},
xtick={-5,0,5,10,15,20,25,30,35},
y grid style={grey},
ylabel={\(\displaystyle x_{i,1}\) in m},
ymajorgrids,
ymin=-0.377411557256483, ymax=0.620916492215227,
ytick distance=0.2,
ytick style={color=black},
ytick={-0.4,-0.2,0,0.2,0.4,0.6,0.8}
]
\addplot [thick, blue]
table {%
0 1.00135803222656e-05
5.40000009536743 0.000336647033691406
7 0.00105011463165283
8 0.00213265419006348
8.69999980926514 0.00350010395050049
9.30000019073486 0.00535047054290771
9.80000019073486 0.00761926174163818
10.1999998092651 0.0101082324981689
10.6000003814697 0.0134077072143555
10.8999996185303 0.0165685415267944
11.1999998092651 0.0204699039459229
11.5 0.0252811908721924
11.6999998092651 0.0290937423706055
11.8999996185303 0.0334712266921997
12.1000003814697 0.0384924411773682
12.3000001907349 0.0442442893981934
12.5 0.0508221387863159
12.6999998092651 0.0583277940750122
12.8999996185303 0.0668684244155884
13.1000003814697 0.0765517950057983
13.3000001907349 0.0874824523925781
13.5 0.0997534990310669
13.6000003814697 0.106415629386902
13.6999998092651 0.113438248634338
13.8000001907349 0.120825886726379
13.8999996185303 0.128580331802368
14 0.136700868606567
14.1000003814697 0.145183444023132
14.3000001907349 0.163202524185181
14.5 0.182538390159607
14.6999998092651 0.203036546707153
14.8999996185303 0.224492073059082
15.1000003814697 0.246658205986023
15.5 0.292009472846985
15.8000001907349 0.325788974761963
16 0.347694873809814
16.2000007629395 0.368841171264648
16.3999996185303 0.389039754867554
16.6000003814697 0.408142685890198
16.7999992370605 0.426043033599854
17 0.442672610282898
17.2000007629395 0.457998156547546
17.3999996185303 0.472017288208008
17.6000003814697 0.484752655029297
17.7999992370605 0.496247887611389
18 0.506561636924744
18.2000007629395 0.515763759613037
18.3999996185303 0.523931503295898
18.6000003814697 0.531145572662354
18.7999992370605 0.537488222122192
19 0.543040752410889
19.2000007629395 0.547881364822388
19.3999996185303 0.552085041999817
19.6000003814697 0.555721879005432
19.8999996185303 0.56025505065918
20.2000007629395 0.563854932785034
20.5 0.566691637039185
20.7999992370605 0.568909406661987
21.2000007629395 0.571110486984253
21.6000003814697 0.572655320167542
22.1000003814697 0.573928833007812
22.7999992370605 0.574913501739502
23.7000007629395 0.575422406196594
25.2999992370605 0.575515389442444
30 0.575181007385254
};
\addlegendentry{\footnotesize Agent 1}
\addplot [thick, orange, dash pattern=on 1pt off 3pt on 3pt off 3pt]
table {%
0 -1.00135803222656e-05
6.19999980926514 -0.000332951545715332
7.80000019073486 -0.00104212760925293
8.80000019073486 -0.00213134288787842
9.5 -0.00351798534393311
10.1000003814697 -0.00540578365325928
10.6000003814697 -0.00773108005523682
11 -0.0102901458740234
11.3000001907349 -0.0127474069595337
11.6000003814697 -0.015784740447998
11.8999996185303 -0.0195332765579224
12.1999998092651 -0.0241485834121704
12.3999996185303 -0.0277953147888184
12.6000003814697 -0.0319662094116211
12.8000001907349 -0.0367238521575928
13 -0.0421321392059326
13.1999998092651 -0.0482537746429443
13.3999996185303 -0.0551462173461914
13.6000003814697 -0.0628564357757568
13.8000001907349 -0.0714156627655029
14 -0.080832839012146
14.1999998092651 -0.091089129447937
14.3999996185303 -0.102134585380554
14.6000003814697 -0.11388635635376
14.8000001907349 -0.126230835914612
15.1000003814697 -0.145547389984131
15.8999996185303 -0.197845697402954
16.1000003814697 -0.210242033004761
16.2999992370605 -0.222106099128723
16.5 -0.233346343040466
16.7000007629395 -0.243894934654236
16.8999996185303 -0.25370717048645
17.1000003814697 -0.262759923934937
17.2999992370605 -0.271048069000244
17.5 -0.278582453727722
17.7000007629395 -0.285386562347412
17.8999996185303 -0.29149317741394
18.1000003814697 -0.296942472457886
18.2999992370605 -0.30177891254425
18.5 -0.306049823760986
18.7000007629395 -0.309803485870361
19 -0.314568042755127
19.2999992370605 -0.318430542945862
19.6000003814697 -0.321533679962158
19.8999996185303 -0.324005246162415
20.2999992370605 -0.326509714126587
20.7000007629395 -0.328311085700989
21.2000007629395 -0.329839944839478
21.7999992370605 -0.330948829650879
22.6000003814697 -0.33168363571167
23.7999992370605 -0.332015514373779
26.7999992370605 -0.331885814666748
30 -0.331745505332947
};
\addlegendentry{\footnotesize Agent 2}
\addplot [thick, green, dashed]
table {%
0 -1.00135803222656e-05
6.40000009536743 -0.000335812568664551
8.10000038146973 -0.00108671188354492
9.10000038146973 -0.00217294692993164
9.80000019073486 -0.00353026390075684
10.3999996185303 -0.00535058975219727
10.8999996185303 -0.00756430625915527
11.3000001907349 -0.00997400283813477
11.6999998092651 -0.0131421089172363
12 -0.0161502361297607
12.3000001907349 -0.0198267698287964
12.6000003814697 -0.0243033170700073
12.8999996185303 -0.0297245979309082
13.1000003814697 -0.0339373350143433
13.3000001907349 -0.0386782884597778
13.5 -0.043984055519104
13.6999998092651 -0.0498825311660767
13.8999996185303 -0.0563883781433105
14.1000003814697 -0.063498854637146
14.3000001907349 -0.0711908340454102
14.5 -0.0794185400009155
14.8000001907349 -0.0926095247268677
15.1000003814697 -0.106534957885742
16 -0.148855566978455
16.2999992370605 -0.161919116973877
16.5 -0.170092582702637
16.7000007629395 -0.17777681350708
16.8999996185303 -0.184938311576843
17.1000003814697 -0.191558599472046
17.2999992370605 -0.197632431983948
17.5 -0.203166127204895
17.7000007629395 -0.20817506313324
17.8999996185303 -0.212681531906128
18.2000007629395 -0.218560934066772
18.5 -0.223479628562927
18.7999992370605 -0.227551937103271
19.1000003814697 -0.230891466140747
19.3999996185303 -0.233606338500977
19.7999992370605 -0.236424207687378
20.2000007629395 -0.238511443138123
20.7000007629395 -0.240348815917969
21.2999992370605 -0.241755604743958
22 -0.242690563201904
23 -0.243298292160034
24.7000007629395 -0.243540048599243
30 -0.243445634841919
};
\addlegendentry{\footnotesize Agent 3}
\end{axis}

\end{tikzpicture}

%% file: alt_x2_manual.tex
% Change size of legend entries to footnotesize.
% Remove xlabel
\begin{tikzpicture}

\definecolor{grey}{RGB}{176,176,176}
\definecolor{orange}{RGB}{217,83,25}
\definecolor{green}{RGB}{119,172,48}
\definecolor{lightgrey}{RGB}{204,204,204}
\definecolor{blue}{RGB}{0,114,189}

\begin{axis}[
height=\axisheight,
legend cell align={left},
legend style={
  fill opacity=0.8,
  draw opacity=1,
  text opacity=1,
  at={(0.03,0.5)},
  anchor=north west,
  draw=lightgrey
},
minor tick num=2,
minor xtick={-1,1,2,3,4,6,7,8,9,11,12,13,14,16,17,18,19,21,22,23,24,26,27,28,29,31},
minor ytick={-0.55,-0.5,-0.45,-0.35,-0.3,-0.25,-0.15,-0.1,-0.05,0.05,0.1,0.15,0.25,0.3,0.35,0.45,0.5},
tick align=outside,
tick pos=left,
width=\axiswidth,
x grid style={grey},
xmajorgrids,
xmin=-1.5, xmax=31.5,
xtick distance=10,
xtick style={color=black},
xtick={-5,0,5,10,15,20,25,30,35},
y grid style={grey},
ylabel={\(\displaystyle x_{i,2}\) in m},
ymajorgrids,
ymin=-0.573234230777272, ymax=0.522280559393578,
ytick distance=0.2,
ytick style={color=black},
ytick={-0.6,-0.4,-0.2,0,0.2,0.4,0.6}
]
\addplot [thick, blue]
table {%
0 0
9.89999961853027 0.000375032424926758
11.3000001907349 0.0011744499206543
12.1999998092651 0.00240767002105713
12.8999996185303 0.0041501522064209
13.5 0.00650036334991455
14 0.00924551486968994
14.5 0.0127707719802856
15.1000003814697 0.0178946256637573
17 0.0350829362869263
17.6000003814697 0.0392204523086548
18.2000007629395 0.0424686670303345
18.7999992370605 0.0449120998382568
19.5 0.0469367504119873
20.2999992370605 0.0484499931335449
21.2999992370605 0.0495624542236328
22.7999992370605 0.0503607988357544
25.2999992370605 0.0507966279983521
30 0.0509417057037354
};
\addlegendentry{\footnotesize Agent 1}
\addplot [thick, orange, dash pattern=on 1pt off 3pt on 3pt off 3pt]
table {%
0 1.00135803222656e-05
6 0.000383257865905762
7.59999990463257 0.00118064880371094
8.60000038146973 0.00238502025604248
9.30000019073486 0.00390148162841797
9.89999961853027 0.00594818592071533
10.3999996185303 0.00845170021057129
10.8000001907349 0.0111920833587646
11.1999998092651 0.0148164033889771
11.5 0.0182796716690063
11.8000001907349 0.0225424766540527
12.1000003814697 0.0277804136276245
12.3000001907349 0.0319151878356934
12.5 0.0366439819335938
12.6999998092651 0.0420417785644531
12.8999996185303 0.0481877326965332
13.1000003814697 0.0551636219024658
13.3000001907349 0.0630500316619873
13.5 0.0719225406646729
13.6999998092651 0.0818454027175903
13.8999996185303 0.092864990234375
14.1000003814697 0.105002403259277
14.3000001907349 0.118246912956238
14.5 0.132551074028015
14.6999998092651 0.147827744483948
14.8999996185303 0.163951635360718
15.1000003814697 0.180763483047485
15.3999996185303 0.206860899925232
16 0.259689688682556
16.2999992370605 0.285101413726807
16.5 0.301294326782227
16.7000007629395 0.316744804382324
16.8999996185303 0.331361532211304
17.1000003814697 0.345080018043518
17.2999992370605 0.357861757278442
17.5 0.369690775871277
17.7000007629395 0.380570411682129
17.8999996185303 0.390520095825195
18.1000003814697 0.399572491645813
18.2999992370605 0.407769322395325
18.5 0.415159463882446
18.7000007629395 0.421796441078186
18.8999996185303 0.427735686302185
19.1000003814697 0.433033585548401
19.2999992370605 0.437745809555054
19.6000003814697 0.443833231925964
19.8999996185303 0.448894500732422
20.2000007629395 0.45308518409729
20.5 0.45654296875
20.8999996185303 0.460219264030457
21.2999992370605 0.46304452419281
21.7999992370605 0.465668439865112
22.3999996185303 0.46785581111908
23.1000003814697 0.469517707824707
24 0.470790982246399
25.2000007629395 0.471671581268311
27.2000007629395 0.472254872322083
30 0.472484469413757
};
\addlegendentry{\footnotesize Agent 2}
\addplot [thick, green, dashed]
table {%
0 -1.00135803222656e-05
5.90000009536743 -0.000377893447875977
7.5 -0.00117850303649902
8.5 -0.00239861011505127
9.19999980926514 -0.0039442777633667
9.80000019073486 -0.00604057312011719
10.3000001907349 -0.00861525535583496
10.6999998092651 -0.0114431381225586
11 -0.0141552686691284
11.3000001907349 -0.0175061225891113
11.6000003814697 -0.0216422080993652
11.8999996185303 -0.0267406702041626
12.1000003814697 -0.0307772159576416
12.3000001907349 -0.0354061126708984
12.5 -0.0407058000564575
12.6999998092651 -0.0467610359191895
12.8999996185303 -0.0536613464355469
13.1000003814697 -0.0614984035491943
13.3000001907349 -0.0703624486923218
13.5 -0.0803366899490356
13.6999998092651 -0.0914907455444336
13.8999996185303 -0.103873014450073
14.1000003814697 -0.117502212524414
14.3000001907349 -0.132360339164734
14.5 -0.148387432098389
14.6999998092651 -0.165479302406311
14.8999996185303 -0.183489799499512
15.1000003814697 -0.20223593711853
15.3999996185303 -0.231269955635071
16 -0.289807558059692
16.2000007629395 -0.30866265296936
16.3999996185303 -0.32688045501709
16.6000003814697 -0.344317436218262
16.7999992370605 -0.360861897468567
17 -0.376432418823242
17.2000007629395 -0.390977144241333
17.3999996185303 -0.404469966888428
17.6000003814697 -0.416907787322998
17.7999992370605 -0.428306579589844
18 -0.438697695732117
18.2000007629395 -0.448124408721924
18.3999996185303 -0.45663845539093
18.6000003814697 -0.464297294616699
18.7999992370605 -0.471161842346191
19 -0.477294325828552
19.2000007629395 -0.482756614685059
19.3999996185303 -0.487609148025513
19.6000003814697 -0.491909623146057
19.8999996185303 -0.497444629669189
20.2000007629395 -0.50203013420105
20.5 -0.50581693649292
20.7999992370605 -0.508935928344727
21.2000007629395 -0.512248754501343
21.6000003814697 -0.514794945716858
22.1000003814697 -0.517163276672363
22.7000007629395 -0.519144892692566
23.3999996185303 -0.520660400390625
24.2999992370605 -0.521833896636963
25.6000003814697 -0.522705078125
27.7000007629395 -0.523250341415405
30 -0.523438096046448
};
\addlegendentry{\footnotesize Agent 3}
\end{axis}

\end{tikzpicture}

%% file: alt_x3_manual.tex
% Change size of legend entries to footnotesize.
\begin{tikzpicture}

\definecolor{grey}{RGB}{176,176,176}
\definecolor{orange}{RGB}{217,83,25}
\definecolor{green}{RGB}{119,172,48}
\definecolor{lightgrey}{RGB}{204,204,204}
\definecolor{blue}{RGB}{0,114,189}

\begin{axis}[
height=\axisheight,
legend cell align={left},
legend style={fill opacity=0.8, draw opacity=1, text opacity=1, draw=lightgrey},
minor tick num=2,
minor xtick={-1,1,2,3,4,6,7,8,9,11,12,13,14,16,17,18,19,21,22,23,24,26,27,28,29,31},
minor ytick={0.95,1.05,1.1,1.15,1.2,1.3,1.35,1.4,1.45,1.55,1.6,1.65,1.7,1.8,1.85,1.9,1.95,2.05,2.1,2.15,2.2,2.3,2.35,2.4,2.45,2.55,2.6,2.65,2.7,2.8,2.85,2.9,2.95,3.05},
tick align=outside,
tick pos=left,
width=\axiswidth,
x grid style={grey},
xlabel={time in s},
xmajorgrids,
xmin=-1.5, xmax=31.5,
xtick distance=10,
xtick style={color=black},
xtick={-5,0,5,10,15,20,25,30,35},
y grid style={grey},
ylabel={\(\displaystyle x_{i,3}\) in m},
ymajorgrids,
ymin=0.9, ymax=3.1,
ytick distance=0.2,
ytick style={color=black},
ytick={0.75,1,1.25,1.5,1.75,2,2.25,2.5,2.75,3,3.25}
]
\addplot [thick, blue]
table {%
0 1
0.100000023841858 1
0.200000047683716 1.00266444683075
0.299999952316284 1.00753903388977
0.399999976158142 1.01433002948761
0.5 1.02278172969818
0.600000023841858 1.03266632556915
0.700000047683716 1.04378092288971
0.799999952316284 1.05594432353973
1 1.0827910900116
1.20000004768372 1.11212027072906
1.39999997615814 1.14307522773743
1.79999995231628 1.2073210477829
2.29999995231628 1.28760290145874
2.59999990463257 1.33414173126221
2.90000009536743 1.37882161140442
3.20000004768372 1.42136013507843
3.40000009536743 1.44845283031464
3.59999990463257 1.47450757026672
3.79999995231628 1.4995197057724
4 1.52349555492401
4.30000019073486 1.55754995346069
4.59999990463257 1.58938109874725
4.90000009536743 1.61908292770386
5.19999980926514 1.64676141738892
5.5 1.67252790927887
5.80000019073486 1.6964955329895
6.09999990463257 1.7187762260437
6.40000009536743 1.73947882652283
6.69999980926514 1.75870788097382
7 1.77656328678131
7.30000019073486 1.79313921928406
7.69999980926514 1.81340336799622
8.10000038146973 1.83174669742584
8.5 1.84834825992584
8.89999961853027 1.86337172985077
9.30000019073486 1.87696599960327
9.69999980926514 1.88926613330841
10.1999998092651 1.90300738811493
10.6999998092651 1.91513228416443
11.1999998092651 1.92583060264587
11.8000001907349 1.93702006340027
12.3999996185303 1.946648478508
13.1000003814697 1.9561972618103
13.8000001907349 1.96421051025391
14.6000003814697 1.97180306911469
15.5 1.97870981693268
16.5 1.98476183414459
17.6000003814697 1.98988234996796
18.8999996185303 1.99436843395233
20.5 1.99821698665619
22.3999996185303 2.00117325782776
24.8999996185303 2.00343060493469
28.3999996185303 2.00494456291199
30 2.00530099868774
};
\addlegendentry{\footnotesize Agent 1}
\addplot [thick, orange, dash pattern=on 1pt off 3pt on 3pt off 3pt]
table {%
0 2
0.600000023841858 2.00071001052856
3.5 2.00394320487976
6.30000019073486 2.00544548034668
10.8000001907349 2.00609493255615
29.2000007629395 2.0060350894928
30 2.00603342056274
};
\addlegendentry{\footnotesize Agent 2}
\addplot [thick, green, dashed]
table {%
0 3
0.100000023841858 3
0.200000047683716 2.99748969078064
0.299999952316284 2.99277591705322
0.399999976158142 2.98614501953125
0.5 2.97785258293152
0.600000023841858 2.96812701225281
0.700000047683716 2.95717120170593
0.799999952316284 2.94516658782959
1 2.9186372756958
1.20000004768372 2.88962531089783
1.39999997615814 2.8589870929718
1.79999995231628 2.79536986351013
2.29999995231628 2.71585536003113
2.59999990463257 2.66976261138916
2.90000009536743 2.62551522254944
3.20000004768372 2.58339405059814
3.5 2.54354524612427
3.70000004768372 2.51827025413513
3.90000009536743 2.49402832984924
4.19999980926514 2.45957684516907
4.5 2.42735886573792
4.80000019073486 2.39728617668152
5.09999990463257 2.36925625801086
5.40000009536743 2.34315967559814
5.69999980926514 2.31888389587402
6 2.29631686210632
6.30000019073486 2.27534961700439
6.59999990463257 2.25587630271912
6.90000009536743 2.23779606819153
7.19999980926514 2.22101354598999
7.5 2.2054386138916
7.90000009536743 2.18640446662903
8.30000019073486 2.16918087005615
8.69999980926514 2.15359783172607
9.10000038146973 2.13950037956238
9.5 2.12674784660339
10 2.11250519752502
10.5 2.09994173049927
11 2.08886003494263
11.6000003814697 2.07727360725403
12.1999998092651 2.06730771064758
12.8999996185303 2.05742812156677
13.6000003814697 2.04914116859436
14.3999996185303 2.04129338264465
15.3000001907349 2.03415822982788
16.2999992370605 2.02790999412537
17.3999996185303 2.02262711524963
18.7000007629395 2.01800227165222
20.2000007629395 2.01424241065979
22.1000003814697 2.01112365722656
24.5 2.00881505012512
27.8999996185303 2.00721216201782
30 2.00672507286072
};
\addlegendentry{\footnotesize Agent 3}
\end{axis}

\end{tikzpicture}

%% file: alt_outputs.pgf
\begingroup%
\makeatletter%
\begin{pgfpicture}%
\pgfpathrectangle{\pgfpointorigin}{\pgfqpoint{4.169636in}{3.896000in}}%
\pgfusepath{use as bounding box, clip}%
\begin{pgfscope}%
\pgfsetbuttcap%
\pgfsetmiterjoin%
\definecolor{currentfill}{rgb}{1.000000,1.000000,1.000000}%
\pgfsetfillcolor{currentfill}%
\pgfsetlinewidth{0.000000pt}%
\definecolor{currentstroke}{rgb}{1.000000,1.000000,1.000000}%
\pgfsetstrokecolor{currentstroke}%
\pgfsetdash{}{0pt}%
\pgfpathmoveto{\pgfqpoint{0.000000in}{0.000000in}}%
\pgfpathlineto{\pgfqpoint{4.169636in}{0.000000in}}%
\pgfpathlineto{\pgfqpoint{4.169636in}{3.896000in}}%
\pgfpathlineto{\pgfqpoint{0.000000in}{3.896000in}}%
\pgfpathlineto{\pgfqpoint{0.000000in}{0.000000in}}%
\pgfpathclose%
\pgfusepath{fill}%
\end{pgfscope}%
\begin{pgfscope}%
\pgfsetbuttcap%
\pgfsetmiterjoin%
\definecolor{currentfill}{rgb}{1.000000,1.000000,1.000000}%
\pgfsetfillcolor{currentfill}%
\pgfsetlinewidth{0.000000pt}%
\definecolor{currentstroke}{rgb}{0.000000,0.000000,0.000000}%
\pgfsetstrokecolor{currentstroke}%
\pgfsetstrokeopacity{0.000000}%
\pgfsetdash{}{0pt}%
\pgfpathmoveto{\pgfqpoint{0.100000in}{0.100000in}}%
\pgfpathlineto{\pgfqpoint{3.796000in}{0.100000in}}%
\pgfpathlineto{\pgfqpoint{3.796000in}{3.796000in}}%
\pgfpathlineto{\pgfqpoint{0.100000in}{3.796000in}}%
\pgfpathlineto{\pgfqpoint{0.100000in}{0.100000in}}%
\pgfpathclose%
\pgfusepath{fill}%
\end{pgfscope}%
\begin{pgfscope}%
\pgfsetbuttcap%
\pgfsetmiterjoin%
\definecolor{currentfill}{rgb}{0.950000,0.950000,0.950000}%
\pgfsetfillcolor{currentfill}%
\pgfsetfillopacity{0.500000}%
\pgfsetlinewidth{1.003750pt}%
\definecolor{currentstroke}{rgb}{0.950000,0.950000,0.950000}%
\pgfsetstrokecolor{currentstroke}%
\pgfsetstrokeopacity{0.500000}%
\pgfsetdash{}{0pt}%
\pgfpathmoveto{\pgfqpoint{0.435156in}{0.907452in}}%
\pgfpathlineto{\pgfqpoint{1.368966in}{1.710475in}}%
\pgfpathlineto{\pgfqpoint{1.350794in}{3.287994in}}%
\pgfpathlineto{\pgfqpoint{0.385055in}{2.632186in}}%
\pgfusepath{stroke,fill}%
\end{pgfscope}%
\begin{pgfscope}%
\pgfsetbuttcap%
\pgfsetmiterjoin%
\definecolor{currentfill}{rgb}{0.900000,0.900000,0.900000}%
\pgfsetfillcolor{currentfill}%
\pgfsetfillopacity{0.500000}%
\pgfsetlinewidth{1.003750pt}%
\definecolor{currentstroke}{rgb}{0.900000,0.900000,0.900000}%
\pgfsetstrokecolor{currentstroke}%
\pgfsetstrokeopacity{0.500000}%
\pgfsetdash{}{0pt}%
\pgfpathmoveto{\pgfqpoint{1.368966in}{1.710475in}}%
\pgfpathlineto{\pgfqpoint{3.464729in}{1.421160in}}%
\pgfpathlineto{\pgfqpoint{3.508776in}{3.052182in}}%
\pgfpathlineto{\pgfqpoint{1.350794in}{3.287994in}}%
\pgfusepath{stroke,fill}%
\end{pgfscope}%
\begin{pgfscope}%
\pgfsetbuttcap%
\pgfsetmiterjoin%
\definecolor{currentfill}{rgb}{0.925000,0.925000,0.925000}%
\pgfsetfillcolor{currentfill}%
\pgfsetfillopacity{0.500000}%
\pgfsetlinewidth{1.003750pt}%
\definecolor{currentstroke}{rgb}{0.925000,0.925000,0.925000}%
\pgfsetstrokecolor{currentstroke}%
\pgfsetstrokeopacity{0.500000}%
\pgfsetdash{}{0pt}%
\pgfpathmoveto{\pgfqpoint{0.435156in}{0.907452in}}%
\pgfpathlineto{\pgfqpoint{2.723326in}{0.551957in}}%
\pgfpathlineto{\pgfqpoint{3.464729in}{1.421160in}}%
\pgfpathlineto{\pgfqpoint{1.368966in}{1.710475in}}%
\pgfusepath{stroke,fill}%
\end{pgfscope}%
\begin{pgfscope}%
\pgfsetrectcap%
\pgfsetroundjoin%
\pgfsetlinewidth{0.803000pt}%
\definecolor{currentstroke}{rgb}{0.000000,0.000000,0.000000}%
\pgfsetstrokecolor{currentstroke}%
\pgfsetdash{}{0pt}%
\pgfpathmoveto{\pgfqpoint{0.435156in}{0.907452in}}%
\pgfpathlineto{\pgfqpoint{2.723326in}{0.551957in}}%
\pgfusepath{stroke}%
\end{pgfscope}%
\begin{pgfscope}%
\definecolor{textcolor}{rgb}{0.000000,0.000000,0.000000}%
\pgfsetstrokecolor{textcolor}%
\pgfsetfillcolor{textcolor}%
\pgftext[x=1.130034in, y=0.262435in, left, base,rotate=351.169007]{\color{textcolor}\rmfamily\fontsize{10.000000}{12.000000}\selectfont \(\displaystyle x_{i,1}\) in m}%
\end{pgfscope}%
\begin{pgfscope}%
\pgfsetbuttcap%
\pgfsetroundjoin%
\pgfsetlinewidth{0.803000pt}%
\definecolor{currentstroke}{rgb}{0.690196,0.690196,0.690196}%
\pgfsetstrokecolor{currentstroke}%
\pgfsetdash{}{0pt}%
\pgfpathmoveto{\pgfqpoint{0.856774in}{0.841949in}}%
\pgfpathlineto{\pgfqpoint{1.756416in}{1.656988in}}%
\pgfpathlineto{\pgfqpoint{1.749387in}{3.244438in}}%
\pgfusepath{stroke}%
\end{pgfscope}%
\begin{pgfscope}%
\pgfsetbuttcap%
\pgfsetroundjoin%
\pgfsetlinewidth{0.803000pt}%
\definecolor{currentstroke}{rgb}{0.690196,0.690196,0.690196}%
\pgfsetstrokecolor{currentstroke}%
\pgfsetdash{}{0pt}%
\pgfpathmoveto{\pgfqpoint{1.288948in}{0.774805in}}%
\pgfpathlineto{\pgfqpoint{2.152963in}{1.602246in}}%
\pgfpathlineto{\pgfqpoint{2.157507in}{3.199841in}}%
\pgfusepath{stroke}%
\end{pgfscope}%
\begin{pgfscope}%
\pgfsetbuttcap%
\pgfsetroundjoin%
\pgfsetlinewidth{0.803000pt}%
\definecolor{currentstroke}{rgb}{0.690196,0.690196,0.690196}%
\pgfsetstrokecolor{currentstroke}%
\pgfsetdash{}{0pt}%
\pgfpathmoveto{\pgfqpoint{1.728063in}{0.706584in}}%
\pgfpathlineto{\pgfqpoint{2.555254in}{1.546711in}}%
\pgfpathlineto{\pgfqpoint{2.571714in}{3.154579in}}%
\pgfusepath{stroke}%
\end{pgfscope}%
\begin{pgfscope}%
\pgfsetbuttcap%
\pgfsetroundjoin%
\pgfsetlinewidth{0.803000pt}%
\definecolor{currentstroke}{rgb}{0.690196,0.690196,0.690196}%
\pgfsetstrokecolor{currentstroke}%
\pgfsetdash{}{0pt}%
\pgfpathmoveto{\pgfqpoint{2.174289in}{0.637257in}}%
\pgfpathlineto{\pgfqpoint{2.963415in}{1.490365in}}%
\pgfpathlineto{\pgfqpoint{2.992145in}{3.108637in}}%
\pgfusepath{stroke}%
\end{pgfscope}%
\begin{pgfscope}%
\pgfsetbuttcap%
\pgfsetroundjoin%
\pgfsetlinewidth{0.803000pt}%
\definecolor{currentstroke}{rgb}{0.690196,0.690196,0.690196}%
\pgfsetstrokecolor{currentstroke}%
\pgfsetdash{}{0pt}%
\pgfpathmoveto{\pgfqpoint{2.627800in}{0.566799in}}%
\pgfpathlineto{\pgfqpoint{3.377577in}{1.433191in}}%
\pgfpathlineto{\pgfqpoint{3.418941in}{3.061999in}}%
\pgfusepath{stroke}%
\end{pgfscope}%
\begin{pgfscope}%
\pgfsetrectcap%
\pgfsetroundjoin%
\pgfsetlinewidth{0.803000pt}%
\definecolor{currentstroke}{rgb}{0.000000,0.000000,0.000000}%
\pgfsetstrokecolor{currentstroke}%
\pgfsetdash{}{0pt}%
\pgfpathmoveto{\pgfqpoint{0.864735in}{0.849161in}}%
\pgfpathlineto{\pgfqpoint{0.840811in}{0.827487in}}%
\pgfusepath{stroke}%
\end{pgfscope}%
\begin{pgfscope}%
\definecolor{textcolor}{rgb}{0.000000,0.000000,0.000000}%
\pgfsetstrokecolor{textcolor}%
\pgfsetfillcolor{textcolor}%
\pgftext[x=0.776596in,y=0.626150in,,top]{\color{textcolor}\rmfamily\fontsize{10.000000}{12.000000}\selectfont \ensuremath{-}0.2}%
\end{pgfscope}%
\begin{pgfscope}%
\pgfsetrectcap%
\pgfsetroundjoin%
\pgfsetlinewidth{0.803000pt}%
\definecolor{currentstroke}{rgb}{0.000000,0.000000,0.000000}%
\pgfsetstrokecolor{currentstroke}%
\pgfsetdash{}{0pt}%
\pgfpathmoveto{\pgfqpoint{1.296599in}{0.782133in}}%
\pgfpathlineto{\pgfqpoint{1.273605in}{0.760112in}}%
\pgfusepath{stroke}%
\end{pgfscope}%
\begin{pgfscope}%
\definecolor{textcolor}{rgb}{0.000000,0.000000,0.000000}%
\pgfsetstrokecolor{textcolor}%
\pgfsetfillcolor{textcolor}%
\pgftext[x=1.210883in,y=0.556848in,,top]{\color{textcolor}\rmfamily\fontsize{10.000000}{12.000000}\selectfont 0.0}%
\end{pgfscope}%
\begin{pgfscope}%
\pgfsetrectcap%
\pgfsetroundjoin%
\pgfsetlinewidth{0.803000pt}%
\definecolor{currentstroke}{rgb}{0.000000,0.000000,0.000000}%
\pgfsetstrokecolor{currentstroke}%
\pgfsetdash{}{0pt}%
\pgfpathmoveto{\pgfqpoint{1.735394in}{0.714029in}}%
\pgfpathlineto{\pgfqpoint{1.713363in}{0.691653in}}%
\pgfusepath{stroke}%
\end{pgfscope}%
\begin{pgfscope}%
\definecolor{textcolor}{rgb}{0.000000,0.000000,0.000000}%
\pgfsetstrokecolor{textcolor}%
\pgfsetfillcolor{textcolor}%
\pgftext[x=1.652191in,y=0.486426in,,top]{\color{textcolor}\rmfamily\fontsize{10.000000}{12.000000}\selectfont 0.2}%
\end{pgfscope}%
\begin{pgfscope}%
\pgfsetrectcap%
\pgfsetroundjoin%
\pgfsetlinewidth{0.803000pt}%
\definecolor{currentstroke}{rgb}{0.000000,0.000000,0.000000}%
\pgfsetstrokecolor{currentstroke}%
\pgfsetdash{}{0pt}%
\pgfpathmoveto{\pgfqpoint{2.181289in}{0.644824in}}%
\pgfpathlineto{\pgfqpoint{2.160254in}{0.622084in}}%
\pgfusepath{stroke}%
\end{pgfscope}%
\begin{pgfscope}%
\definecolor{textcolor}{rgb}{0.000000,0.000000,0.000000}%
\pgfsetstrokecolor{textcolor}%
\pgfsetfillcolor{textcolor}%
\pgftext[x=2.100690in,y=0.414856in,,top]{\color{textcolor}\rmfamily\fontsize{10.000000}{12.000000}\selectfont 0.4}%
\end{pgfscope}%
\begin{pgfscope}%
\pgfsetrectcap%
\pgfsetroundjoin%
\pgfsetlinewidth{0.803000pt}%
\definecolor{currentstroke}{rgb}{0.000000,0.000000,0.000000}%
\pgfsetstrokecolor{currentstroke}%
\pgfsetdash{}{0pt}%
\pgfpathmoveto{\pgfqpoint{2.634455in}{0.574489in}}%
\pgfpathlineto{\pgfqpoint{2.614454in}{0.551376in}}%
\pgfusepath{stroke}%
\end{pgfscope}%
\begin{pgfscope}%
\definecolor{textcolor}{rgb}{0.000000,0.000000,0.000000}%
\pgfsetstrokecolor{textcolor}%
\pgfsetfillcolor{textcolor}%
\pgftext[x=2.556559in,y=0.342110in,,top]{\color{textcolor}\rmfamily\fontsize{10.000000}{12.000000}\selectfont 0.6}%
\end{pgfscope}%
\begin{pgfscope}%
\pgfsetrectcap%
\pgfsetroundjoin%
\pgfsetlinewidth{0.803000pt}%
\definecolor{currentstroke}{rgb}{0.000000,0.000000,0.000000}%
\pgfsetstrokecolor{currentstroke}%
\pgfsetdash{}{0pt}%
\pgfpathmoveto{\pgfqpoint{3.464729in}{1.421160in}}%
\pgfpathlineto{\pgfqpoint{2.723326in}{0.551957in}}%
\pgfusepath{stroke}%
\end{pgfscope}%
\begin{pgfscope}%
\definecolor{textcolor}{rgb}{0.000000,0.000000,0.000000}%
\pgfsetstrokecolor{textcolor}%
\pgfsetfillcolor{textcolor}%
\pgftext[x=3.375618in, y=0.436233in, left, base,rotate=49.536842]{\color{textcolor}\rmfamily\fontsize{10.000000}{12.000000}\selectfont \(\displaystyle x_{i,2}\) in m}%
\end{pgfscope}%
\begin{pgfscope}%
\pgfsetbuttcap%
\pgfsetroundjoin%
\pgfsetlinewidth{0.803000pt}%
\definecolor{currentstroke}{rgb}{0.690196,0.690196,0.690196}%
\pgfsetstrokecolor{currentstroke}%
\pgfsetdash{}{0pt}%
\pgfpathmoveto{\pgfqpoint{0.565761in}{2.754899in}}%
\pgfpathlineto{\pgfqpoint{0.609451in}{1.057337in}}%
\pgfpathlineto{\pgfqpoint{2.862168in}{0.714733in}}%
\pgfusepath{stroke}%
\end{pgfscope}%
\begin{pgfscope}%
\pgfsetbuttcap%
\pgfsetroundjoin%
\pgfsetlinewidth{0.803000pt}%
\definecolor{currentstroke}{rgb}{0.690196,0.690196,0.690196}%
\pgfsetstrokecolor{currentstroke}%
\pgfsetdash{}{0pt}%
\pgfpathmoveto{\pgfqpoint{0.743344in}{2.875491in}}%
\pgfpathlineto{\pgfqpoint{0.780930in}{1.204798in}}%
\pgfpathlineto{\pgfqpoint{2.998559in}{0.874635in}}%
\pgfusepath{stroke}%
\end{pgfscope}%
\begin{pgfscope}%
\pgfsetbuttcap%
\pgfsetroundjoin%
\pgfsetlinewidth{0.803000pt}%
\definecolor{currentstroke}{rgb}{0.690196,0.690196,0.690196}%
\pgfsetstrokecolor{currentstroke}%
\pgfsetdash{}{0pt}%
\pgfpathmoveto{\pgfqpoint{0.914472in}{2.991699in}}%
\pgfpathlineto{\pgfqpoint{0.946358in}{1.347057in}}%
\pgfpathlineto{\pgfqpoint{3.129946in}{1.028669in}}%
\pgfusepath{stroke}%
\end{pgfscope}%
\begin{pgfscope}%
\pgfsetbuttcap%
\pgfsetroundjoin%
\pgfsetlinewidth{0.803000pt}%
\definecolor{currentstroke}{rgb}{0.690196,0.690196,0.690196}%
\pgfsetstrokecolor{currentstroke}%
\pgfsetdash{}{0pt}%
\pgfpathmoveto{\pgfqpoint{1.079491in}{3.103759in}}%
\pgfpathlineto{\pgfqpoint{1.106052in}{1.484384in}}%
\pgfpathlineto{\pgfqpoint{3.256597in}{1.177152in}}%
\pgfusepath{stroke}%
\end{pgfscope}%
\begin{pgfscope}%
\pgfsetbuttcap%
\pgfsetroundjoin%
\pgfsetlinewidth{0.803000pt}%
\definecolor{currentstroke}{rgb}{0.690196,0.690196,0.690196}%
\pgfsetstrokecolor{currentstroke}%
\pgfsetdash{}{0pt}%
\pgfpathmoveto{\pgfqpoint{1.238721in}{3.211888in}}%
\pgfpathlineto{\pgfqpoint{1.260303in}{1.617031in}}%
\pgfpathlineto{\pgfqpoint{3.378766in}{1.320379in}}%
\pgfusepath{stroke}%
\end{pgfscope}%
\begin{pgfscope}%
\pgfsetrectcap%
\pgfsetroundjoin%
\pgfsetlinewidth{0.803000pt}%
\definecolor{currentstroke}{rgb}{0.000000,0.000000,0.000000}%
\pgfsetstrokecolor{currentstroke}%
\pgfsetdash{}{0pt}%
\pgfpathmoveto{\pgfqpoint{2.843400in}{0.717587in}}%
\pgfpathlineto{\pgfqpoint{2.899740in}{0.709018in}}%
\pgfusepath{stroke}%
\end{pgfscope}%
\begin{pgfscope}%
\definecolor{textcolor}{rgb}{0.000000,0.000000,0.000000}%
\pgfsetstrokecolor{textcolor}%
\pgfsetfillcolor{textcolor}%
\pgftext[x=3.054759in,y=0.543631in,,top]{\color{textcolor}\rmfamily\fontsize{10.000000}{12.000000}\selectfont \ensuremath{-}0.4}%
\end{pgfscope}%
\begin{pgfscope}%
\pgfsetrectcap%
\pgfsetroundjoin%
\pgfsetlinewidth{0.803000pt}%
\definecolor{currentstroke}{rgb}{0.000000,0.000000,0.000000}%
\pgfsetstrokecolor{currentstroke}%
\pgfsetdash{}{0pt}%
\pgfpathmoveto{\pgfqpoint{2.980098in}{0.877383in}}%
\pgfpathlineto{\pgfqpoint{3.035518in}{0.869132in}}%
\pgfusepath{stroke}%
\end{pgfscope}%
\begin{pgfscope}%
\definecolor{textcolor}{rgb}{0.000000,0.000000,0.000000}%
\pgfsetstrokecolor{textcolor}%
\pgfsetfillcolor{textcolor}%
\pgftext[x=3.187659in,y=0.706863in,,top]{\color{textcolor}\rmfamily\fontsize{10.000000}{12.000000}\selectfont \ensuremath{-}0.2}%
\end{pgfscope}%
\begin{pgfscope}%
\pgfsetrectcap%
\pgfsetroundjoin%
\pgfsetlinewidth{0.803000pt}%
\definecolor{currentstroke}{rgb}{0.000000,0.000000,0.000000}%
\pgfsetstrokecolor{currentstroke}%
\pgfsetdash{}{0pt}%
\pgfpathmoveto{\pgfqpoint{3.111781in}{1.031317in}}%
\pgfpathlineto{\pgfqpoint{3.166311in}{1.023366in}}%
\pgfusepath{stroke}%
\end{pgfscope}%
\begin{pgfscope}%
\definecolor{textcolor}{rgb}{0.000000,0.000000,0.000000}%
\pgfsetstrokecolor{textcolor}%
\pgfsetfillcolor{textcolor}%
\pgftext[x=3.315677in,y=0.864101in,,top]{\color{textcolor}\rmfamily\fontsize{10.000000}{12.000000}\selectfont 0.0}%
\end{pgfscope}%
\begin{pgfscope}%
\pgfsetrectcap%
\pgfsetroundjoin%
\pgfsetlinewidth{0.803000pt}%
\definecolor{currentstroke}{rgb}{0.000000,0.000000,0.000000}%
\pgfsetstrokecolor{currentstroke}%
\pgfsetdash{}{0pt}%
\pgfpathmoveto{\pgfqpoint{3.238720in}{1.179706in}}%
\pgfpathlineto{\pgfqpoint{3.292387in}{1.172039in}}%
\pgfusepath{stroke}%
\end{pgfscope}%
\begin{pgfscope}%
\definecolor{textcolor}{rgb}{0.000000,0.000000,0.000000}%
\pgfsetstrokecolor{textcolor}%
\pgfsetfillcolor{textcolor}%
\pgftext[x=3.439079in,y=1.015667in,,top]{\color{textcolor}\rmfamily\fontsize{10.000000}{12.000000}\selectfont 0.2}%
\end{pgfscope}%
\begin{pgfscope}%
\pgfsetrectcap%
\pgfsetroundjoin%
\pgfsetlinewidth{0.803000pt}%
\definecolor{currentstroke}{rgb}{0.000000,0.000000,0.000000}%
\pgfsetstrokecolor{currentstroke}%
\pgfsetdash{}{0pt}%
\pgfpathmoveto{\pgfqpoint{3.361167in}{1.322844in}}%
\pgfpathlineto{\pgfqpoint{3.413997in}{1.315446in}}%
\pgfusepath{stroke}%
\end{pgfscope}%
\begin{pgfscope}%
\definecolor{textcolor}{rgb}{0.000000,0.000000,0.000000}%
\pgfsetstrokecolor{textcolor}%
\pgfsetfillcolor{textcolor}%
\pgftext[x=3.558109in,y=1.161865in,,top]{\color{textcolor}\rmfamily\fontsize{10.000000}{12.000000}\selectfont 0.4}%
\end{pgfscope}%
\begin{pgfscope}%
\pgfsetrectcap%
\pgfsetroundjoin%
\pgfsetlinewidth{0.803000pt}%
\definecolor{currentstroke}{rgb}{0.000000,0.000000,0.000000}%
\pgfsetstrokecolor{currentstroke}%
\pgfsetdash{}{0pt}%
\pgfpathmoveto{\pgfqpoint{3.464729in}{1.421160in}}%
\pgfpathlineto{\pgfqpoint{3.508776in}{3.052182in}}%
\pgfusepath{stroke}%
\end{pgfscope}%
\begin{pgfscope}%
\definecolor{textcolor}{rgb}{0.000000,0.000000,0.000000}%
\pgfsetstrokecolor{textcolor}%
\pgfsetfillcolor{textcolor}%
\pgftext[x=4.024207in, y=2.046532in, left, base,rotate=88.453050]{\color{textcolor}\rmfamily\fontsize{10.000000}{12.000000}\selectfont \(\displaystyle x_{i,3}\) in m}%
\end{pgfscope}%
\begin{pgfscope}%
\pgfsetbuttcap%
\pgfsetroundjoin%
\pgfsetlinewidth{0.803000pt}%
\definecolor{currentstroke}{rgb}{0.690196,0.690196,0.690196}%
\pgfsetstrokecolor{currentstroke}%
\pgfsetdash{}{0pt}%
\pgfpathmoveto{\pgfqpoint{3.465585in}{1.452848in}}%
\pgfpathlineto{\pgfqpoint{1.368612in}{1.741156in}}%
\pgfpathlineto{\pgfqpoint{0.434184in}{0.940896in}}%
\pgfusepath{stroke}%
\end{pgfscope}%
\begin{pgfscope}%
\pgfsetbuttcap%
\pgfsetroundjoin%
\pgfsetlinewidth{0.803000pt}%
\definecolor{currentstroke}{rgb}{0.690196,0.690196,0.690196}%
\pgfsetstrokecolor{currentstroke}%
\pgfsetdash{}{0pt}%
\pgfpathmoveto{\pgfqpoint{3.470740in}{1.643755in}}%
\pgfpathlineto{\pgfqpoint{1.366483in}{1.925973in}}%
\pgfpathlineto{\pgfqpoint{0.428330in}{1.142437in}}%
\pgfusepath{stroke}%
\end{pgfscope}%
\begin{pgfscope}%
\pgfsetbuttcap%
\pgfsetroundjoin%
\pgfsetlinewidth{0.803000pt}%
\definecolor{currentstroke}{rgb}{0.690196,0.690196,0.690196}%
\pgfsetstrokecolor{currentstroke}%
\pgfsetdash{}{0pt}%
\pgfpathmoveto{\pgfqpoint{3.475932in}{1.836007in}}%
\pgfpathlineto{\pgfqpoint{1.364340in}{2.112044in}}%
\pgfpathlineto{\pgfqpoint{0.422431in}{1.345491in}}%
\pgfusepath{stroke}%
\end{pgfscope}%
\begin{pgfscope}%
\pgfsetbuttcap%
\pgfsetroundjoin%
\pgfsetlinewidth{0.803000pt}%
\definecolor{currentstroke}{rgb}{0.690196,0.690196,0.690196}%
\pgfsetstrokecolor{currentstroke}%
\pgfsetdash{}{0pt}%
\pgfpathmoveto{\pgfqpoint{3.481161in}{2.029620in}}%
\pgfpathlineto{\pgfqpoint{1.362182in}{2.299383in}}%
\pgfpathlineto{\pgfqpoint{0.416488in}{1.550078in}}%
\pgfusepath{stroke}%
\end{pgfscope}%
\begin{pgfscope}%
\pgfsetbuttcap%
\pgfsetroundjoin%
\pgfsetlinewidth{0.803000pt}%
\definecolor{currentstroke}{rgb}{0.690196,0.690196,0.690196}%
\pgfsetstrokecolor{currentstroke}%
\pgfsetdash{}{0pt}%
\pgfpathmoveto{\pgfqpoint{3.486427in}{2.224608in}}%
\pgfpathlineto{\pgfqpoint{1.360009in}{2.488002in}}%
\pgfpathlineto{\pgfqpoint{0.410500in}{1.756214in}}%
\pgfusepath{stroke}%
\end{pgfscope}%
\begin{pgfscope}%
\pgfsetbuttcap%
\pgfsetroundjoin%
\pgfsetlinewidth{0.803000pt}%
\definecolor{currentstroke}{rgb}{0.690196,0.690196,0.690196}%
\pgfsetstrokecolor{currentstroke}%
\pgfsetdash{}{0pt}%
\pgfpathmoveto{\pgfqpoint{3.491730in}{2.420985in}}%
\pgfpathlineto{\pgfqpoint{1.357822in}{2.677915in}}%
\pgfpathlineto{\pgfqpoint{0.404467in}{1.963917in}}%
\pgfusepath{stroke}%
\end{pgfscope}%
\begin{pgfscope}%
\pgfsetbuttcap%
\pgfsetroundjoin%
\pgfsetlinewidth{0.803000pt}%
\definecolor{currentstroke}{rgb}{0.690196,0.690196,0.690196}%
\pgfsetstrokecolor{currentstroke}%
\pgfsetdash{}{0pt}%
\pgfpathmoveto{\pgfqpoint{3.497071in}{2.618766in}}%
\pgfpathlineto{\pgfqpoint{1.355619in}{2.869135in}}%
\pgfpathlineto{\pgfqpoint{0.398388in}{2.173205in}}%
\pgfusepath{stroke}%
\end{pgfscope}%
\begin{pgfscope}%
\pgfsetbuttcap%
\pgfsetroundjoin%
\pgfsetlinewidth{0.803000pt}%
\definecolor{currentstroke}{rgb}{0.690196,0.690196,0.690196}%
\pgfsetstrokecolor{currentstroke}%
\pgfsetdash{}{0pt}%
\pgfpathmoveto{\pgfqpoint{3.502451in}{2.817967in}}%
\pgfpathlineto{\pgfqpoint{1.353401in}{3.061676in}}%
\pgfpathlineto{\pgfqpoint{0.392262in}{2.384096in}}%
\pgfusepath{stroke}%
\end{pgfscope}%
\begin{pgfscope}%
\pgfsetbuttcap%
\pgfsetroundjoin%
\pgfsetlinewidth{0.803000pt}%
\definecolor{currentstroke}{rgb}{0.690196,0.690196,0.690196}%
\pgfsetstrokecolor{currentstroke}%
\pgfsetdash{}{0pt}%
\pgfpathmoveto{\pgfqpoint{3.507869in}{3.018602in}}%
\pgfpathlineto{\pgfqpoint{1.351168in}{3.255550in}}%
\pgfpathlineto{\pgfqpoint{0.386088in}{2.596608in}}%
\pgfusepath{stroke}%
\end{pgfscope}%
\begin{pgfscope}%
\pgfsetrectcap%
\pgfsetroundjoin%
\pgfsetlinewidth{0.803000pt}%
\definecolor{currentstroke}{rgb}{0.000000,0.000000,0.000000}%
\pgfsetstrokecolor{currentstroke}%
\pgfsetdash{}{0pt}%
\pgfpathmoveto{\pgfqpoint{3.448172in}{1.455242in}}%
\pgfpathlineto{\pgfqpoint{3.500442in}{1.448056in}}%
\pgfusepath{stroke}%
\end{pgfscope}%
\begin{pgfscope}%
\definecolor{textcolor}{rgb}{0.000000,0.000000,0.000000}%
\pgfsetstrokecolor{textcolor}%
\pgfsetfillcolor{textcolor}%
\pgftext[x=3.700449in,y=1.492513in,,top]{\color{textcolor}\rmfamily\fontsize{10.000000}{12.000000}\selectfont 1.00}%
\end{pgfscope}%
\begin{pgfscope}%
\pgfsetrectcap%
\pgfsetroundjoin%
\pgfsetlinewidth{0.803000pt}%
\definecolor{currentstroke}{rgb}{0.000000,0.000000,0.000000}%
\pgfsetstrokecolor{currentstroke}%
\pgfsetdash{}{0pt}%
\pgfpathmoveto{\pgfqpoint{3.453265in}{1.646098in}}%
\pgfpathlineto{\pgfqpoint{3.505724in}{1.639063in}}%
\pgfusepath{stroke}%
\end{pgfscope}%
\begin{pgfscope}%
\definecolor{textcolor}{rgb}{0.000000,0.000000,0.000000}%
\pgfsetstrokecolor{textcolor}%
\pgfsetfillcolor{textcolor}%
\pgftext[x=3.706399in,y=1.682586in,,top]{\color{textcolor}\rmfamily\fontsize{10.000000}{12.000000}\selectfont 1.25}%
\end{pgfscope}%
\begin{pgfscope}%
\pgfsetrectcap%
\pgfsetroundjoin%
\pgfsetlinewidth{0.803000pt}%
\definecolor{currentstroke}{rgb}{0.000000,0.000000,0.000000}%
\pgfsetstrokecolor{currentstroke}%
\pgfsetdash{}{0pt}%
\pgfpathmoveto{\pgfqpoint{3.458393in}{1.838300in}}%
\pgfpathlineto{\pgfqpoint{3.511042in}{1.831417in}}%
\pgfusepath{stroke}%
\end{pgfscope}%
\begin{pgfscope}%
\definecolor{textcolor}{rgb}{0.000000,0.000000,0.000000}%
\pgfsetstrokecolor{textcolor}%
\pgfsetfillcolor{textcolor}%
\pgftext[x=3.712391in,y=1.873992in,,top]{\color{textcolor}\rmfamily\fontsize{10.000000}{12.000000}\selectfont 1.50}%
\end{pgfscope}%
\begin{pgfscope}%
\pgfsetrectcap%
\pgfsetroundjoin%
\pgfsetlinewidth{0.803000pt}%
\definecolor{currentstroke}{rgb}{0.000000,0.000000,0.000000}%
\pgfsetstrokecolor{currentstroke}%
\pgfsetdash{}{0pt}%
\pgfpathmoveto{\pgfqpoint{3.463558in}{2.031861in}}%
\pgfpathlineto{\pgfqpoint{3.516398in}{2.025134in}}%
\pgfusepath{stroke}%
\end{pgfscope}%
\begin{pgfscope}%
\definecolor{textcolor}{rgb}{0.000000,0.000000,0.000000}%
\pgfsetstrokecolor{textcolor}%
\pgfsetfillcolor{textcolor}%
\pgftext[x=3.718425in,y=2.066746in,,top]{\color{textcolor}\rmfamily\fontsize{10.000000}{12.000000}\selectfont 1.75}%
\end{pgfscope}%
\begin{pgfscope}%
\pgfsetrectcap%
\pgfsetroundjoin%
\pgfsetlinewidth{0.803000pt}%
\definecolor{currentstroke}{rgb}{0.000000,0.000000,0.000000}%
\pgfsetstrokecolor{currentstroke}%
\pgfsetdash{}{0pt}%
\pgfpathmoveto{\pgfqpoint{3.468760in}{2.226796in}}%
\pgfpathlineto{\pgfqpoint{3.521793in}{2.220227in}}%
\pgfusepath{stroke}%
\end{pgfscope}%
\begin{pgfscope}%
\definecolor{textcolor}{rgb}{0.000000,0.000000,0.000000}%
\pgfsetstrokecolor{textcolor}%
\pgfsetfillcolor{textcolor}%
\pgftext[x=3.724502in,y=2.260861in,,top]{\color{textcolor}\rmfamily\fontsize{10.000000}{12.000000}\selectfont 2.00}%
\end{pgfscope}%
\begin{pgfscope}%
\pgfsetrectcap%
\pgfsetroundjoin%
\pgfsetlinewidth{0.803000pt}%
\definecolor{currentstroke}{rgb}{0.000000,0.000000,0.000000}%
\pgfsetstrokecolor{currentstroke}%
\pgfsetdash{}{0pt}%
\pgfpathmoveto{\pgfqpoint{3.473999in}{2.423119in}}%
\pgfpathlineto{\pgfqpoint{3.527225in}{2.416711in}}%
\pgfusepath{stroke}%
\end{pgfscope}%
\begin{pgfscope}%
\definecolor{textcolor}{rgb}{0.000000,0.000000,0.000000}%
\pgfsetstrokecolor{textcolor}%
\pgfsetfillcolor{textcolor}%
\pgftext[x=3.730622in,y=2.456352in,,top]{\color{textcolor}\rmfamily\fontsize{10.000000}{12.000000}\selectfont 2.25}%
\end{pgfscope}%
\begin{pgfscope}%
\pgfsetrectcap%
\pgfsetroundjoin%
\pgfsetlinewidth{0.803000pt}%
\definecolor{currentstroke}{rgb}{0.000000,0.000000,0.000000}%
\pgfsetstrokecolor{currentstroke}%
\pgfsetdash{}{0pt}%
\pgfpathmoveto{\pgfqpoint{3.479275in}{2.620847in}}%
\pgfpathlineto{\pgfqpoint{3.532697in}{2.614601in}}%
\pgfusepath{stroke}%
\end{pgfscope}%
\begin{pgfscope}%
\definecolor{textcolor}{rgb}{0.000000,0.000000,0.000000}%
\pgfsetstrokecolor{textcolor}%
\pgfsetfillcolor{textcolor}%
\pgftext[x=3.736785in,y=2.653235in,,top]{\color{textcolor}\rmfamily\fontsize{10.000000}{12.000000}\selectfont 2.50}%
\end{pgfscope}%
\begin{pgfscope}%
\pgfsetrectcap%
\pgfsetroundjoin%
\pgfsetlinewidth{0.803000pt}%
\definecolor{currentstroke}{rgb}{0.000000,0.000000,0.000000}%
\pgfsetstrokecolor{currentstroke}%
\pgfsetdash{}{0pt}%
\pgfpathmoveto{\pgfqpoint{3.484589in}{2.819992in}}%
\pgfpathlineto{\pgfqpoint{3.538208in}{2.813912in}}%
\pgfusepath{stroke}%
\end{pgfscope}%
\begin{pgfscope}%
\definecolor{textcolor}{rgb}{0.000000,0.000000,0.000000}%
\pgfsetstrokecolor{textcolor}%
\pgfsetfillcolor{textcolor}%
\pgftext[x=3.742992in,y=2.851522in,,top]{\color{textcolor}\rmfamily\fontsize{10.000000}{12.000000}\selectfont 2.75}%
\end{pgfscope}%
\begin{pgfscope}%
\pgfsetrectcap%
\pgfsetroundjoin%
\pgfsetlinewidth{0.803000pt}%
\definecolor{currentstroke}{rgb}{0.000000,0.000000,0.000000}%
\pgfsetstrokecolor{currentstroke}%
\pgfsetdash{}{0pt}%
\pgfpathmoveto{\pgfqpoint{3.489941in}{3.020572in}}%
\pgfpathlineto{\pgfqpoint{3.543759in}{3.014659in}}%
\pgfusepath{stroke}%
\end{pgfscope}%
\begin{pgfscope}%
\definecolor{textcolor}{rgb}{0.000000,0.000000,0.000000}%
\pgfsetstrokecolor{textcolor}%
\pgfsetfillcolor{textcolor}%
\pgftext[x=3.749244in,y=3.051231in,,top]{\color{textcolor}\rmfamily\fontsize{10.000000}{12.000000}\selectfont 3.00}%
\end{pgfscope}%
\begin{pgfscope}%
\pgfpathrectangle{\pgfqpoint{0.100000in}{0.100000in}}{\pgfqpoint{3.696000in}{3.696000in}}%
\pgfusepath{clip}%
\pgfsetrectcap%
\pgfsetroundjoin%
\pgfsetlinewidth{1.003750pt}%
\definecolor{currentstroke}{rgb}{0, 0.4471, 0.7412}%
\pgfsetstrokecolor{currentstroke}%
\pgfsetdash{}{0pt}%
\pgfpathmoveto{\pgfqpoint{1.762148in}{1.260439in}}%
\pgfpathlineto{\pgfqpoint{1.760790in}{1.818810in}}%
\pgfpathlineto{\pgfqpoint{1.762565in}{1.891610in}}%
\pgfpathlineto{\pgfqpoint{1.765279in}{1.923760in}}%
\pgfpathlineto{\pgfqpoint{1.769081in}{1.944007in}}%
\pgfpathlineto{\pgfqpoint{1.773257in}{1.956247in}}%
\pgfpathlineto{\pgfqpoint{1.779239in}{1.966798in}}%
\pgfpathlineto{\pgfqpoint{1.785826in}{1.974086in}}%
\pgfpathlineto{\pgfqpoint{1.794579in}{1.980376in}}%
\pgfpathlineto{\pgfqpoint{1.806194in}{1.985666in}}%
\pgfpathlineto{\pgfqpoint{1.821577in}{1.989928in}}%
\pgfpathlineto{\pgfqpoint{1.841892in}{1.993099in}}%
\pgfpathlineto{\pgfqpoint{1.868592in}{1.995085in}}%
\pgfpathlineto{\pgfqpoint{1.913602in}{1.995703in}}%
\pgfpathlineto{\pgfqpoint{1.974958in}{1.993989in}}%
\pgfpathlineto{\pgfqpoint{2.074175in}{1.988579in}}%
\pgfpathlineto{\pgfqpoint{2.224457in}{1.977942in}}%
\pgfpathlineto{\pgfqpoint{2.568163in}{1.950660in}}%
\pgfpathlineto{\pgfqpoint{2.865804in}{1.927897in}}%
\pgfpathlineto{\pgfqpoint{2.980428in}{1.921153in}}%
\pgfpathlineto{\pgfqpoint{3.022994in}{1.920572in}}%
\pgfpathlineto{\pgfqpoint{3.033827in}{1.922230in}}%
\pgfpathlineto{\pgfqpoint{3.034653in}{1.923986in}}%
\pgfpathlineto{\pgfqpoint{3.034112in}{1.925191in}}%
\pgfpathlineto{\pgfqpoint{3.034112in}{1.925191in}}%
\pgfusepath{stroke}%
\end{pgfscope}%
\begin{pgfscope}%
\pgfpathrectangle{\pgfqpoint{0.100000in}{0.100000in}}{\pgfqpoint{3.696000in}{3.696000in}}%
\pgfusepath{clip}%
\pgfsetrectcap%
\pgfsetroundjoin%
\pgfsetlinewidth{1.003750pt}%
\definecolor{currentstroke}{rgb}{0, 0.4471, 0.7412}%
\pgfsetstrokecolor{currentstroke}%
\pgfsetdash{}{0pt}%
\pgfpathmoveto{\pgfqpoint{3.034112in}{1.925191in}}%
\pgfusepath{stroke}%
\end{pgfscope}%
\begin{pgfscope}%
\pgfpathrectangle{\pgfqpoint{0.100000in}{0.100000in}}{\pgfqpoint{3.696000in}{3.696000in}}%
\pgfusepath{clip}%
\pgfsetbuttcap%
\pgfsetroundjoin%
\definecolor{currentfill}{rgb}{0, 0.4471, 0.7412}%
\pgfsetfillcolor{currentfill}%
\pgfsetlinewidth{1.003750pt}%
\definecolor{currentstroke}{rgb}{0, 0.4471, 0.7412}%
\pgfsetstrokecolor{currentstroke}%
\pgfsetdash{}{0pt}%
\pgfsys@defobject{currentmarker}{\pgfqpoint{-0.041667in}{-0.041667in}}{\pgfqpoint{0.041667in}{0.041667in}}{%
\pgfpathmoveto{\pgfqpoint{-0.041667in}{-0.041667in}}%
\pgfpathlineto{\pgfqpoint{0.041667in}{0.041667in}}%
\pgfpathmoveto{\pgfqpoint{-0.041667in}{0.041667in}}%
\pgfpathlineto{\pgfqpoint{0.041667in}{-0.041667in}}%
\pgfusepath{stroke,fill}%
}%
\begin{pgfscope}%
\pgfsys@transformshift{3.034112in}{1.925191in}%
\pgfsys@useobject{currentmarker}{}%
\end{pgfscope}%
\end{pgfscope}%
\begin{pgfscope}%
\pgfpathrectangle{\pgfqpoint{0.100000in}{0.100000in}}{\pgfqpoint{3.696000in}{3.696000in}}%
\pgfusepath{clip}%
\pgfsetrectcap%
\pgfsetroundjoin%
\pgfsetlinewidth{1.003750pt}%
\definecolor{currentstroke}{rgb}{0, 0.4471, 0.7412}%
\pgfsetstrokecolor{currentstroke}%
\pgfsetdash{}{0pt}%
\pgfpathmoveto{\pgfqpoint{1.762148in}{1.260439in}}%
\pgfusepath{stroke}%
\end{pgfscope}%
\begin{pgfscope}%
\pgfpathrectangle{\pgfqpoint{0.100000in}{0.100000in}}{\pgfqpoint{3.696000in}{3.696000in}}%
\pgfusepath{clip}%
\pgfsetbuttcap%
\pgfsetroundjoin%
\definecolor{currentfill}{rgb}{0, 0.4471, 0.7412}%
\pgfsetfillcolor{currentfill}%
\pgfsetlinewidth{1.003750pt}%
\definecolor{currentstroke}{rgb}{0, 0.4471, 0.7412}%
\pgfsetstrokecolor{currentstroke}%
\pgfsetdash{}{0pt}%
\pgfsys@defobject{currentmarker}{\pgfqpoint{-0.034722in}{-0.034722in}}{\pgfqpoint{0.034722in}{0.034722in}}{%
\pgfpathmoveto{\pgfqpoint{0.000000in}{-0.034722in}}%
\pgfpathcurveto{\pgfqpoint{0.009208in}{-0.034722in}}{\pgfqpoint{0.018041in}{-0.031064in}}{\pgfqpoint{0.024552in}{-0.024552in}}%
\pgfpathcurveto{\pgfqpoint{0.031064in}{-0.018041in}}{\pgfqpoint{0.034722in}{-0.009208in}}{\pgfqpoint{0.034722in}{0.000000in}}%
\pgfpathcurveto{\pgfqpoint{0.034722in}{0.009208in}}{\pgfqpoint{0.031064in}{0.018041in}}{\pgfqpoint{0.024552in}{0.024552in}}%
\pgfpathcurveto{\pgfqpoint{0.018041in}{0.031064in}}{\pgfqpoint{0.009208in}{0.034722in}}{\pgfqpoint{0.000000in}{0.034722in}}%
\pgfpathcurveto{\pgfqpoint{-0.009208in}{0.034722in}}{\pgfqpoint{-0.018041in}{0.031064in}}{\pgfqpoint{-0.024552in}{0.024552in}}%
\pgfpathcurveto{\pgfqpoint{-0.031064in}{0.018041in}}{\pgfqpoint{-0.034722in}{0.009208in}}{\pgfqpoint{-0.034722in}{0.000000in}}%
\pgfpathcurveto{\pgfqpoint{-0.034722in}{-0.009208in}}{\pgfqpoint{-0.031064in}{-0.018041in}}{\pgfqpoint{-0.024552in}{-0.024552in}}%
\pgfpathcurveto{\pgfqpoint{-0.018041in}{-0.031064in}}{\pgfqpoint{-0.009208in}{-0.034722in}}{\pgfqpoint{0.000000in}{-0.034722in}}%
\pgfpathlineto{\pgfqpoint{0.000000in}{-0.034722in}}%
\pgfpathclose%
\pgfusepath{stroke,fill}%
}%
\begin{pgfscope}%
\pgfsys@transformshift{1.762148in}{1.260439in}%
\pgfsys@useobject{currentmarker}{}%
\end{pgfscope}%
\end{pgfscope}%
\begin{pgfscope}%
\pgfpathrectangle{\pgfqpoint{0.100000in}{0.100000in}}{\pgfqpoint{3.696000in}{3.696000in}}%
\pgfusepath{clip}%
\pgfsetbuttcap%
\pgfsetroundjoin%
\pgfsetlinewidth{1.003750pt}%
\definecolor{currentstroke}{rgb}{0.8510,0.3255,0.0980}%
\pgfsetstrokecolor{currentstroke}%
\pgfsetdash{{6.400000pt}{1.600000pt}{1.000000pt}{1.600000pt}}{0.000000pt}%
\pgfpathmoveto{\pgfqpoint{1.758681in}{2.048683in}}%
\pgfpathlineto{\pgfqpoint{1.757593in}{2.054403in}}%
\pgfpathlineto{\pgfqpoint{1.749118in}{2.062671in}}%
\pgfpathlineto{\pgfqpoint{1.685252in}{2.120655in}}%
\pgfpathlineto{\pgfqpoint{1.627416in}{2.174839in}}%
\pgfpathlineto{\pgfqpoint{1.579118in}{2.223250in}}%
\pgfpathlineto{\pgfqpoint{1.539265in}{2.266378in}}%
\pgfpathlineto{\pgfqpoint{1.505253in}{2.306452in}}%
\pgfpathlineto{\pgfqpoint{1.478966in}{2.340584in}}%
\pgfpathlineto{\pgfqpoint{1.457834in}{2.371541in}}%
\pgfpathlineto{\pgfqpoint{1.443695in}{2.396083in}}%
\pgfpathlineto{\pgfqpoint{1.435694in}{2.414456in}}%
\pgfpathlineto{\pgfqpoint{1.433179in}{2.425417in}}%
\pgfpathlineto{\pgfqpoint{1.433891in}{2.431300in}}%
\pgfpathlineto{\pgfqpoint{1.436217in}{2.433658in}}%
\pgfpathlineto{\pgfqpoint{1.436909in}{2.433863in}}%
\pgfpathlineto{\pgfqpoint{1.436909in}{2.433863in}}%
\pgfusepath{stroke}%
\end{pgfscope}%
\begin{pgfscope}%
\pgfpathrectangle{\pgfqpoint{0.100000in}{0.100000in}}{\pgfqpoint{3.696000in}{3.696000in}}%
\pgfusepath{clip}%
\pgfsetbuttcap%
\pgfsetroundjoin%
\pgfsetlinewidth{1.003750pt}%
\definecolor{currentstroke}{rgb}{0.8510,0.3255,0.0980}%
\pgfsetstrokecolor{currentstroke}%
\pgfsetdash{{6.400000pt}{1.600000pt}{1.000000pt}{1.600000pt}}{0.000000pt}%
\pgfpathmoveto{\pgfqpoint{1.436909in}{2.433863in}}%
\pgfusepath{stroke}%
\end{pgfscope}%
\begin{pgfscope}%
\pgfpathrectangle{\pgfqpoint{0.100000in}{0.100000in}}{\pgfqpoint{3.696000in}{3.696000in}}%
\pgfusepath{clip}%
\pgfsetbuttcap%
\pgfsetroundjoin%
\definecolor{currentfill}{rgb}{0.8510,0.3255,0.0980}%
\pgfsetfillcolor{currentfill}%
\pgfsetlinewidth{1.003750pt}%
\definecolor{currentstroke}{rgb}{0.8510,0.3255,0.0980}%
\pgfsetstrokecolor{currentstroke}%
\pgfsetdash{}{0pt}%
\pgfsys@defobject{currentmarker}{\pgfqpoint{-0.041667in}{-0.041667in}}{\pgfqpoint{0.041667in}{0.041667in}}{%
\pgfpathmoveto{\pgfqpoint{-0.041667in}{-0.041667in}}%
\pgfpathlineto{\pgfqpoint{0.041667in}{0.041667in}}%
\pgfpathmoveto{\pgfqpoint{-0.041667in}{0.041667in}}%
\pgfpathlineto{\pgfqpoint{0.041667in}{-0.041667in}}%
\pgfusepath{stroke,fill}%
}%
\begin{pgfscope}%
\pgfsys@transformshift{1.436909in}{2.433863in}%
\pgfsys@useobject{currentmarker}{}%
\end{pgfscope}%
\end{pgfscope}%
\begin{pgfscope}%
\pgfpathrectangle{\pgfqpoint{0.100000in}{0.100000in}}{\pgfqpoint{3.696000in}{3.696000in}}%
\pgfusepath{clip}%
\pgfsetbuttcap%
\pgfsetroundjoin%
\pgfsetlinewidth{1.003750pt}%
\definecolor{currentstroke}{rgb}{0.8510,0.3255,0.0980}%
\pgfsetstrokecolor{currentstroke}%
\pgfsetdash{{6.400000pt}{1.600000pt}{1.000000pt}{1.600000pt}}{0.000000pt}%
\pgfpathmoveto{\pgfqpoint{1.758681in}{2.048683in}}%
\pgfusepath{stroke}%
\end{pgfscope}%
\begin{pgfscope}%
\pgfpathrectangle{\pgfqpoint{0.100000in}{0.100000in}}{\pgfqpoint{3.696000in}{3.696000in}}%
\pgfusepath{clip}%
\pgfsetbuttcap%
\pgfsetroundjoin%
\definecolor{currentfill}{rgb}{0.8510,0.3255,0.0980}%
\pgfsetfillcolor{currentfill}%
\pgfsetlinewidth{1.003750pt}%
\definecolor{currentstroke}{rgb}{0.8510,0.3255,0.0980}%
\pgfsetstrokecolor{currentstroke}%
\pgfsetdash{}{0pt}%
\pgfsys@defobject{currentmarker}{\pgfqpoint{-0.034722in}{-0.034722in}}{\pgfqpoint{0.034722in}{0.034722in}}{%
\pgfpathmoveto{\pgfqpoint{0.000000in}{-0.034722in}}%
\pgfpathcurveto{\pgfqpoint{0.009208in}{-0.034722in}}{\pgfqpoint{0.018041in}{-0.031064in}}{\pgfqpoint{0.024552in}{-0.024552in}}%
\pgfpathcurveto{\pgfqpoint{0.031064in}{-0.018041in}}{\pgfqpoint{0.034722in}{-0.009208in}}{\pgfqpoint{0.034722in}{0.000000in}}%
\pgfpathcurveto{\pgfqpoint{0.034722in}{0.009208in}}{\pgfqpoint{0.031064in}{0.018041in}}{\pgfqpoint{0.024552in}{0.024552in}}%
\pgfpathcurveto{\pgfqpoint{0.018041in}{0.031064in}}{\pgfqpoint{0.009208in}{0.034722in}}{\pgfqpoint{0.000000in}{0.034722in}}%
\pgfpathcurveto{\pgfqpoint{-0.009208in}{0.034722in}}{\pgfqpoint{-0.018041in}{0.031064in}}{\pgfqpoint{-0.024552in}{0.024552in}}%
\pgfpathcurveto{\pgfqpoint{-0.031064in}{0.018041in}}{\pgfqpoint{-0.034722in}{0.009208in}}{\pgfqpoint{-0.034722in}{0.000000in}}%
\pgfpathcurveto{\pgfqpoint{-0.034722in}{-0.009208in}}{\pgfqpoint{-0.031064in}{-0.018041in}}{\pgfqpoint{-0.024552in}{-0.024552in}}%
\pgfpathcurveto{\pgfqpoint{-0.018041in}{-0.031064in}}{\pgfqpoint{-0.009208in}{-0.034722in}}{\pgfqpoint{0.000000in}{-0.034722in}}%
\pgfpathlineto{\pgfqpoint{0.000000in}{-0.034722in}}%
\pgfpathclose%
\pgfusepath{stroke,fill}%
}%
\begin{pgfscope}%
\pgfsys@transformshift{1.758681in}{2.048683in}%
\pgfsys@useobject{currentmarker}{}%
\end{pgfscope}%
\end{pgfscope}%
\begin{pgfscope}%
\pgfpathrectangle{\pgfqpoint{0.100000in}{0.100000in}}{\pgfqpoint{3.696000in}{3.696000in}}%
\pgfusepath{clip}%
\pgfsetbuttcap%
\pgfsetroundjoin%
\pgfsetlinewidth{1.003750pt}%
\definecolor{currentstroke}{rgb}{0.4667,0.6745,0.1882}%
\pgfsetstrokecolor{currentstroke}%
\pgfsetdash{{3.700000pt}{1.600000pt}}{0.000000pt}%
\pgfpathmoveto{\pgfqpoint{1.755131in}{2.860188in}}%
\pgfpathlineto{\pgfqpoint{1.756630in}{2.264142in}}%
\pgfpathlineto{\pgfqpoint{1.754841in}{2.197396in}}%
\pgfpathlineto{\pgfqpoint{1.752164in}{2.167405in}}%
\pgfpathlineto{\pgfqpoint{1.748451in}{2.147840in}}%
\pgfpathlineto{\pgfqpoint{1.743397in}{2.133008in}}%
\pgfpathlineto{\pgfqpoint{1.737196in}{2.121596in}}%
\pgfpathlineto{\pgfqpoint{1.728384in}{2.110666in}}%
\pgfpathlineto{\pgfqpoint{1.718749in}{2.101987in}}%
\pgfpathlineto{\pgfqpoint{1.706042in}{2.093084in}}%
\pgfpathlineto{\pgfqpoint{1.684415in}{2.081167in}}%
\pgfpathlineto{\pgfqpoint{1.654298in}{2.067689in}}%
\pgfpathlineto{\pgfqpoint{1.613126in}{2.051820in}}%
\pgfpathlineto{\pgfqpoint{1.545975in}{2.028523in}}%
\pgfpathlineto{\pgfqpoint{1.408584in}{1.983948in}}%
\pgfpathlineto{\pgfqpoint{1.143867in}{1.898386in}}%
\pgfpathlineto{\pgfqpoint{1.004439in}{1.850941in}}%
\pgfpathlineto{\pgfqpoint{0.913344in}{1.817828in}}%
\pgfpathlineto{\pgfqpoint{0.854241in}{1.794175in}}%
\pgfpathlineto{\pgfqpoint{0.823611in}{1.779725in}}%
\pgfpathlineto{\pgfqpoint{0.811593in}{1.771873in}}%
\pgfpathlineto{\pgfqpoint{0.808779in}{1.767765in}}%
\pgfpathlineto{\pgfqpoint{0.808689in}{1.767051in}}%
\pgfpathlineto{\pgfqpoint{0.808689in}{1.767051in}}%
\pgfusepath{stroke}%
\end{pgfscope}%
\begin{pgfscope}%
\pgfpathrectangle{\pgfqpoint{0.100000in}{0.100000in}}{\pgfqpoint{3.696000in}{3.696000in}}%
\pgfusepath{clip}%
\pgfsetbuttcap%
\pgfsetroundjoin%
\pgfsetlinewidth{1.003750pt}%
\definecolor{currentstroke}{rgb}{0.4667,0.6745,0.1882}%
\pgfsetstrokecolor{currentstroke}%
\pgfsetdash{{3.700000pt}{1.600000pt}}{0.000000pt}%
\pgfpathmoveto{\pgfqpoint{0.808689in}{1.767051in}}%
\pgfusepath{stroke}%
\end{pgfscope}%
\begin{pgfscope}%
\pgfpathrectangle{\pgfqpoint{0.100000in}{0.100000in}}{\pgfqpoint{3.696000in}{3.696000in}}%
\pgfusepath{clip}%
\pgfsetbuttcap%
\pgfsetroundjoin%
\definecolor{currentfill}{rgb}{0.4667,0.6745,0.1882}%
\pgfsetfillcolor{currentfill}%
\pgfsetlinewidth{1.003750pt}%
\definecolor{currentstroke}{rgb}{0.4667,0.6745,0.1882}%
\pgfsetstrokecolor{currentstroke}%
\pgfsetdash{}{0pt}%
\pgfsys@defobject{currentmarker}{\pgfqpoint{-0.041667in}{-0.041667in}}{\pgfqpoint{0.041667in}{0.041667in}}{%
\pgfpathmoveto{\pgfqpoint{-0.041667in}{-0.041667in}}%
\pgfpathlineto{\pgfqpoint{0.041667in}{0.041667in}}%
\pgfpathmoveto{\pgfqpoint{-0.041667in}{0.041667in}}%
\pgfpathlineto{\pgfqpoint{0.041667in}{-0.041667in}}%
\pgfusepath{stroke,fill}%
}%
\begin{pgfscope}%
\pgfsys@transformshift{0.808689in}{1.767051in}%
\pgfsys@useobject{currentmarker}{}%
\end{pgfscope}%
\end{pgfscope}%
\begin{pgfscope}%
\pgfpathrectangle{\pgfqpoint{0.100000in}{0.100000in}}{\pgfqpoint{3.696000in}{3.696000in}}%
\pgfusepath{clip}%
\pgfsetbuttcap%
\pgfsetroundjoin%
\pgfsetlinewidth{1.003750pt}%
\definecolor{currentstroke}{rgb}{0.4667,0.6745,0.1882}%
\pgfsetstrokecolor{currentstroke}%
\pgfsetdash{{3.700000pt}{1.600000pt}}{0.000000pt}%
\pgfpathmoveto{\pgfqpoint{1.755131in}{2.860188in}}%
\pgfusepath{stroke}%
\end{pgfscope}%
\begin{pgfscope}%
\pgfpathrectangle{\pgfqpoint{0.100000in}{0.100000in}}{\pgfqpoint{3.696000in}{3.696000in}}%
\pgfusepath{clip}%
\pgfsetbuttcap%
\pgfsetroundjoin%
\definecolor{currentfill}{rgb}{0.4667,0.6745,0.1882}%
\pgfsetfillcolor{currentfill}%
\pgfsetlinewidth{1.003750pt}%
\definecolor{currentstroke}{rgb}{0.4667,0.6745,0.1882}%
\pgfsetstrokecolor{currentstroke}%
\pgfsetdash{}{0pt}%
\pgfsys@defobject{currentmarker}{\pgfqpoint{-0.034722in}{-0.034722in}}{\pgfqpoint{0.034722in}{0.034722in}}{%
\pgfpathmoveto{\pgfqpoint{0.000000in}{-0.034722in}}%
\pgfpathcurveto{\pgfqpoint{0.009208in}{-0.034722in}}{\pgfqpoint{0.018041in}{-0.031064in}}{\pgfqpoint{0.024552in}{-0.024552in}}%
\pgfpathcurveto{\pgfqpoint{0.031064in}{-0.018041in}}{\pgfqpoint{0.034722in}{-0.009208in}}{\pgfqpoint{0.034722in}{0.000000in}}%
\pgfpathcurveto{\pgfqpoint{0.034722in}{0.009208in}}{\pgfqpoint{0.031064in}{0.018041in}}{\pgfqpoint{0.024552in}{0.024552in}}%
\pgfpathcurveto{\pgfqpoint{0.018041in}{0.031064in}}{\pgfqpoint{0.009208in}{0.034722in}}{\pgfqpoint{0.000000in}{0.034722in}}%
\pgfpathcurveto{\pgfqpoint{-0.009208in}{0.034722in}}{\pgfqpoint{-0.018041in}{0.031064in}}{\pgfqpoint{-0.024552in}{0.024552in}}%
\pgfpathcurveto{\pgfqpoint{-0.031064in}{0.018041in}}{\pgfqpoint{-0.034722in}{0.009208in}}{\pgfqpoint{-0.034722in}{0.000000in}}%
\pgfpathcurveto{\pgfqpoint{-0.034722in}{-0.009208in}}{\pgfqpoint{-0.031064in}{-0.018041in}}{\pgfqpoint{-0.024552in}{-0.024552in}}%
\pgfpathcurveto{\pgfqpoint{-0.018041in}{-0.031064in}}{\pgfqpoint{-0.009208in}{-0.034722in}}{\pgfqpoint{0.000000in}{-0.034722in}}%
\pgfpathlineto{\pgfqpoint{0.000000in}{-0.034722in}}%
\pgfpathclose%
\pgfusepath{stroke,fill}%
}%
\begin{pgfscope}%
\pgfsys@transformshift{1.755131in}{2.860188in}%
\pgfsys@useobject{currentmarker}{}%
\end{pgfscope}%
\end{pgfscope}%
\begin{pgfscope}%
\pgfsetbuttcap%
\pgfsetmiterjoin%
\definecolor{currentfill}{rgb}{1.000000,1.000000,1.000000}%
\pgfsetfillcolor{currentfill}%
\pgfsetfillopacity{0.800000}%
\pgfsetlinewidth{1.003750pt}%
\definecolor{currentstroke}{rgb}{0.800000,0.800000,0.800000}%
\pgfsetstrokecolor{currentstroke}%
\pgfsetstrokeopacity{0.800000}%
\pgfsetdash{}{0pt}%
\pgfpathmoveto{\pgfqpoint{2.772567in}{3.103890in}}%
\pgfpathlineto{\pgfqpoint{3.698778in}{3.103890in}}%
\pgfpathquadraticcurveto{\pgfqpoint{3.726556in}{3.103890in}}{\pgfqpoint{3.726556in}{3.131668in}}%
\pgfpathlineto{\pgfqpoint{3.726556in}{3.698778in}}%
\pgfpathquadraticcurveto{\pgfqpoint{3.726556in}{3.726556in}}{\pgfqpoint{3.698778in}{3.726556in}}%
\pgfpathlineto{\pgfqpoint{2.772567in}{3.726556in}}%
\pgfpathquadraticcurveto{\pgfqpoint{2.744789in}{3.726556in}}{\pgfqpoint{2.744789in}{3.698778in}}%
\pgfpathlineto{\pgfqpoint{2.744789in}{3.131668in}}%
\pgfpathquadraticcurveto{\pgfqpoint{2.744789in}{3.103890in}}{\pgfqpoint{2.772567in}{3.103890in}}%
\pgfpathlineto{\pgfqpoint{2.772567in}{3.103890in}}%
\pgfpathclose%
\pgfusepath{stroke,fill}%
\end{pgfscope}%
\begin{pgfscope}%
\pgfsetrectcap%
\pgfsetroundjoin%
\pgfsetlinewidth{1.003750pt}%
\definecolor{currentstroke}{rgb}{0, 0.4471, 0.7412}%
\pgfsetstrokecolor{currentstroke}%
\pgfsetdash{}{0pt}%
\pgfpathmoveto{\pgfqpoint{2.800344in}{3.622389in}}%
\pgfpathlineto{\pgfqpoint{2.939233in}{3.622389in}}%
\pgfpathlineto{\pgfqpoint{3.078122in}{3.622389in}}%
\pgfusepath{stroke}%
\end{pgfscope}%
\begin{pgfscope}%
\definecolor{textcolor}{rgb}{0.000000,0.000000,0.000000}%
\pgfsetstrokecolor{textcolor}%
\pgfsetfillcolor{textcolor}%
\pgftext[x=3.189233in,y=3.573778in,left,base]{\color{textcolor}\rmfamily\fontsize{10.000000}{12.000000}\selectfont Agent \(\displaystyle 1\)}%
\end{pgfscope}%
\begin{pgfscope}%
\pgfsetbuttcap%
\pgfsetroundjoin%
\pgfsetlinewidth{1.003750pt}%
\definecolor{currentstroke}{rgb}{0.8510,0.3255,0.0980}%
\pgfsetstrokecolor{currentstroke}%
\pgfsetdash{{6.400000pt}{1.600000pt}{1.000000pt}{1.600000pt}}{0.000000pt}%
\pgfpathmoveto{\pgfqpoint{2.800344in}{3.428723in}}%
\pgfpathlineto{\pgfqpoint{2.939233in}{3.428723in}}%
\pgfpathlineto{\pgfqpoint{3.078122in}{3.428723in}}%
\pgfusepath{stroke}%
\end{pgfscope}%
\begin{pgfscope}%
\definecolor{textcolor}{rgb}{0.000000,0.000000,0.000000}%
\pgfsetstrokecolor{textcolor}%
\pgfsetfillcolor{textcolor}%
\pgftext[x=3.189233in,y=3.380112in,left,base]{\color{textcolor}\rmfamily\fontsize{10.000000}{12.000000}\selectfont Agent \(\displaystyle 2\)}%
\end{pgfscope}%
\begin{pgfscope}%
\pgfsetbuttcap%
\pgfsetroundjoin%
\pgfsetlinewidth{1.003750pt}%
\definecolor{currentstroke}{rgb}{0.4667,0.6745,0.1882}%
\pgfsetstrokecolor{currentstroke}%
\pgfsetdash{{3.700000pt}{1.600000pt}}{0.000000pt}%
\pgfpathmoveto{\pgfqpoint{2.800344in}{3.235056in}}%
\pgfpathlineto{\pgfqpoint{2.939233in}{3.235056in}}%
\pgfpathlineto{\pgfqpoint{3.078122in}{3.235056in}}%
\pgfusepath{stroke}%
\end{pgfscope}%
\begin{pgfscope}%
\definecolor{textcolor}{rgb}{0.000000,0.000000,0.000000}%
\pgfsetstrokecolor{textcolor}%
\pgfsetfillcolor{textcolor}%
\pgftext[x=3.189233in,y=3.186445in,left,base]{\color{textcolor}\rmfamily\fontsize{10.000000}{12.000000}\selectfont Agent \(\displaystyle 3\)}%
\end{pgfscope}%
\end{pgfpicture}%
\makeatother%
\endgroup%

%% file: 02_dmpc_for_cooperation_journal_extended_version.bbl
\begin{thebibliography}{10}

\bibitem{Smith2005}
R.~S. Smith and F.~Y. Hadaegh.
\newblock Control of {D}eep-{S}pace {F}ormation-{F}lying {S}pacecraft; {R}elative {S}ensing and {S}witched {I}nformation.
\newblock {\em J. Guid., Control, and Dyn.}, 28(1):106--114, 2005.

\bibitem{Dimarogonas2009}
D.~V. Dimarogonas and K.~H. Johansson.
\newblock {Further results on the stability of distance-based multi-robot formations}.
\newblock In {\em Proc. Am. Control Conf. (ACC)}, pages 2972--2977, St. Louis, USA, 2009.

\bibitem{Ritz2013}
R.~Ritz and R.~D'Andrea.
\newblock Carrying a flexible payload with multiple flying vehicles.
\newblock In {\em Proc. IEEE/RSJ Int. Conf. Intell. Robots Syst. (IROS)}, pages 3465--3471, Tokyo, Japan, 2013.

\bibitem{Zheng2017}
Y.~Zheng, S.~E. Li, K.~Li, F.~Borrelli, and J.~K. Hedrick.
\newblock {Distributed Model Predictive Control for Heterogeneous Vehicle Platoons Under Unidirectional Topologies}.
\newblock {\em {IEEE Trans. Control Syst. Technol}}, 25(3):899--910, 2017.

\bibitem{Fonseca2019}
J.~Fonseca, J.~Wei, K.~H. Johansson, and T.~A. Johansen.
\newblock Cooperative decentralised circumnavigation with application to algal bloom tracking.
\newblock In {\em Proc. IEEE/RSJ Int. Conf. Intell. Robots Syst. (IROS)}, Macau, China, 2019.

\bibitem{Klausen2020}
K.~Klausen, C.~Meissen, T.~I. Fossen, M.~Arcak, and T.~A. Johansen.
\newblock Cooperative control for multirotors transporting an unknown suspended load under environmental disturbances.
\newblock {\em {IEEE} Trans. Control Syst. Technol.}, 28(2):653--660, 2020.

\bibitem{Ryan2004}
A.~Ryan, M.~Zennaro, A.~Howell, R.~Sengupta, and J.~K. Hedrick.
\newblock An overview of emerging results in cooperative {UAV} control.
\newblock In {\em Proc. 43rd {IEEE} Conf. Decision Control (CDC)}, pages 602--607, Nassau, Bahamas, 2004.

\bibitem{Cao2013}
Y.~Cao, W.~Yu, W.~Ren, and G.~Chen.
\newblock An overview of recent progress in the study of distributed multi-agent coordination.
\newblock {\em {IEEE} Transactions on Industrial Informatics}, 9(1):427--438, 2013.

\bibitem{Yu2021}
S.~Yu, M.~Hirche, Y.~Huang, H.~Chen, and F.~Allg{\"o}wer.
\newblock Model predictive control for autonomous ground vehicles: a review.
\newblock {\em Auton. Intell. Syst.}, 1(1), 2021.

\bibitem{OlfatiSaber2008}
R.~Olfati-Saber and N.~F. Sandell.
\newblock Distributed tracking in sensor networks with limited sensing range.
\newblock In {\em Proc. Am. Control Conf. (ACC)}, pages 3157--3162, Seattle, USA, 2008.

\bibitem{Tomlin1998}
C.~Tomlin, G.J. Pappas, and S.~Sastry.
\newblock Conflict resolution for air traffic management: a study in multiagent hybrid systems.
\newblock {\em {IEEE} Trans. Autom. Control}, 43(4):509--521, 1998.

\bibitem{Scattolini2009}
R.~Scattolini.
\newblock {Architectures for distributed and hierarchical Model Predictive Control -- A review}.
\newblock {\em {J. Process Control}}, 19(5):723--731, 2009.

\bibitem{Mueller2017b}
M.~A. M{\"u}ller and F.~Allg{\"o}wer.
\newblock {Economic and Distributed Model Predictive Control: Recent Developments in Optimization-Based Control}.
\newblock {\em {SICE JCMSI}}, 10(2):39--52, 2017.

\bibitem{OlfatiSaber2007}
R.~Olfati-Saber, J.~A. Fax, and R.~M. Murray.
\newblock {Consensus and Cooperation in Networked Multi-Agent Systems}.
\newblock {\em {Proc. IEEE}}, 95(1):215--233, 2007.

\bibitem{Cheng2015}
Z.~Cheng, H.-T. Zhang, M.-C. Fan, and G.~Chen.
\newblock {Distributed Consensus of Multi-Agent Systems With Input Constraints: A Model Predictive Control Approach}.
\newblock {\em {IEEE Trans. Circuits Syst. I, Reg. Papers}}, 62(3):825--834, 2015.

\bibitem{Li2015}
H.~Li and W.~Yan.
\newblock {Receding horizon control based consensus scheme in general linear multi-agent systems}.
\newblock {\em {Automatica}}, 56:12--18, 2015.

\bibitem{Johansson2006}
B.~Johansson, A.~Speranzon, M.~Johansson, and K.~H. Johansson.
\newblock {Distributed Model Predictive Consensus}.
\newblock In {\em Proc. 17th Int. Symp. Math. Theory Netw. Syst.}, pages 2438--2444, Kyoto, Japan, 2006.

\bibitem{Hirche2020}
M.~Hirche, P.~N. K{\"o}hler, M.~A. M{\"u}ller, and F.~Allg{\"o}wer.
\newblock {Distributed Model Predictive Control for Consensus of Constrained Heterogeneous Linear Systems}.
\newblock In {\em Proc. 59th IEEE Conf. Decision Control (CDC)}, pages 1248--1253, Jeju Island, Republic of Korea, 2020.

\bibitem{Muller.2012}
M.~A. M{\"u}ller, M.~Reble, and F.~Allg{\"o}wer.
\newblock {Cooperative control of dynamically decoupled systems via distributed model predictive control}.
\newblock {\em {Int. J. Robust Nonlinear Control}}, 22(12):1376--1397, 2012.

\bibitem{Limon.2008}
D.~Limon, I.~Alvarado, T.~Alamo, and E.~F. Camacho.
\newblock {MPC for tracking piecewise constant references for constrained linear systems}.
\newblock {\em {Automatica}}, 44(9):2382--2387, 2008.

\bibitem{Limon2018}
D.~Lim{\'o}n, A.~Ferramosca, I.~Alvarado, and T.~Alamo.
\newblock {Nonlinear MPC for Tracking Piece-Wise Constant Reference Signals}.
\newblock {\em {IEEE Trans. Autom. Control}}, 63(11):3735--3750, 2018.

\bibitem{Ferramosca2011}
A.~Ferramosca, D.~Lim{\'o}n, J.~B. Rawlings, and E.~F. Camacho.
\newblock {Cooperative distributed MPC for tracking}.
\newblock {\em {IFAC Proceedings Volumes}}, 44(1):1584--1589, 2011.

\bibitem{Conte2016}
C.~Conte, C.~N. Jones, M.~Morari, and M.~N. Zeilinger.
\newblock {Distributed synthesis and stability of cooperative distributed model predictive control for linear systems}.
\newblock {\em {Automatica}}, 69:117--125, 2016.

\bibitem{Carron2020a}
A.~Carron and M.~N. Zeilinger.
\newblock {Model Predictive Coverage Control}.
\newblock {\em {IFAC-PapersOnLine}}, 53(2):6107--6112, 2020.

\bibitem{Santana2022}
D.~D. Santana, D.~Odloak, T.~L.M. Santos, and M.~A.~F. Martins.
\newblock {A stabilizing cooperative-distributed gradient-based economic model predictive control strategy for constrained linear systems}.
\newblock {\em {J. Process Control}}, 112:36--48, 2022.

\bibitem{Koehler2024extended}
M.~K\"ohler, M.~A. M\"uller, and F.~Allg\"ower.
\newblock {Distributed MPC for Self-Organized Cooperation of Multiagent Systems}.
\newblock {\em IEEE Trans. Autom. Control}, 2024.

\bibitem{Koehler2020b}
J.~K{\"o}hler, M.~A. M{\"u}ller, and F.~Allg{\"o}wer.
\newblock {A nonlinear tracking model predictive control scheme for dynamic target signals}.
\newblock {\em {Automatica}}, 118:109030, 2020.

\bibitem{Kohler.2020}
J.~K{\"o}hler, M.~A. M{\"u}ller, and F.~Allg{\"o}wer.
\newblock {A Nonlinear Model Predictive Control Framework Using Reference Generic Terminal Ingredients}.
\newblock {\em {IEEE Trans. Autom. Control}}, 65(8):3576--3583, 2020.

\bibitem{Muller.2014}
M.~A. M{\"u}ller.
\newblock {\em {Distributed and Economic Model Predictive Control: Beyond Setpoint Stabilization}}.
\newblock {LOGOS Verlag BERLIN}, Berlin, 2014.

\bibitem{Trodden.2013}
P.~Trodden and A.~Richards.
\newblock {Cooperative distributed MPC of linear systems with coupled constraints}.
\newblock {\em {Automatica}}, 49(2):479--487, 2013.

\bibitem{Rawlings.2020}
J.~Rawlings, D.~Mayne, and M.~Diehl.
\newblock {\em {Model Predictive Control: Theory, Computation and Design}}.
\newblock {Nob Hill Publishing LLC}, Santa Barbara, 2nd edition, 2020.

\bibitem{Bertsekas.2016}
D.~P. Bertsekas.
\newblock {\em {N}onlinear {P}rogramming}.
\newblock {Athena Scientific}, Belmont, Mass., 3rd edition, 2016.

\bibitem{Bertsekas.1997}
D.~P. Bertsekas and J.~N. Tsitsiklis.
\newblock {\em {P}arallel and {D}istributed {C}omputation: {N}umerical {M}ethods}.
\newblock {Athena Scientific}, Belmont, Mass., 1997.

\bibitem{Sontag.1989}
E.~D. Sontag.
\newblock {Smooth stabilization implies coprime factorization}.
\newblock {\em {IEEE Trans. Autom. Control}}, 34(4):435--443, 1989.

\bibitem{Andersson2019}
J.~A.~E. Andersson, J.~Gillis, G.~Horn, J.~B. Rawlings, and M.~Diehl.
\newblock {CasADi} -- {A} software framework for nonlinear optimization and optimal control.
\newblock {\em Math. Program. Comput.}, 11(1):1--36, 2019.

\bibitem{Waechter2005}
A.~W\"achter and L.~T. Biegler.
\newblock On the implementation of an interior-point filter line-search algorithm for large-scale nonlinear programming.
\newblock {\em Math. Program.}, 106(1):25--57, 2005.

\bibitem{Hu2018}
H.~Hu, X.~Feng, R.~Quirynen, M.~E. Villanueva, and B.~Houska.
\newblock {Real-Time Tube MPC Applied to a 10-State Quadrotor Model}.
\newblock In {\em {Proc. Am. Control Conf. (ACC)}}, pages 3135--3140, Milwaukee, USA, 2018.

\bibitem{Koehler2023}
M.~K\"ohler, M.~A. M\"uller, and F.~Allg\"ower.
\newblock {Distributed Model Predictive Control for Periodic Cooperation of Multi-Agent Systems}.
\newblock {\em IFAC-PapersOnLine}, 56(2):3158--3163, 2023.

\end{thebibliography}
